\documentclass[letter,11pt]{article}
\usepackage{amssymb}
\usepackage{amsmath}
\usepackage{amsfonts}
\usepackage{amsmath}
\usepackage{txfonts}
\usepackage{type1cm}

\newcounter{tmptheorem}

\newtheorem{theorem}{Theorem}
\newtheorem{definition}[theorem]{Definition}
\newtheorem{lemma}[theorem]{Lemma}
\newtheorem{corollary}[theorem]{Corollary}

\newenvironment{remark}
{\textit{Remark}
}
{}

\newtheorem{claim}[theorem]{Claim}
\newcommand{\qed}{$\square$}

\newcommand{\minor}{\mathrm{minor}}
\newcommand{\status}{\mathbf{status}}
\newcommand{\eligible}{\mathsf{eligible}}
\newcommand{\consistency}{\mathbf{consistency}}
\newcommand{\ineligible}{\mathsf{ineligible}}
\newcommand{\consistent}{\mathsf{consistent}}
\newcommand{\inconsistent}{\mathsf{inconsistent}}

\newcommand{\LE}{\mathsf{LE}}
\newcommand{\error}{\mathsf{error}}
\newcommand{\depth}{\mathsf{depth}}

\newcommand{\Label}{\mathsf{label}}

\newcommand{\ekey}{\mathsf{ekey}}
\newcommand{\primary}{\mathbf{primary}}
\newcommand{\primarykey}{\mathbf{primarykey}}
\newcommand{\fsize}{\mathbf{size}}
\newcommand{\CONTINUE}{\mathtt{CONTINUE}}
\newcommand{\CON}{\mathtt{CON}}
\newcommand{\ENQUEUE}{\mathtt{ENQUEUE}}
\newcommand{\DEQUEUE}{\mathtt{DEQUEUE}}
\newcommand{\BREAK}{\mathtt{BREAK}}
\newcommand{\Adj}{\mathsf{Adj}}
\newcommand{\id}{\mathsf{id}}
\newcommand{\texttilde}[1]{$\widetilde{\text{#1}}$}

\newenvironment{proof}{%
  \noindent{\it Proof\ }}{%
  \hspace*{\fill}\qed
  \vspace{2ex}\\}

\newenvironment{proofsketch}{%
  \noindent{\it Proof (Sketch).\ }}{%
  \hspace*{\fill}\qed
  \vspace{2ex}}

\newcommand{\calC}{\mathcal{C}}

\newcommand{\calP}{\mathcal{P}}

\newcommand{\calS}{\mathcal{S}}

\newcommand{\bfR}{\mathbf{R}}
\newcommand{\bfS}{\mathbf{S}}

\newcommand{\bfX}{\mathbf{X}}

\newcommand{\ket}[1]{| #1 \rangle}
\newcommand{\abs}[1]{\vert #1 \vert}
\newcommand{\ceil}[1]{\lceil #1 \rceil}

\begin{document}

\pagestyle{plain}
\sloppy

%\mainmatter

\title{
  Exact Quantum Algorithms\\
  for the Leader Election Problem\footnote{A preliminary version of
  this paper appeared in~\cite{TanKobMatSTACS05}.}
}

\author{
{\large \hspace*{-1ex} $\mbox{\bf Seiichiro Tani}^{\ast\dagger}$
 \hspace*{-1ex}}\\
{\tt tani@theory.brl.ntt.co.jp}
\and
{\large \hspace*{-1ex} $\mbox{\bf Hirotada Kobayashi}^{\ddagger}$
 \hspace*{-1ex}}\\
{\tt hirotada@nii.ac.jp}
\and
{\large \hspace*{-1ex} $\mbox{\bf Keiji Matsumoto}^{\ddagger \dagger}$
 \hspace*{-1ex}}\\
{\tt keiji@nii.ac.jp}
}

\date{}

\maketitle
\thispagestyle{plain}
\pagestyle{plain}
\vspace*{-5mm}
\begin{center}
  ${}^{\ast}$NTT Communication Science Laboratories, NTT Corporation\\
  ${}^{\dagger}$Quantum Computation and Information Project, ERATO-SORST, JST\\
  ${}^{\ddagger}$Principles of Informatics Research Division,  National Institute of Informatics
\end{center}

\maketitle
\bibliographystyle{plain}
\begin{abstract}
This paper gives the first separation of quantum and classical
pure (i.e., non-cryptographic) computing abilities with no
restriction on the amount of available computing resources, by considering the exact
solvability of a celebrated unsolvable problem in classical distributed
computing, the ``leader election problem'' in anonymous networks.
The goal of the leader election problem
is to elect a unique leader from among distributed parties.
The paper considers this problem
for anonymous networks,
in which
each party has the same
identifier.
It is well-known that no classical algorithm can solve exactly
(i.e., in bounded time without error)
the leader election problem in anonymous networks, even if it is given
the number of parties.
This paper gives two quantum algorithms that,
given the number of parties,
can exactly solve the problem for any network topology
in polynomial rounds and polynomial communication/time
complexity with respect to the number of parties,
when the parties are connected by quantum communication links.
The two algorithms each have their own characteristics
with respect to complexity and the property of the networks they can work on.
Our first algorithm offers much lower time and communication complexity
than our second one, while the second one
is more general than the first one in that
the second one can run even on any network, even those whose underlying graph
is directed, whereas
the first one works well only on those with undirected graphs.
Moreover, our algorithms work well even in the case
where only the upper bound of the number of parties is given.
No classical algorithm can solve the problem even with zero error
(i.e., without error but possibly in unbounded running time)
in such cases, if the upper bound may be more than twice the number of
parties.
In order to keep the complexity of the second algorithm polynomially bounded,
a new classical technique is developed;
the technique quadratically improves the previous bound on the number of rounds required 
to compute a Boolean function on anonymous networks,
without increasing the communication complexity.
\end{abstract}

\section{Introduction}
\label{Section: Introduction}
\subsection{Background}
Quantum computation and communication are turning out
to be much more powerful than the classical equivalents
in various computational tasks.
Perhaps the most exciting developments in quantum computation
would be polynomial-time quantum algorithms
for factoring integers and computing discrete
logarithms~\cite{Sho97SIComp};
these give a 
separation of quantum and classical
computation in terms of the amount of computational
resource required to solve the problems, 
on the assumption that the problems are hard to solve 
in polynomial-time with classical algorithms.
From a practical point of view,
the algorithms also have a great impact on the real cryptosystems used in E-commerce,
since most of them assume the hardness of integer
factoring
or discrete logarithms for their security.

Many other algorithms such as Grover's
search~\cite{grover96,Bra00AMS,hoyer-mosca-dewolf03} and quantum walk~\cite{childs-cleve-deotto-farhi-gutmann-spielman03,ambainis04},
and
protocols~\cite{buhrman-cleve-wigderson98,Raz99STOC,BuhCleWatWol01PRL,BarJayKer04STOC}
have been proposed to give separations in terms of the amount of computational
resources (e.g., computational steps, communicated bits or work space)
needed to compute some functions.

From the view point of computability,  there are many results on
languages recognizable by quantum
automata~\cite{ambainis-freivalds98,amano-iwama99,yamasaki-kobayashi-tokunaga-imai02,ambainis-watrous02,yamasaki-kobayashi-imai05};
they showed that there are some languages that quantum automata can
recognize but their classical counterparts cannot.  
This gives the separation of
quantum and classical models in terms of computability, instead of
in terms of the amount of
required computational resources, when placing a sort of restriction on
computational ability (i.e., the number of internal states) of the
models.

In the cryptographic field, the most remarkable quantum result
would be the quantum key distribution protocols~\cite{BenBra84ICCSSP,Ben92PRL}
that have been proved unconditionally
secure~\cite{May01JACM,ShoPre00PRL,TamKoaImo03PRA,TamKoaImo03PRL,TamLut04PRA}.
In contrast, no unconditionally secure key distribution protocol is
possible in classical settings.
Many other studies demonstrate the superiority
of quantum computation and communication for
cryptography~\cite{DumMaySal00EUROCRYPT,Amb04JCSS,CreLegSal01EUROCRYPT,CreGotSmi02STOC,BarCreGotSmiTap02FOCS,AmbBuhDodRoh04CCC,benor-hassidim05}

This paper gives the first separation of quantum and classical
abilities for a pure (i.e., non-cryptographic) computational task with no
restriction on the amount of available computing resources;
its key advance is to consider the exact
solvability of a celebrated unsolvable problem in classical distributed
computing, the ``leader election problem'' in anonymous networks.

The leader election problem is a core problem
in traditional distributed computing
in the sense that, once it is solved, it becomes
possible to efficiently solve many substantial problems
in distributed computing
such as finding the maximum value and constructing a spanning tree
(see, e.g., \cite{Lyn96Book}).
The goal of the leader election problem
is to elect a unique leader from among distributed parties.
When each party has a unique identifier,
the problem can be
deterministically solved by
selecting the party
that has the largest identifier as the leader;
many classical deterministic algorithms 
in this setting
have been proposed~\cite{DolKlaRod82JAlgo,Pet82ToPLaS,GalHumSpi83ToPLaS,FreLyn87JACM,LeeTan87DistComp}.
As the number of parties grows, however,
it becomes difficult to preserve the uniqueness of the identifiers.
Thus, other studies have examined the cases wherein each party is anonymous,
i.e.,
each party has the same
identifier~\cite{Ang80STOC,ItaRod90InfoComp,YamKam96IEEETPDS-1,YamKam96IEEETPDS-2},
as an extreme case.
In this setting, every party has to  be in a common initial state and
run a common algorithm; if there are two parties who are in different
initial states or who run different algorithms, they can be
distinguished by regarding their initial states or algorithms as their
identifiers.  
A simple algorithm meets this condition:
initially, all parties are eligible to be the unique leader, and
repeats common subroutines that drop eligible parties
until only one party is eligible.
In the subroutines,
(1) every eligible party independently generates a random bit, 
(2) all parties then collaborate to check if all eligible parties have the
same bit,
and (3) if not, the eligible parties having bit ``0''
are made ineligible (otherwise nothing changes).
Thus, the problems can be solved probabilistically.
Obviously, there is a positive probability that
all parties get identical values from independent random number
generators.
In fact, the problem cannot be solved exactly
(i.e.,  in bounded time and with zero error)
on networks having symmetric structures such as rings,
even if every party can have unbounded computational power or
perform analogue computation with infinite precision.
The situation is unchanged if every party is allowed to
share infinitely many random strings.
Strictly speaking,
no classical exact algorithm
(i.e., an algorithm that runs in bounded time and solves the problem with zero error)
exists
for a broad class of network topologies including regular graphs,
even if the network topology (and thus the number of parties)
is known to each party prior to algorithm invocation~\cite{YamKam96IEEETPDS-1}.
Moreover, 
no classical zero-error algorithm
exists in such cases for any topology
that has a cycle as its subgraph~\cite{ItaRod90InfoComp},
if every party can get only the upper bound of the number of the parties.

\subsection{Our Results}
This paper considers the distributed computing model
in which the network is anonymous and consists of quantum
communication links,\footnote{Independently,
quantum leader election in the broadcast model
was studied in~\cite{dhondt-panangaden06QIC}.}
and gives two \emph{exact} quantum algorithms
both of which, given the number of parties, elect a unique leader from among $n$ parties
in polynomial time
for \emph{any} topology of synchronous networks
(note that no party knows the topology of the network).
Throughout this paper,
by time complexity we mean
the maximum number of steps,
including steps for the local computation,
necessary for each party to execute the protocol,
where the maximum is taken over all parties.
In synchronous networks, 
the number of simultaneous message passings
is also an important measure.
Each turn of simultaneous message passing is referred to
as a \emph{round}.

We first summarize our results before giving explanations of our
algorithms.
Outlines of our algorithms are given in subsection~\ref{subsec:outline-algorithms}.

Our first algorithm, ``Algorithm~I,'' runs in $O(n^3)$ time.
The total communication complexity of this algorithm is $O(n^4)$,
which includes the quantum communication of $O(n^4)$ qubits.
More precisely, we prove the next theorem.
\begin{theorem}
Let $\abs{E}$ and $D$ be the number of edges and the maximum degree of the
underlying graph, respectively.
Given $n$,
the number of parties,
Algorithm~I exactly elects a unique leader
in $O(n^2)$ rounds and $O(D n^2)$ time.
Each party connected with $d$ parties
requires $O(d n^2)$-qubit communication, and 
the total communication complexity over all parties
is $O(\abs{E} n^2)$.
\label{Section1: Theorem: Complexity of Algorithm I for n}
\end{theorem}

The first algorithm works in a similar way to the simple probabilistic
algorithm described the above, except that
the first algorithm can reduce the number of parties
even in the situations corresponding to the classical cases
where all eligible parties obtain the same values.

Our second algorithm, ``Algorithm II,'' is more general than the first one in that
the second one can work even on any network, even those whose underlying graph
is directed and strongly-connected\footnote{Note that it is natural to assume that the underlying graph is 
strongly-connected when it is a directed graph; otherwise
there is a party that can know the information of only a part of the network.}
while the first one cannot.
Roughly speaking, the first algorithm has to invert parts of the quantum computation
and communication already performed 
to erase garbage for subsequent computation;
this inverting operation is hard to perform on directed networks (except
for some special cases), i.e., it demands communication links be
bidirectional in general. In contrast, the second algorithm does not have
to perform such inverting operations.
Another desirable property is that the second algorithm needs
less \emph{quantum} communication than the first one;
since sending a qubit would cost more than sending a
classical bit, reducing the quantum communication complexity 
is
desirable.
Our second algorithm
incurs $O(n^6 (\log n)^2)$ time complexity,
but demands the quantum communication of only $O(n^2 \log n)$ qubits
(plus classical communication of $O(n^6 (\log n)^2)$ bits).
The second algorithm is also superior to the first one in terms of round complexity.
While the first algorithm needs $O(n^2)$ rounds
of quantum communication,
the second algorithm needs only one round of quantum communication
at the beginning of the protocol
to share a sufficient amount of entanglement,
and after the first round,
the protocol performs only local quantum operations
and classical communications (LOCCs) of $O(n \log n)$ rounds.
More precisely, we prove the next theorem.
\begin{theorem}
Let $\abs{E}$ and $D$ be the number of edges and the maximum degree of the
underlying graph, respectively.
Given the number $n$ of parties,
Algorithm~II exactly elects a unique leader
in $O(Dn^5(\log n)^2)$ time and $O(n\log n)$ rounds
of which only the first round requires quantum communication.
The total communication complexity over all parties
is $O(D\abs{E} n^3(\log D)\log n)$
which includes the communication of only $O(\abs{E} \log n)$ qubits.
\label{Section1: Theorem: Complexity of Algorithm II for n}
\end{theorem}

Algorithms~I~and~II are easily modified to allow their use in asynchronous networks.
We summarize the complexity of the two algorithms in Table~\ref{tab:complexity}.
\begin{table}[t]
  \centering
  \caption{Complexity of the two algorithms for $n$, the number  of parties}
\begin{tabular}{|c||c|c|c|c|}\hline
& Time & Quantum Communication & Total Communication &Round \\\hline
Algorithm I & $O(n^4)$ & $O(n^4)$ & $O(n^4)$ & $O(n^2)$\\\hline
Algorithm II & $O(n^6(\log n)^2)$ & $O(n^2)$ & $O(n^6 (\log n)^2)$ &
$O(n\log n)$\\\hline
\end{tabular}
  \label{tab:complexity}
\end{table}

Furthermore, both algorithms can be modified
so that they work well even when
each party initially knows only the upper bound $N$ of the number $n$ of parties.
For Algorithm~I, each party has only to perform the algorithm
with 
$N$ instead of $n$.
The complexity is described
simply by replacing every 
$n$ by $N$.
\begin{corollary}
Let $\abs{E}$ and $D$ be the number of edges and the maximum degree of the
underlying graph, respectively.
Given $N$,
the number  of parties,
Algorithm~I exactly elects a unique leader
in $O(N^2)$ rounds and $O(D N^2)$ time.
Each party connected with $d$ parties
requires $O(d N^2)$-qubit communication, and 
the total communication complexity over all parties
is $O(\abs{E} N^2)$.
\label{Section1: Corollary: Complexity of Algorithm I for N}
\end{corollary}
Algorithm~II 
strongly depends on counting the exact number of eligible parties and this
requires knowledge of the exact number of parties.
Thus, we cannot apply Algorithm~II as it is, when each party initially
knows only the upper bound $N$ of $n$;
We need considerable elaboration to
modify Algorithm~II.
We call this variant of Algorithm~II Generalized Algorithm~II.

\begin{corollary}
Let $\abs{E}$ and $D$ be the number of edges and the maximum degree of the
underlying graph, respectively.
Given $N$, the upper bound of the number of parties,
Generalized Algorithm~II exactly elects a unique leader
in $O(DN^6(\log N)^2)$ time and $O(N\log N)$ rounds
of which only the first round requires quantum communication.
The total communication complexity over all parties
is $O(D\abs{E} N^4(\log D)\log N)$
which includes the communication of only $O(\abs{E} N\log N)$ qubits.
\label{Section1: Corollary: Complexity of Algorithm II for N}
\end{corollary}

These corollaries imply that the exact number of parties can be computed
when its upper bound is given.
No classical zero-error algorithm
exists in such cases for any topology
that has a cycle as its subgraph~\cite{ItaRod90InfoComp}.

In general, most quantum algorithms use
well-known techniques such as the quantum amplitude amplification~\cite{Bra00AMS}
and the quantum Fourier transform~\cite{Sho97SIComp}, as building blocks.
In contrast, our algorithms use quite new quantum operations in
combination
with (improvements of) classical techniques.
The quantum operations are be interesting in their own right
and may have the potential to trigger the development of 
other tools that can 
elucidate
the advantage of quantum computing over the classical computing.
The folded view we shall introduce later in this paper in Algorithm~II
can be used to compute any (computable) Boolean function 
with distributed inputs
on anonymous
networks, as used in~\cite{kranakis-krizanc-vandenberg94},
when every party knows the number of parties but not
the topology of the network;
this quadratically decreases the number of rounds needed
without increasing the communication complexity.
More precisely, 
an $n$-bit Boolean function can be computed
by using folded view with $O(n)$ rounds of classical communication of
$O(n^6\log n)$ bits,
while the algorithm
in~\cite{kranakis-krizanc-vandenberg94}
computes the function in $O(n^2)$ rounds with the same amount of
classical communication (their second algorithm
can compute a symmetric Boolean function with 
lower communication complexity, i.e., $O(n^5(\log n)^2)$,
and $O(n^3\log n)$ rounds, but it requires 
that every party knows the topology of the network).
From a technical viewpoint,
folded view
is a generalization of Ordered Binary Decision Diagrams
(OBDD)~\cite{Bry86IEEETC},
which are used in commercial VLSI CAD systems
as data structures to process Boolean functions;
folded view would be interesting in its own right.

From a practical point of view, classical probabilistic algorithms
would do in many situations, since they are expected to run with the
sufficiently small time/communication complexity.  Furthermore, our
model  does not allow any noise on the communication links and our algorithms
use unitary operators whose matrices have elements depending on
$e^{i\cdot O(\frac{1}{n})}$ for the problem of $n$ parties.  In practical
environments, in which communication noise is inevitable and all
physical devices have some limits to their precision, our algorithms cannot
avoid errors.  In this sense, our algorithm would be fully theoretical
and in a position to reveal a new aspect of quantum distributed
computing.
From a theoretical point of view,
however, no classical computation (including analog computation)
with a finite amount of
resources can exactly solve the problem when given the number of
parties, and solve even with zero error when given only the upper bound of
the number of parties;
our results demonstrate examples
of significant
superiority in computability
of distributed quantum computing
over its classical counterpart
in computability,
beyond the reduction of communication cost.
To the best knowledge of the authors,
this is the first separation of quantum and classical
pure (i.e., non-cryptographic) computing abilities with no
restriction on the amount of available computing resources.
\subsection{Outline of the Algorithms}
\label{subsec:outline-algorithms}
\subsubsection{Algorithm I}
The algorithm repeats one procedure exactly ${(n-1)}$ times,
each of which is called a \emph{phase}.
Intuitively, a phase corresponds to a coin flip.

In each phase $i$,
let ${S_i \subseteq \{1, \ldots, n\}}$
be the set of all $l$s such that party $l$ is still eligible.
First, each eligible party prepares the state
${(\ket{0}+\ket{1})/\sqrt{2}}$ in one-qubit register $\bfR_0$, instead of generating a random bit,
while each ineligible party prepares the state $\ket{0}$ in $\bfR_0$.
Every party then collaborates to check in a superposition if all
eligible parties have the same
content in $\bfR_0$. They then store the result into another
one-qubit register ${\bf S}$, followed by inversion of the computation
and communication performed for this checking in order to erase
garbage. (This inverting step makes it impossible for the algorithm to
work on directed networks.)
After measuring ${\bf S}$, exactly one of the two cases is chosen
by the laws of quantum mechanics: the first case is that 
the qubits of all eligible parties' $\bfR_0$ are in a quantum state
that superposes
the classical situations where all eligible parties \emph{do not}
have the same bit, and the second case is that
the qubits are in a state that superposes the complement situations,
i.e., cat-state ${(\ket{0}^{\otimes \abs{S_i}}+\ket{1}^{\otimes \abs{S_i}})/\sqrt{2}}$
for set $S_i$
of eligible parties. Note that the ineligible parties' qubits in 
$\bfR_0$ are not entangled.

In the first case, every eligible party measures $\bfR_0$ and gets a
classical bit. Since this corresponds to one of the classical
situations superposed in the quantum state in which all eligible parties do
not have the same bit, this always reduces the number of the eligible parties.
In the second case, however, every eligible party would get the same bit
if he measured $\bfR_0$.  To overcome this, we introduce two families
of novel unitary operations: $\{ U_k \}$ and $\{ V_k \}$.  Suppose
that the current phase is the $i$th one.  Every eligible party  then
performs $U_k$ or $V_k$, depending on whether $k=(n-i+1)$ is even or
odd.  We prove that, if $(n-i+1)$ equals $\abs{S_i}$, the resulting state superposes
the classical states in which all eligible parties \emph{do not} have the same values;
eligible parties can be reduced by using the values obtained by the
measurement. 
If $(n-i+1)\neq \abs{S_i}$, the resulting quantum state may include a classical
state in which all parties have the same values;
the set of eligible parties may not be changed by using the measurement results.
(To make this step work, we erase
the garbage after computing the data stored into ${\bf S}$.)
It is clear from the above that $k$ is always at least $|S_i|$ in each
phase $i$, since $k$ is $n=|S_1|$ in the first phase and is decreased
by 1 after each phase. It follows that exactly one leader is elected
after the last phase.

\subsubsection{Algorithm II}
The algorithm consists of two stages, Stages 1 and 2.
Stage 1 aims to have the $n$ parties share a certain type of entanglement,
and thus, this stage requires the parties to exchange quantum messages.
In Stage 1, every party exchanges only one message
to share $s$ pure quantum states
${\ket{\phi^{(1)}}, \ldots, \ket{\phi^{(s)}}}$ of $n$ qubits.
Here, each $\ket{\phi^{(i)}}$ is of the form
${(\ket{x^{(i)}} + \ket{\overline{x}^{(i)}})/\sqrt{2}}$
for an $n$-bit string $x^{(i)}$ and its bitwise negation $\overline{x}^{(i)}$,
and the $l$th qubit of each $\ket{\phi^{(i)}}$
is possessed by the $l$th party.

Stage 2 selects a unique leader from among the $n$ parties only by LOCCs,
with the help of the shared entanglement prepared in Stage 1.  This
stage consists of at most $s$ phases, each of which reduces the number
of eligible parties by at least half,
and maintains variable $k$ that represents the number of eligible parties.
At the beginning of Stage 2, $k$ is set to $n$ (i.e., the number of
all parties).
Let ${S_i \subseteq \{1,
  \ldots, n\}}$ be the set of all $l$s such that party $l$ is still
eligible just before entering phase $i$.  First every party exchanges
classical messages to decide if  all eligible parties have the same
content for state $\ket{\phi^{(i)}}$ or not.

If so,
the parties
transform $\ket{\phi^{(i)}}$
into the $\abs{S_i}$-cat state
${(\ket{0}^{\otimes \abs{S_i}} + \ket{1}^{\otimes \abs{S_i}})/\sqrt{2}}$
shared only by eligible parties
and then use $U_k$ and $V_k$ as in Algorithm~I
to obtain a state that superposes
the classical states in which all eligible parties \emph{do not} have the
same values.

Each party $l$ then measures his qubits to obtain a label 
and find the minority among all parties' labels.
The number of eligible parties is then 
reduced by at least half via minority voting with respect to the labels.
The updated number of eligible parties is set to $k$.

To count the exact number of eligible parties that have a certain label,
we make use of
a classical technique, called
\emph{view}~\cite{YamKam96IEEETPDS-1,yamashita-kameda99}.
However, a na\"{\i}ve application of view
incurs exponential classical time/communication complexity,
since view is essentially a tree that covers the underlying graph, $G$, of
the network, i.e., the universal cover of $G$~\cite{norris95}.
To keep the complexity moderate,
we introduce the new technique of \emph{folded view},
with which the algorithm still runs
in time/communication polynomial with respect to the number of parties.

To generalize Algorithm~II so that it can work given the upper bound of
the number of the parties,
we modify Algorithm~II
so that 
it can halt in at most $\lceil \log N\rceil$ phases
even when it is given the wrong number of
parties as input,
and it can output
$\mbox{``$\error$''}$ 
when it concludes that the leader has not been elected yet
after the last phase.
The basic idea is to simultaneously run $N-1$ processes of 
the modified algorithm,
each of which is given $2,3,\ldots, N$, respectively,  as the number of parties.
Let $M$ be the largest $m\in \{ 2,3,\ldots ,N\}$ such that
the process of the modified algorithm for $m$
terminates with output $\mbox{``$\eligible$''}$ or
$\mbox{``$\ineligible$''}$ 
(i.e., without outputting $\mbox{``$\error$''}$).
We will prove that
$M$ is equal to the hidden number of parties, i.e., $n$,  and thus
the process for $M$ elects the unique leader.
We call this Generalized Algorithm~II.

\subsection{Organization of the Paper}
Section~\ref{sec:Preliminaries} defines the network
model and the leader election problem in an anonymous network.
Section~\ref{sec:Quantum Leader Election Algorithm I} first gives Algorithm~I
when the number of parties is given to every party, and then
generalizes it to the case where only the upper bound of the number of
parties is given.
Section~\ref{sec:Quantum Leader Election Algorithm II}
describes Algorithm~II 
when the number of parties is given to every party, 
and also its generalization for the case where 
only the upper bound of the number of
parties is given. 
Algorithm~II is first described on undirected networks
just for ease of understanding,
and then modified so that it can work well when the underlying
graph is directed and strongly-connected.
Section~\ref{sec: view compression} defines folded view and 
proves that every party can construct it
and extract from it the number of parties that have some specified value
in polynomial time and communication cost.
Finally, section~\ref{sec: conclusion} concludes the paper.

\section{Preliminaries}
\label{sec:Preliminaries}

\subsection{Quantum Computation}
Here we briefly introduce quantum computation
(for more detailed introduction, see~\cite{nielsen-chuang00,KitSheVya02Book}).  
A unit of quantum information corresponding to a bit is called a
\emph{qubit}.
A \emph{pure quantum state} (or simply a \emph{pure state})
of the quantum system consisting of $n$ qubits is a vector of
unit-length
in the $2^n$-dimensional Hilbert space.
For any basis $\{ \ket{B_0},\ldots, \ket{B_{2^n-1}}\}$ of the space,
any pure state can be represented by 
$$
\sum _{i=0}^{2^n-1} \alpha_i|B_i \rangle,
$$
where complex number $\alpha_i$, called \emph{amplitude}, holds $\sum _i
|\alpha_i|^2=1$.
There is a simple basis in which every basis state $\ket{B_i}$
corresponds to one of the $2^n$ possible classical positions
in the space: $(i_0,i_1,\ldots ,i_{n-1})\in \{0,1\}^n$.
We often denote $\ket{B_i}$ by 
$$\ket{i_0}\otimes \ket{i_1}\otimes \cdots \otimes \ket{i_{n-1}},$$
$$\ket{i_0}\ket{i_1}\cdots \ket{i_{n-1}},$$ or
$$\ket{i_0i_1\cdots i_{n-1}}.$$
This basis is called the \emph{computational basis}.

If the quantum system is measured
with respect to any basis $\{ B_0,\ldots, B_{2^n-1}\}$, the
probability of observing basis state $|B_i\rangle$ is $|\alpha_i|^2$.  As
a result of the measurement, the state is projected onto the observed basis state. 
Measurement can also be performed on a part of the system or some of
the qubits forming the system.
For example, let $\sum _{i=0}^{2^n-1} \alpha_i|B_i \rangle$ be $|\phi _0\rangle|0\rangle+
|\phi _1\rangle|1\rangle$.
If we measure the last qubit 
with respect to basis $(|0\rangle, |1\rangle)$,
we obtain $|0\rangle$ with probability $\langle \phi _0|\phi
_0\rangle$ and $|1\rangle$ with probability $\langle \phi _1|\phi
_1\rangle$, where $\langle \phi|\psi\rangle$  is the inner product of
vectors $|\phi\rangle$ and $|\psi\rangle$,
and the first ${n-1}$ qubits collapse to
$|\phi _0\rangle/\sqrt{\langle \phi _0|\phi
_0\rangle}$ and $|\phi _1\rangle/\sqrt{\langle \phi _1|\phi
_1\rangle}$, respectively.
In the same way, we may measure the system with respect to other bases.
In particular,
we will use the \emph{Hadamard basis}
${\{|+\rangle, |-\rangle\}}$
of the 2-dimensional Hilbert space, where
$\ket{+}=\frac{1}{\sqrt{2}}(|0\rangle+|1\rangle)$ and $\ket{-}=\frac{1}{\sqrt{2}}(|0\rangle-|1\rangle)$.
If $\sum _{i=0}^{2^n-1} \alpha_i|B_i \rangle$ is expressed as $|\phi _+\rangle|+\rangle+
|\phi _-\rangle|-\rangle$,
by measuring the last qubit
with respect to the Hadamard basis,
we observe $|+\rangle$ with probability $\langle \phi _+|\phi
_+\rangle$ and $|-\rangle$ with probability $\langle \phi _-|\phi
_-\rangle$,
and the post-measurement states are
$|\phi _+\rangle/\sqrt{\langle \phi _+|\phi
_+\rangle}$ and $|\phi _-\rangle/\sqrt{\langle \phi _-|\phi
_-\rangle}$, respectively.

In order to perform computation over the quantum system, 
we want to apply transformations to the state of the system.
The laws of quantum mechanics permit
only unitary transformations over the Hilbert space.
These transformations are represented by unitary matrices, where
a unitary matrix is one whose conjugate transpose equals its inverse.

\subsection{The distributed network model}
\label{subsec:model}
A \emph{distributed system} (or \emph{network}) is composed of
multiple parties
and bidirectional classical communication links connecting parties.
In a quantum distributed system,
every party can perform quantum computation and communication,
and each adjacent pair of parties
has a bidirectional quantum communication link between them.
When the
parties and links are regarded as nodes and edges, respectively,
the topology of the distributed system is expressed by an undirected
connected graph, denoted by ${G=(V,E)}$. 
In what follows,
we may identify each party/link with its corresponding node/edge
in the underlying graph for the system,
if it is not confusing.
Every party has \emph{ports} 
corresponding one-to-one to communication links incident to the party.
Every port of party $l$ has a unique label $i$, $1 \leq i \leq d_l$,
where $d_l$ is the number of parties adjacent to $l$.
More formally,
$G$ has a \emph{port numbering},
which is a set $\sigma$ of functions $\{ \sigma[v] \mid v\in V\}$
such that, for each node $v$ of degree $d_v$,
$\sigma[v]$ is a bijection from the set of edges incident to $v$
to $\{ 1, 2, \ldots, d_v \}$.
It is stressed that
each function $\sigma[v]$ may be defined independently of the others.
Just for ease of explanation, we assume that
port $i$ corresponds to the link
connected to the $i$th adjacent party of party $l$.
In our model,
each party knows the number of his ports and
can appropriately choose one of his ports
whenever he transmits or receives a message.

Initially, every party has local information,
such as his local state,
and global information,
such as the number of nodes in the system (or an upper bound).
Every party runs the same algorithm,
which has local and global information as its arguments.
If all parties have the same local and global information except
for the number of ports they have,
the system is said to be \emph{anonymous}.
This is essentially equivalent
to the situation in which every party has the same identifier
since we can regard the local/global information of the party as his identifier.

Traditionally, distributed system are either
synchronous or asynchronous. In the synchronous case, message passing is performed synchronously.
The unit interval of synchronization is called a \textit{round}, which is the
combination of the following two steps \cite{Lyn96Book}, where the two functions
that generate messages and change local states are defined by the algorithm
invoked by each party:
(1) each party changes the local state as a function of 
  the current local state and the incoming messages, and removes the
  messages from the ports;
(2) each party generates messages and decides ports
through which the messages should be sent
as another function of the current
  local state, and sends the messages via the ports.
If message passing is performed synchronously,
a distributed system is called \emph{synchronous}.

\subsection{Leader election problem in anonymous networks}
The leader election problem is formally defined as follows.
\begin{definition}[Leader Election Problem ($\LE_n$)]
Suppose that there is an $n$-party distributed system with
underlying graph $G$, and that each party
$i\in \{ 1,2,\ldots, n\}$ in the
system has a variable $x_i$ initialized to some constant $c_i$. 
Create the situation in which ${x_k=1}$ for only one of
$k\in \{ 1,2,\ldots, n\}$
and ${x_i=0}$ for every $i$ in the rest 
$\{ 1,2,\ldots, n\}\setminus \{k\}$.
\end{definition}
When each party $i$ has his own unique identifier, i.e.,  $c_i\in \{
1,2,\ldots, n\}$ such that $c_i\neq c_j$ for $i\neq j$, $\LE_n$ can be
deterministically solved in  $\Theta(n)$ rounds in the synchronous case
and $\Theta(n\log n)$ rounds in the asynchronous case~\cite{Lyn96Book}.  

When $c_i=c_j$ for all $i$ and $j$ ($i\neq j$), 
the parties are said to be \emph{anonymous} and the
distributed system (network) consisting of anonymous parties is also said
to be \emph{anonymous}.  The leader election
problem in an anonymous network was first investigated by
Angluin~\cite{Ang80STOC}. 
Her model allows 
any two adjacent parties can collaborate to toss a coin to decide a winner:
one party receives a bit ``1'' (winner) 
if and only if the other party receives a bit ``0'' (loser),
as the result of the coin toss
Among the two parties, thus,  a leader can exactly be
elected, where ``exactly'' means ``in bounded time and without error.''
Even in her model, she showed that there are infinitely many
topologies of
anonymous networks for which
no algorithms exist that can exactly solve the problem,
and gave a necessary and
sufficient condition in terms of graph covering for exactly solving the
leader election problem.  In the usual anonymous network model
which does not assume such coin-tossing
between adjacent parties, Yamashita and
Kameda~\cite{YamKam96IEEETPDS-1} proved that,  if the ``symmetricity''
(defined in \cite{YamKam96IEEETPDS-1}) of the network topology is more
than one, $\LE_n$ cannot be solved exactly (more rigorously speaking,
there are some port numberings for which $\LE_n$ cannot be solved
exactly) by any classical algorithm even if all parties know the
topology of the network (and thus the number of parties).  The
condition that symmetricity of
more than one implies that,  a certain port numbering
satisfies that, for any party $i$,  
there is 
an automorphism 
on the underlying graph with the port numbering
that
exchanges $i$ and another party $i'$.
This condition holds for a broad class of graphs, including
regular graphs. Formally, the symmetricity is defined later by using ``views.''

Since it is impossible in many cases to exactly solve the problem in the classical setting, many
probabilistic algorithms have benn proposed. Itai and Rodeh~\cite{ItaRod81FOCS,ItaRod90InfoComp} gave a
zero-error algorithm for a synchronous/asynchronous unidirectional
ring of size $n$; it
is
expected to take $O(n)$/$O(n\log n)$ rounds with the communication of
$O(n)$/$O(n\log n)$ bits.

When every party knows only the upper bound
of the number of the parties, which is at most twice the exact number,
Itai and Rodeh~\cite{ItaRod90InfoComp} showed that the problem can still be 
solved with zero error on a ring.
However, if 
the given upper bound
can be more than twice the exact number of the parties,
they showed that there is no
zero-error classical algorithm for a ring.
This impossibility result 
can be extended to a general topology having cycles.
They also proved that there is a 
bounded-error algorithm for a ring,
given an upper bound of the number of parties.

Schieber and
Snir~\cite{schieber-snir94} gave bounded-error algorithms
for any topology even when no information on the number of
parties is available, although no party can detect the termination
of the algorithms 
(such algorithms are called \emph{message termination algorithms}).
Subsequently, Afek and Matias~\cite{afek-matias94} described more
efficient bounded-error message termination algorithms.

Yamashita and Kameda ~\cite{yamashita-kameda99} examined the case in
which every party is allowed to send messages only via the broadcast channel and/or
to receive messages only via their own mailboxes (i.e.,
no party can know which port was used to
send and/or receive messages).

\section{Quantum leader election algorithm I}
\label{sec:Quantum Leader Election Algorithm I}

For simplicity, we assume that the network is synchronous
and each party knows the number of parties, $n$,  prior to algorithm invocation.
It is easy to generalize our algorithm to the asynchronous case
and to the case where only the upper bound $N$ of the number of parties
is given,
as will be discussed at the end of this section.

Initially, all parties are eligible to become the leader.
The key to solving the leader election problem in an anonymous network
is to break symmetry,
i.e., 
to have some parties possess a certain state different from 
those of the other parties.

First we introduce the concept of
\emph{consistent} and \emph{inconsistent} strings.
Suppose that each party $l$ has $c$-bit string $x_l$
(i.e., the $n$ parties share $cn$-bit string ${x = x_1x_2 \cdots x_n}$).
For convenience, we may consider that each $x_l$ expresses an integer,
and identify string $x_l$ with the integer it expresses.
Given a set ${S \subseteq \{1, \ldots, n\}}$,
string $x$ is said to be \emph{consistent} over $S$
if $x_l$ 
has the same value for all $l$ in $S$.
Otherwise, $x$ is said to be \emph{inconsistent} over $S$.
We also say that the $cn$-qubit pure state
${\ket{\psi} = \sum_{x} \alpha_x \ket{x}}$
shared by the $n$ parties is \emph{consistent (inconsistent)} over $S$
if ${\alpha_x \neq 0}$ only for $x$'s that are consistent
(inconsistent) over $S$. 
Note that there are states that are neither consistent nor
inconsistent
(i.e., states which are in a superposition of 
consistent and inconsistent states).
By \emph{$m$-cat state}, we mean 
the pure state of the
form of $(\ket{0}^{\otimes m}+\ket{1}^{\otimes m})/\sqrt{2}$,
for positive integer $m$.
When we 
apply the operator of the form
\[
\sum_j\alpha_j\ket{\eta_j}\ket{0}\mapsto\sum_j\alpha_j\ket{\eta_j}\ket{\eta_j}
\]
for the computational basis $\{ \ket{\eta_j}\}$ over the Hilbert space
of known dimensions, we just say ``copy'' if
it is not confusing. 

\subsection{The algorithm}
\label{Subsection: The Algorithm I}

The algorithm repeats one procedure exactly ${(n-1)}$ times,
each of which is called a \emph{phase}.
In each phase, the number of eligible parties either decreases
or remains the same, but never increases or becomes zero.
After ${(n-1)}$ phases, the number of eligible parties becomes
one with certainty.

Each phase has a parameter denoted by $k$,
whose value is ${(n-i+1)}$ in phase $i$.
In each phase $i$,
let ${S_i \subseteq \{1, \ldots, n\}}$
be the set of all $l$s such that party $l$ is still eligible.
First, each eligible party prepares the state
${(\ket{0}+\ket{1})/\sqrt{2}}$ in register $\bfR_0$,
while each ineligible party prepares the state $\ket{0}$ in $\bfR_0$.
Next every party calls Subroutine~A, followed by partial measurement.
This transforms the system state, i.e., the state in all parties' $\bfR_0$s into either
$(\ket{0}^{\otimes \abs{S_i}}+\ket{1}^{\otimes \abs{S_i}})\otimes
\ket{0}^{\otimes (n-|S_i|)}/\sqrt{2}$
or a state that is inconsistent over $S_i$,
where the first $\abs{S_i}$ qubits  represent the qubits in the eligible
parties' $\bfR_0$s.
In the former case, each eligible party calls Subroutine~B,
which uses a new ancilla qubit in register $\bfR_1$.
If $k$ equals $|S_i|$,
Subroutine~B always succeeds in transforming the $|S_i|$-cat state in
the eligible parties' $\bfR_0$s
into a 2$|S_i|$-qubit state that is inconsistent over $S_i$ by using
the $|S_i|$ ancilla qubits.
In the latter case, each eligible party simply initializes the qubit
in $\bfR_1$ to state $\ket{0}$.
Now, each eligible party $l$ measures his qubits in $\bfR_0$ and $\bfR_1$ in the computational basis
to obtain (a binary expression of) some two-bit integer $z_l$.
Parties then compute the maximum value of $z_l$ over all eligible parties $l$,
by calling Subroutine~C.
Finally,  parties with the maximum value remain eligible,
while the other parties become ineligible.
More precisely, 
each party $l$ 
having $d_l$ adjacent parties
performs Algorithm~I
which is described in Figure~\ref{Figure: Quantum leader election algorithm I}
with parameters ``$\eligible$,'' $n$, and $d_l$.
The party who obtains the output ``$\eligible$'' is the unique leader.
Precise descriptions of Subroutines A, B, and C
are to be found in subsections~\ref{Subsection: Subroutine A}, 
\ref{Subsection: Subroutine B},
and~\ref{Subsection: Subroutine C}, respectively.

\begin{figure}[t]
\begin{center}
\hrulefill\\
\textbf{Algorithm~I}
\vspace{-2mm}
\begin{description}
\setlength\itemsep{-3pt}
  \item[Input:]
    a classical variable $\status\in \{ \mbox{``$\eligible$''}, \mbox{``$\ineligible$''}\}$, integers ${n,d}$
  \item[Output:]
    a classical variable $\status\in \{ \mbox{``$\eligible$''}, \mbox{``$\ineligible$''}\}$
\end{description}
\vspace{-6mm}
\begin{enumerate}
\setlength\itemsep{-3pt}
  \item
    Prepare one-qubit quantum registers $\bfR _0 $, $\bfR _1$, and $\bfS$.
  \item
    For ${k:=n}$ down to $2$, do the following:
\vspace{-6pt}
\begin{enumerate}
\setlength\itemsep{0pt}
      \item
        If ${\status = \mbox{``$\eligible$,''}}$
        prepare the states
        ${(\ket{0}+\ket{1})/\sqrt{2}}$ and $\ket{\mbox{``$\consistent$''}}$
        in $\bfR _0 $ and $\bfS$,
        otherwise prepare the states $\ket{0}$ and
        $\ket{\mbox{``$\consistent$''}}$ in $\bfR _0 $ and $\bfS$.
      \item
        Perform Subroutine~A with $\bfR _0 $, $\bfS$, $\status$, $n$, and $d$.
      \item
        Measure the qubit in $\bfS$
        in the 
        $\{ \ket{\mbox{``$\consistent$''}},
        \ket{\mbox{``$\inconsistent$''}} \}$ basis.\\
        If this results in $\ket{\mbox{``$\consistent$''}}$
        and ${\status = \mbox{``$\eligible$,''}}$
        prepare the state $\ket{0}$ in $\bfR _1$
        and perform Subroutine~B with $\bfR _0 $, $\bfR _1$, and
        $k$;\\
        otherwise if this results in
        $\ket{\mbox{``$\inconsistent$''}}$, just prepare the state
        $\ket{0}$ in $\bfR _1$.
      \item
        If ${\status = \mbox{``$\eligible$,''}}$
        measure the qubits in $\bfR _0 $ and $\bfR _1$
        in the $\{ \ket{0}, \ket{1} \}$ basis
        to obtain the nonnegative integer $z$ expressed by the two bits;
        otherwise let ${z := -1}$.\\
      \item
        Perform Subroutine~C with $z$, $n$, and $d$
        to know the maximum value $z_{\max}$ of $z$ over all parties.\\
        If ${z \neq z_{\max}}$, let ${\status := \mbox{``$\ineligible$.''}}$
\end{enumerate}
\vspace{-6pt}
    \item
      Output $\status$.
  \end{enumerate}
\vspace{-1\baselineskip}
  \hrulefill
\end{center}
\vspace{-1\baselineskip}
\caption{Quantum leader election algorithm I.} 
\label{Figure: Quantum leader election algorithm I}
\end{figure}

\subsection{Subroutine~A}
\label{Subsection: Subroutine A}

Subroutine~A is essentially for the purpose of 
checking the consistency over $S_i$ of each string that is superposed
in the quantum state shared by the parties.
We use commute operator ``$\circ$''
over set ${\calS = \{ 0, 1, \ast, \times \}}$
whose operations are summarized in Table~\ref{Table: operator circ}.
Intuitively,
``$0$'' and ``$1$'' represent the possible values all eligible
parties will
have
when the string finally turns out to be consistent;
``$\ast$'' represents ``don't care,''
which means that the corresponding party 
has no information on the values possessed by eligible parties;
and ``$\times$'' represents ``inconsistent,''
which means that the corresponding party already knows
that the string is inconsistent.
Although Subroutine~A is essentially 
a trivial modification of the algorithm
in~\cite{kranakis-krizanc-vandenberg94}
to handle the quantum case,
we give a precise description of Subroutine~A
in Figure~\ref{Figure: Subroutine A} for completeness.

As one can see from the description of Algorithm~I,
the content of $\bfS$ is ``$\consistent$''
whenever Subroutine~A is called.
Therefore, after every party finishes Subroutine~A,
the state shared by parties in their $\bfR_0$s
is decomposed into a consistent state
for which each party has the content ``$\consistent$'' in his $\bfS$, 
and an inconsistent state
for which each party has the content ``$\inconsistent$'' in his $\bfS$.
Steps 4 and 5 are performed to
disentangle
work quantum registers $\bfX_i^{(t)}$s from the rest.

The next two lemmas prove the correctness and complexity of
Subroutine~A.
\begin{lemma}
\label{lm:subroutineA-correctness}
Suppose that $n$ parties share $n$-qubit state
$\ket{\psi}=\sum_{i=0}^{2^n-1}\alpha _i\ket{i}$ in $n$ one-qubit registers $\bfR_0$s,
where $\alpha_i$s are
any complex numbers such that $\sum_{i=0}^{2^n-1}\abs{\alpha _i}^2=1$.
Suppose, moreover, that each party runs 
Subroutine~A with the following objects as input:
(1) ${\bfR_0}$, 
and another one-qubit quantum register ${\bfS}$, 
which is initialized to $\ket{\mbox{``$\consistent$''}}$,
(2) a classical variable $\status\in \{\eligible,\ineligible \}$,
(3) $n$ and the number $d$ of neighbors of the party.
Let $S$ be the set of indices of parties whose $\status$ is ``$\eligible$.''
Subroutine~A then outputs 
${\bfR_0}$ and ${\bfS}$ such that the qubits in all parties' $\bfR_0$s
and $\bfS$s are in the state
$\sum _{i=0}^{2^n-1}(\alpha _i\ket{i}\otimes\ket{s_i}^{\otimes n})$, where
 $s_i$ is
``$\consistent$'' if $\ket{i}$ is consistent over $S$, 
 and ``$\inconsistent$'' otherwise.
\end{lemma}
\begin{proof}
  Subroutine~A just superposes an application of a
  reversible classical algorithm to each basis state.
  Furthermore, no interference occurs since the contents of ${\bf R_0}$s are never
  changed during the execution of the subroutine.
Thus, it is sufficient to prove the correctness of Subroutine~A
when the content of $\bfR_0$ is a classical bit.
It is stressed that all ancilla qubits used as work space 
can be disentangled by inverting 
every communication and computation.

Suppose that we are given one-bit \emph{classical} registers ${\bf R_0}$ and
${\bf S}$, \emph{classical} variable $\status$, and integers $n$ and $d$.
For any party $l$ and a positive integer $t$, 
the content of $\bfX_{0,l}^{(t+1)}$ is set to $x_{0,l}^{(t)} \circ x_{1,l}^{(t)}
\circ \cdots \circ x_{d,l}^{(t)}$ in step 2.3, where
$\bfX_{i,l}^{(t)}$ is $\bfX_{i}^{(t)}$ of party $l$, and
$x_{i,l}^{(t)}$ is the content of $\bfX_{i,l}^{(t)}$.
For any $l$, by expanding this recurrence relation, 
the content of $\bfX_{0,l}^{(n)}$ can be expressed in the form of
$x_{0,l_1}^{(1)}\circ \cdots \circ x_{0,l_m}^{(1)}$ for some $m\leq (n-1)^{(n-1)}$.
Since the
diameter of the underlying graph is at most $n-1$,
there is at least one $x_{0,l'}^{(1)}$
in $x_{0,l_1}^{(1)},\ldots, x_{0,l_m}^{(1)}$ for each $l'$.
Thus $x_{0,l_1}^{(1)}\circ\cdots\circ x_{0,l_m}^{(1)}$ is equal to 
$x_{0,1}^{(1)}\circ x_{0,2}^{(1)}\circ\cdots\circ x_{0,n}^{(1)}$, since 
$\circ$ is commutative and associative, and $x\circ x=x$ for any $x\in\{ 0, 1, \ast, \times \}$.

Therefore, we can derive the following facts:
(1) if and only if there are both 0 and 1 in the contents of $\bfR_0$s
of all parties, Subroutine~A outputs $\bfS=$`$\times$',
which will be taken as ``$\inconsistent$'';
(2) if and only if there are either 0's or 1's but not both in
the contents of $\bfR_0$s (which possibly include `$\ast$'), 
Subroutine~A outputs $\bfS=$`$0$' or `$1$', respectively,
which are both taken into ``$\consistent$.''
\end{proof}
\begin{lemma}
\label{lm:subroutineA-complexity}
Let $\abs{E}$ and $D$ be the number of edges and the maximum degree of the
underlying graph, respectively.
Subroutine~A takes $O(n)$ rounds  and $O(Dn)$ time.
The total communication complexity over all parties is 
$O(|E|n)$.
\end{lemma}
\begin{proof}
 Since step 3 takes constant time and steps 4 and 5 are just the
 inversions of steps 2 and 1, respectively, it is sufficient to 
consider steps 1 and 2.
  Step 1 takes at most $O(Dn)$ time.
For each $t$, steps 2.1 and 2.1 take $O(D)$ time, and 
step 2.3 can compute ${x_0^{(t)} \circ x_1^{(t)} \circ \cdots \circ
 x_d^{(t)}}$ in $O(D)$ time by performing each $\circ$ one-by-one from 
left to right.
Hence step 2 takes $O(Dn)$ time.
It follows that Subroutine~A takes ${O(Dn)}$ time in total.

As for the number of rounds and communication complexity, it is sufficient to
consider just step 2, since only steps 2 and 4 involve communication and step 4
is the inversion of step 2. It is easy to see that the number of
rounds is $O(n)$.  As for communication complexity, every party sends two
qubits via each link for each iteration in step 2. Hence every party
needs to send $O(nD)$ qubits in step 2. By summing up the
number of qubits sent over all parties, the communication complexity
is $O(|E|n)$.
\end{proof}

\begin{table}[t]
\caption
[The definition of commute operator ``$\circ$''
used in Subroutine~A.
]
{The definition of commute operator ``$\circ$.''}
\label{Table: operator circ}
\begin{center}
\begin{tabular}{cc|c||cc|c||cc|c||cc|c}
\hline
\hspace{0.0cm}$x$\hspace{0.0cm}
& \hspace{0.0cm}$y$\hspace{0.0cm}
& \hspace{0.0cm}${x \circ y}$\hspace{0.0cm}
& \hspace{0.0cm}$x$\hspace{0.0cm}
& \hspace{0.0cm}$y$\hspace{0.0cm}
& \hspace{0.0cm}${x \circ y}$\hspace{0.0cm}
& \hspace{0.0cm}$x$\hspace{0.0cm}
& \hspace{0.0cm}$y$\hspace{0.0cm}
& \hspace{0.0cm}${x \circ y}$\hspace{0.0cm}
& \hspace{0.0cm}$x$\hspace{0.0cm}
& \hspace{0.0cm}$y$\hspace{0.0cm}
& \hspace{0.0cm}${x \circ y}$\hspace*{0.0cm}\\
\hline
     $0$ &      $0$ &      $0$ &
     $1$ &      $0$ & $\times$ &
  $\ast$ &      $0$ &      $0$ &
$\times$ &      $0$ & $\times$\\
     $0$ &      $1$ & $\times$ &
     $1$ &      $1$ &      $1$ &
  $\ast$ &      $1$ &      $1$ &
$\times$ &      $1$ & $\times$\\
     $0$ &   $\ast$ &      $0$ &
     $1$ &   $\ast$ &      $1$ &
  $\ast$ &   $\ast$ &   $\ast$ &
$\times$ &   $\ast$ & $\times$\\
     $0$ & $\times$ & $\times$ &
     $1$ & $\times$ & $\times$ &
  $\ast$ & $\times$ & $\times$ &
$\times$ & $\times$ & $\times$\\
\hline
\end{tabular}
\end{center}
\end{table}

\begin{figure}[t]
\begin{center}
\hrulefill\\
  \textbf{Subroutine~A}
\vspace{-2mm}
\begin{description}
\setlength\itemsep{-3pt}
\item[Input:]
one-qubit quantum registers ${\bfR_0}$, ${\bfS}$,
a classical variable $\status\in \{ \mbox{``$\eligible$''}, \mbox{``$\ineligible$''}\}$,
integers ${n,d}$
\item[Output:]
one-qubit quantum registers ${\bfR_0}$ and ${\bfS}$
\end{description}
\vspace{-6mm}
\begin{enumerate}
\setlength\itemsep{-3pt}
  \item
    Prepare two-qubit quantum registers ${\bfX_0^{(1)}, \ldots, \bfX_d^{(1)},
      \ldots, \bfX_0^{(n-1)}, \ldots, \bfX_d^{(n-1)}, \bfX_0^{(n)}}$.\\
    If ${\status = \mbox{``$\eligible$,''}}$
    copy the content of $\bfR_0$ to $\bfX_0^{(1)}$;
    otherwise
    set the content of $\bfX_0^{(1)}$ to ``$\ast$.''
  \item
    For ${t:=1}$ to ${n-1}$, do the following:
\vspace{-6pt}
    \begin{itemize}
\setlength\itemsep{-0pt}
      \item[2.1]
        Copy the content of $\bfX_0^{(t)}$ to each of ${\bfX_1^{(t)}, \ldots, \bfX_d^{(t)}}$.
      \item[2.2]
        Exchange the qubit in $\bfX_i^{(t)}$
        with the party connected via port $i$ for ${1 \leq i \leq d}$
        (i.e., the original qubit in $\bfX_i^{(t)}$ is sent
         via port $i$,
         and the qubit received 
         via that port
         is newly set in $\bfX_i^{(t)}$).
      \item[2.3]
        Set the content of $\bfX_0^{(t+1)}$ to
        ${x_0^{(t)} \circ x_1^{(t)} \circ \cdots \circ x_d^{(t)}}$,
        where $x_i^{(t)}$ denotes the content of $\bfX_i^{(t)}$ for ${0 \leq i \leq d}$.
    \end{itemize}
\vspace{-6pt}
  \item
    If the content of $\bfX_0^{(n)}$ is ``$\times$,''
    turn the content of $\bfS$ over
    (i.e., if the content of $\bfS$ is ``$\consistent$,''
    it is flipped to ``$\inconsistent$,'' and vice versa).
  \item
    Invert every computation and communication in Step 2.\\
  \item
    Invert every computation in Step 1.
  \item
    Output quantum registers $\bfR_0$ and $\bfS$.
\end{enumerate}
\vspace{-1\baselineskip}
\hrulefill
\end{center}
\vspace{-1\baselineskip}
\caption
{Subroutine~A.}
\label{Figure: Subroutine A}
\end{figure}

\subsection{Subroutine~B}
\label{Subsection: Subroutine B}

Suppose that, among $n$ parties, $k$ parties are still eligible
and share the $k$-cat state ${(\ket{0}^{\otimes k} + \ket{1}^{\otimes k} )/\sqrt{2}}$ in
their $\bfR_0$'s.
Subroutine~B has the goal of transforming the $k$-cat state
to an inconsistent state with certainty by using $k$
fresh ancilla qubits that are initialized to $\ket{0}$, when $k$ is given.
Figure~\ref{Figure: Subroutine B}
gives a precise description of Subroutine~B,
where
$\{ U_k \}$ and $\{ V_k \}$ are two families of unitary operators,
\[
U_k
=
\frac{1}{\sqrt{2}}
\begin{pmatrix}
                   1 & e^{-i \frac{\pi}{k}}\\
-e^{i \frac{\pi}{k}} &                    1
\end{pmatrix},
\]
\[
V_k
=
\frac{1}{\sqrt{R_k+1}}
\begin{pmatrix}
1/\sqrt{2}
&
0
&
\sqrt{R_k}
&
e^{i \frac{\pi}{k}}/\sqrt{2}
\\
1/\sqrt{2}
&
0
&
-\sqrt{R_k}e^{-i \frac{\pi}{k}}
&
e^{-i \frac{\pi}{k}}/\sqrt{2}
\\
\sqrt{R_k}
&
0
&
\frac{e^{-i \frac{\pi}{2k}}I_k}{i\sqrt{2}R_{2k}}
&
-\sqrt{R_k}
\\
0
&
\sqrt{R_k+1}
&
0
&
0
\end{pmatrix},
\iffalse{
V_k
=
\sqrt{\frac{R_k}{R_k+1}}
\begin{pmatrix}
\frac{1}{\sqrt{2R_k}}
&
0
&
\frac{e^{i \frac{\pi}{k}}}{\sqrt{2R_k}}
&
1
\\
\frac{1}{\sqrt{2R_k}}
&
0
&
\frac{e^{-i \frac{\pi}{k}}}{\sqrt{2R_k}}
&
- \frac{1+e^{-i \frac{\pi}{k}}}{1+e^{i \frac{\pi}{k}}}
\\
1
&
0
&
-1
&
\frac{e^{-i \frac{\pi}{k}} -e^{i \frac{\pi}{k}}}{1+e^{i \frac{\pi}{k}}}
\\
0
&
\sqrt{\frac{R_k+1}{R_k}}
&
0
&
0
\end{pmatrix},
}\fi
\iffalse{
V_k
=
\sqrt{\frac{R_k}{R_k+1}}
\begin{pmatrix}
\frac{1}{\sqrt{2R_k}}
&
0
&
1
&
\frac{e^{i \frac{\pi}{k}}}{\sqrt{2R_k}}
\\
\frac{1}{\sqrt{2R_k}}
&
0
&
- \frac{1+e^{-i \frac{\pi}{k}}}{1+e^{i \frac{\pi}{k}}}
&
\frac{e^{-i \frac{\pi}{k}}}{\sqrt{2R_k}}
\\
1
&
0
&
\frac{e^{-i \frac{\pi}{k}} -e^{i \frac{\pi}{k}}}{1+e^{i \frac{\pi}{k}}}
&
-1
\\
0
&
\sqrt{\frac{R_k+1}{R_k}}
&
0
&
0
\end{pmatrix},
}\fi
\]
where $R_k$ and $I_k$  are the real and imaginary parts of $e^{i
  \frac{\pi}{k}}$, respectively.

\begin{figure}[t]
\begin{center}
\hrulefill\\
  \textbf{Subroutine~B}
\vspace{-2mm}
\begin{description}
\item[Input:]
\setlength\itemsep{-3pt}
one-qubit quantum registers ${\bfR _0, \bfR _1}$, an integer $k$
\item[Output:]
one-qubit quantum registers ${\bfR _0, \bfR _1}$
\end{description}
\vspace{-6mm}
\begin{enumerate}
\setlength\itemsep{-3pt}
  \item
    If $k$ is even,
    apply $U_k$ to the qubit in $\bfR _0$;
    otherwise copy the content in $\bfR _0$ to that in 
$\bfR _1$, and then 
apply $V_k$ to the qubits in $\bfR _0$ and $\bfR _1$.
  \item Output quantum registers $\bfR _0$ and $\bfR _1$.
\end{enumerate}
\vspace{-1\baselineskip}
\hrulefill
\end{center}
\vspace{-1\baselineskip}
\caption{Subroutine~B.}  
\label{Figure: Subroutine B}
\end{figure}
The point is that
the amplitudes of the states
$\ket{00}^{\otimes k}$,
$\ket{01}^{\otimes k}$,
$\ket{10}^{\otimes k}$,
and $\ket{11}^{\otimes k}$
shared by $k$ eligible parties
in their registers $\bfR _0$ and $\bfR _1$  
are simultaneously zero
after every eligible party applies Subroutine~B with parameter $k$,
if the qubits
in $\bfR _0$s of all eligible parties form the
$k$-cat state.
The next two lemmas describe this rigorously.

Instead of $U_k$, we give a proof for a more general case.

\begin{lemma}\label{lm:unitary operator U(psi,t)}
Suppose that $k$ parties each have one
of $k$-cat-state qubits for any even integer ${k\geq 2}$. 
After
every party 
applies
\[
U_k(\psi,t) = 
  \frac{1}{\sqrt{2}}
\left( 
  \begin{array}{cc}
e^{i\psi} & e^{i(\psi-\frac{2t+1}{k}\pi)}\\
-e^{-i(\psi-\frac{2t+1}{k}\pi)} & e^{-i\psi}
  \end{array}
\right)
\]
to his qubit,
the resulting $k$-qubit state is inconsistent over the set of the
indices of the $k$ parties,
where $\psi$ and $t$ are any fixed real and integer values, respectively.
\end{lemma}
\begin{proof}
$U_k(\psi,t)$ is unitary since $U_k(\psi,t)U_k(\psi,t)^{\dag} =U_k(\psi,t)^{\dag} U_k(\psi,t)=I$, where 
$U_k(\psi,t)^{\dag}$ is the adjoint of $U_k(\psi,t)$, and
$I$ is
the two-dimensional identity operator.
It is sufficient to prove
that the amplitudes of states
$\ket{0}^{\otimes k}$ and $\ket{1}^{\otimes k}$
are both zero after every party applies $U_k(\psi,t)$ to his $k$-cat-state qubit.

After every party applies $U_k(\psi,t)$,
the amplitude of state 
$\ket{0}^{\otimes k}$
is
$$\frac{1}{\sqrt{2}}
\left(
\left(\frac{e^{i\psi}}{\sqrt{2}}\right)^k 
+
\left(\frac{e^{i(\psi-\frac{2t+1}{k}\pi)}}{\sqrt{2}}\right)^k 
\right)
=0.
$$
The amplitude of state 
$\ket{1}^{\otimes k}$
is 
$$
\frac{1}{\sqrt{2}}
\left(
\left(\frac{-e^{-i(\psi - \frac{2t+1}{k}\pi)}}{\sqrt{2}}\right)^k 
+
\left(\frac{e^{-i\psi}}{\sqrt{2}}\right)^k 
\right)
=0,
$$
since $k$ is even.
\end{proof}

\begin{corollary}
\label{lm:unitary operator U}
Suppose that
$k$ parties each have one
of the $k$ qubits that are in a $k$-cat-state
for any even integer ${k \geq 2}$.
After every party applies $U_k\otimes I$
to the qubit and a fresh ancilla qubit, 
the resulting $2k$-qubit state is inconsistent over $S$, where $S$ is
the set of the indices of the $k$ parties.
\end{corollary}
\begin{proof}
By setting both $\psi$ and $t$ to 0 in Lemma \ref{lm:unitary operator
  U(psi,t)},
the proof is completed.
\end{proof}

The case for $V_k$ can be proved similarly.
\begin{lemma}
Suppose that
$k$ parties each have two
of the $2k$ qubits that are in a $2k$-cat-state
for any odd integer ${k \geq 3}$. 
After every party applies $V_k$
to his two qubits,
the resulting $2k$-qubit state is inconsistent over $S$, where $S$ is the set of
indices of the $k$ parties.
\end{lemma}
\begin{proof}
The matrix of $V_k$ is well-defined since the denominator in any element of
$V_k$ is positive since $R_k+1>0$ and $R_{2k}>0$ for $k\geq 3$.
We can verify that $V_k$ is unitary by some calculation.

To complete the proof, we will show that the amplitudes of states
$\ket{00}^{\otimes k}$, $\ket{01}^{\otimes k}$, 
$\ket{10}^{\otimes k}$, and $\ket{11}^{\otimes k}$ 
are all zero after every party applies
$V_k$ to his two qubits, since these states imply
that all parties observe the same two-bit value by measuring their two
qubits.
Here we assume that the ordering of the $2k$ qubits is such that
party $l$ has the $(2l-1)$st and $2l$th qubits for $l=1,2,\ldots ,k$.

After every party applies $V_k$,
the amplitude of state 
$\ket{00}^{\otimes k}$
is
\[
\frac{1}{\sqrt{2}}
\left\{
\left(
\frac{1}{\sqrt{2(R_k+1)}}
\right)
^k 
+
\left(
\frac{e^{i\frac{\pi}{k}}}{\sqrt{2(R_k+1)}}
\right)
^k 
\right\}
=0.
\]
In the same way, the amplitudes of states 
$\ket{01}^{\otimes k}$ and $\ket{10}^{\otimes k}$ 
are
\[
\frac{1}{\sqrt{2}}
\left\{
\left(
\frac{1}{\sqrt{2(R_k+1)}}
\right)
^k 
+
\left(
\frac{e^{-i\frac{\pi}{k}}}{\sqrt{2(R_k+1)}}
\right)
^k 
\right\}
=0,
\]
\[
\frac{1}{\sqrt{2}}
\left\{
\left(
\frac{R_k}{\sqrt{R_k+1}}
\right)
^k 
+
\left(
- \frac{R_k}{\sqrt{R_k+1}}
\right)
^k 
\right\}
=0,
\]
respectively, 
since $k$ is odd.
The amplitude of state 
$\ket{11}^{\otimes k}$ 
is obviously 0.
\end{proof}

From the above two lemmas, the correctness of Subroutine~B is immediate.
\begin{lemma}
\label{lm:subroutineB-correctness}
 Suppose that $k$ parties each have a qubit in one-qubit register
  ${\bfR _0}$
whose content forms a
  $k$-cat state together with the contents of the $(k-1)$ qubits of the other
  parties; further, suppose that they each prepare a fresh ancilla qubit initialized to
  $\ket{0}$ in another one-qubit register $\bfR_1$. After running Subroutine~B with
  $\bfR_0$, $\bfR_1$ and $k$, the qubits in $\bfR_0$s and $\bfR_0$s
  form an inconsistent state over $S$, where $S$ is the set of
the indices of the $k$ parties.
\end{lemma}

The next lemma is obvious.
\begin{lemma}
\label{lm:subroutineB-complexity}
  Subroutine~B takes $O(1)$ time  and needs no communication.
\end{lemma}

\subsection{Subroutine~C}
\label{Subsection: Subroutine C}

Subroutine~C is a classical algorithm
that computes the maximum value over the values of all parties.
It is very similar to Subroutine~A.
In fact, Subroutines A and C can be merged into one subroutine,
although we will explain them separately for simplicity.
Figure~\ref{Figure: Subroutine C}
gives a precise description of Subroutine~C.

\begin{figure}[t]
\begin{center}
\hrulefill\\
  \textbf{Subroutine~C}
\vspace{-2mm}
\begin{description}
\setlength\itemsep{-3pt}
\item[Input:]
integers $z$, $n$, $d$
\item[Output:]
an integer $z_{\max}$
\end{description}
\vspace{-6mm}
\begin{enumerate}
\setlength\itemsep{-3pt}
  \item
    Let ${z_{\max} := z}$.
  \item
    For ${t:=1}$ to ${n-1}$, do the following:
\vspace{-6pt}
    \begin{itemize}
\setlength\itemsep{-0pt}
      \item[2.1]
        Let ${y_0 := z_{\max}}$. 
      \item[2.2]
        Send $y_0$ via every port $i$ for ${1 \leq i \leq d}$.\\
        Set $y_i$ to the value received
        via port $i$ for ${1 \leq i \leq d}$. 
     \item[2.3]
        Let ${z_{\max} := \max_{0 \leq i \leq d} y_i}$.
    \end{itemize}
\vspace{-6pt}
  \item
    Output $z_{\max}$.
\end{enumerate}
\vspace{-1\baselineskip}
\hrulefill
\end{center}
\vspace{-1\baselineskip}
\caption{Subroutine~C.}
\label{Figure: Subroutine C}
\end{figure}

\begin{lemma}
\label{lm:subroutineC-correctness}
  Suppose that each party $l$ has integer $z_l$ and $d_l$ neighbors in an
  $n$-party distributed system. If every party $l$ runs Subroutine~C
  with ${z:=z_l}$, $n$ and ${d:=d_l}$ as input, Subroutine~C outputs
  the maximum value $z_{\max}$ among all $z_l$s.
\end{lemma}
\begin{proof}
We will prove by induction the next claim:
 after repeating steps 2.1 to 2.3 $t$ times, $z_{\max}$ of party $l$ is the maximum among
  $z_j$s of all parties $j$ who can be reached
  from party $l$ via a path of length at most $t$.
When $t=1$, the claim obviously holds. Assume that the claim holds for
 $t=m$. After the next iteration of steps 2.1 to 2.3, $z_{\max}$ is updated to
 the maximum value
among $y_0$'s of party $l$ and his neighbors.
Since $y_0$ is the $z_{\max}$ of the previous iteration,
the claim holds for $t=m+1$ due to the assumption.
Since any graph has diameter at most $n-1$, Subroutine~C outputs the 
maximum value $z_{\max}$ among all $z_i$s.
\end{proof}

In quite a similar way to the proof of Lemma \ref{lm:subroutineA-complexity},
we have the next lemma.
\begin{lemma}
\label{lm:subroutineC-complexity}
Let $\abs{E}$ and $D$ be the number of edges and the maximum degree of the
underlying graph, respectively.
Subroutine~C takes $O(n)$ rounds  and $O(Dn)$ time.
The total communication complexity over all parties is 
$O(|E|n)$.
\end{lemma}

\subsection{Complexity analysis and generalization}
\label{Subsection: Complexity Analysis}
Now we prove Theorem~\ref{Theorem: Complexity of Algorithm I for n}.
\setcounter{tmptheorem}{\value{theorem}}
\setcounter{theorem}{1}
\addtocounter{theorem}{-1}
\begin{theorem}
Let $\abs{E}$ and $D$ be the number of edges and the maximum degree of the
underlying graph, respectively.
Given a classical variable $\status$ initialized to ``$\eligible$'' and 
the number $n$ of parties,
Algorithm~I exactly elects a unique leader
in $O(n^2)$ rounds and $O(D n^2)$ time.
Each party connected with $d$ parties
requires $O(d n^2)$-qubit communication, and 
the total communication complexity over all parties
is $O(\abs{E} n^2)$.
\label{Theorem: Complexity of Algorithm I for n}
\end{theorem}
\setcounter{theorem}{\value{tmptheorem}}
\begin{proof}
Let $S_i$ be the set of the indices of 
parties with
  $\status=\mbox{``$\eligible$''}$ (i.e., eligible parties)
just before phase $i$.
From Lemmas~\ref{lm:subroutineA-correctness}~and~\ref{lm:subroutineB-correctness},
we can see that, in each phase $i$, Algorithm~I generates
an inconsistent state over $S_i$,
if $k=\abs{S_i}$.
Algorithm~I then decreases the number of the eligible parties
in step 2.5 by at least one, which is implied by Lemma~\ref{lm:subroutineC-correctness}.
If $k$ is not equal to $\abs{S_i}$,
the number of the eligible parties is decreased or unchanged.
We can thus prove that $k$ is always at least $\abs{S_i}$ in any phase
$i$ by induction, since ${k=\abs{S_1}=n}$ before entering phase $1$
and $k$ is decreased by 1 in every phase.
It is stressed that there is always at least one eligible party,
since the eligible parties having $z=z_{\max}$ at step 2.5 remain
eligible.
It follows that, after step 2, the number of eligible parties
is exactly 1. This proves the correctness of Algorithm~I.

As for complexity, Subroutines~A,~B~and~C are dominant in step 2.
Due to Lemmas \ref{lm:subroutineA-complexity}, \ref{lm:subroutineB-complexity} and \ref{lm:subroutineC-complexity}; the total communication complexity is
$O(|E|n)\times n=O(|E|n^2)$ (each party with $d$ neighbors
incurs $O(d n^2)$ communication complexity);
the time complexity is $O(Dn)\times n = O(Dn^2)$;
the number of rounds required is $O(n)\times n=O(n^2)$.
\end{proof}

If each party knows
only the upper bound $N$ of the number of parties in advance,
each party has only to perform Algorithm~I
with 
$N$ instead of $n$.
The correctness in this case is obvious from the proof of 
Theorem~\ref{Theorem: Complexity of Algorithm I for n}.
The complexity is described
simply by replacing every 
$n$ by $N$
in Theorem~\ref{Theorem: Complexity of Algorithm I for n}.

\setcounter{tmptheorem}{\value{theorem}}
\setcounter{theorem}{3}
\addtocounter{theorem}{-1}
\begin{corollary}
Let $\abs{E}$ and $D$ be the number of edges and the maximum degree of the
underlying graph, respectively.
Given a classical variable $\status$ initialized to ``$\eligible$'' and 
the number of parties,  $N$,
Algorithm~I exactly elects a unique leader
in $O(N^2)$ rounds and $O(D N^2)$ time.
Each party connected with $d$ parties
incurs $O(d N^2)$-qubit communication, and 
the total communication complexity over all parties
is $O(\abs{E} N^2)$.
\end{corollary}
\setcounter{theorem}{\value{tmptheorem}}

\iffalse
\begin{theorem}
Let $\abs{E}$ and $D$ be the number of edges and the maximum degree of the
underlying graph, respectively.
Even when each party initially knows
only the upper bound $N$ of the number of parties,
Algorithm~I is an exact quantum leader election algorithm
that runs in $O(N^2)$ rounds and $O(D N^2)$ time.
Each party connected with $d$ parties
requires $O(d N^2)$ communication complexity,
and the total communication complexity over all parties
is $O(\abs{E} N^2)$.
\end{theorem}
\fi

Furthermore, Algorithm~I is easily modified
so that it works well even in the asynchronous settings.
Note that all parties receive messages via each port in each round.
In the modified version,
each party postpones performing the operations of the $(i+1)$st round
until he finishes receiving all messages that are supposed to be received
in the $i$th round.
If all communication links work in the first-in-first-out manner,
it is easy to recognize the messages sent in the $i$th round for any $i$.
Otherwise, we tag every message,
which increases the communication and time complexity
by the multiplicative factor $O(\log n)$,
in order to know in which round every received message was sent.
This modification enables us to simulate synchronous behavior
in asynchronous networks.

%%% Local Variables: 
%%% mode: latex
%%% TeX-master: "TanKobMat07"
%%% End: 

\section{Quantum leader election algorithm II}
\label{sec:Quantum Leader Election Algorithm II}

Our second algorithm works well even on networks whose underlying graph
is directed (and strongly-connected).
Just for ease of understainding, we first describe the second algorithm 
on \emph{undirected} networks, and then modify it 
in a fairly trivial way 
in Subsection~\ref{subsec:algorithm II  for directed network}
so that it works well
on directed networks. 

To work on networks whose underlying graph is directed,
the second algorithm contains no inverting operations involving quantum
processing: the second algorithm makes the most of a classical elegant
technique, called \emph{view},
to make quantum parts ``one-way.''
As a by-product,
the second algorithm requires less quantum communication than the
first algorithm, although the \emph{total} communication complexity increases.

View was originally introduced by Yamashita and
Kameda~\cite{YamKam96IEEETPDS-1,yamashita-kameda99}
to characterize the topology of anonymous networks on which the leader
election problem can  be solved deterministically.
However, a na\"{\i}ve application of view
incurs exponential classical time/communication complexity.
This paper introduces a new technique called \emph{folded view},
which allows the algorithm to still run
in time/communication polynomial with respect to the number of parties.

\subsection{View and folded view}
\label{subsec: view}

First, we briefly review the classical technique, \emph{view}.
Let ${G=(V,E)}$ be the underlying network topology
and let ${n = \abs{V}}$.
Suppose that each party corresponding to node ${v \in V}$, \emph{or} simply party $v$,
has a value ${x_v \in U}$
for a finite subset $U$
of the set of integers,
and a mapping ${X \colon V \rightarrow U}$ is
defined by ${X(v) = x_v}$.
We use the value given by $X$ to identify the label of node in $G$.
For each $v$ and port numbering $\sigma$,
\emph{view} ${T_{G,\sigma, X} (v)}$ is
a labeled, rooted tree with infinite depth defined recursively as follows:
(1) ${T_{G,\sigma, X} (v)}$ has root $u$ with label $X(v)$, corresponding to $v$,
(2) for each vertex $v_j$ adjacent to $v$ in $G$,
${T_{G,\sigma,X} (v)}$ has vertex $u_j$ labeled with $X(v_j)$,
and an edge from root $u$ to $u_j$ with label ${\Label((v,v_j))}$,
where ${\Label((v,v_j))=(\sigma[v] (v,v_j),\sigma[v_j] (v,v_j))}$,
and (3) $u_j$ is the root of ${T_{G,\sigma,X} (v_j)}$.
It should be stressed that $v$, $v_j$, $u$, and  $u_j$ are
not identifiers of parties and are introduced just for definition. 
For simplicity,
we often use 
${T_{X}(v)}$ instead of ${T_{G,\sigma,X}(v)}$,
because we usually discuss views
of some fixed network with some fixed port numbering.
The \emph{view of depth $h$ with respect to $v$}, denoted by ${T^{h}_{X}(v)}$, is
the subtree of ${T_{X}(v)}$
of depth $h$ with the same root as  ${T_{X}(v)}$.

If two views $T_{X} (v)$ and $T_{X} (v')$ for $v,v'\in V$ are isomorphic (including
edge labels and node labels, but ignoring local names of vertices such
as $u_i$), their relation is denoted by $T_{X} (v)\equiv T_{X} (v')$.
With this relation, 
$V$ is divided into
equivalence classes; $v$ and $v'$ are in the same class if and only if
$T_{X} (v)\equiv T_{X}(v')$.  In
\cite{YamKam96IEEETPDS-1,yamashita-kameda99}, it was proved that all
classes have the same cardinality for fixed $G,\sigma$ and $X$; the cardinality
is denoted by $c_{G,\sigma,X}$, or simply $c_{X}$
(the maximum value of
$c_{G,\sigma,X}$ over all port numbering $\sigma$ is called
\emph{symmetricity} $\gamma(G,X)$ and used to give the necessary and
sufficient condition to exactly solve $\LE_n$ in anonymous 
classical networks).
We denote the set of
non-isomorphic views by $\Gamma _{G,\sigma,X}$, i.e., $\Gamma
_{G,\sigma,X}=\{T_{G,\sigma,X}(v): v\in V\}$, and the set of
non-isomorphic views 
of depth $h$ 
by $\Gamma _{G,\sigma,
  X}^{h}$, i.e., $\Gamma ^h_{G,\sigma,X}=\{T^{h}_{G,\sigma,X}(v):
v\in V\}$.  For simplicity, we may use $\Gamma _{X}$ and $\Gamma
^h_{X}$ instead of $\Gamma _{G,\sigma,X}$ and $\Gamma
^h_{G,\sigma,X}$, respectively.  We can see that $c_{X}=n/\abs{\Gamma _{X}}$, since
the number of views isomorphic to $T_X \in \Gamma _{X}$ is constant
over all $T_X$.  For any subset $S$ of $U$, let $\Gamma _{X}(S)$ be
the maximal subset of $\Gamma _{X}$ such that any view $T_X\in
\Gamma _{X}(S)$ has its root labeled with a value in $S$.  Thus the
number $c_{X}(S)$ of
parties having values in $S$ is $c_{X}|\Gamma_{X}(S)|=n|\Gamma
_{X}(S)|/|\Gamma _{X}|$.  
When $S$ is a singleton set $\{ s \}$, we may use 
$\Gamma_{X}(s)$ and $c_{X}(s)$ instead of $\Gamma_{X}(\{ s\})$ and $c_{X}(\{ s\})$.

To compute $c_{X}(S)$,
every party $v$ 
constructs $T^{2(n-1)}_{X}(v)$, and then computes $|\Gamma_{X}|$ and
$|\Gamma_{X}(S)|$.
To construct ${T^{h}_{X}(v)}$,
in the first round,
every party $v$ constructs ${T^0_{X}(v)}$, i.e., the root of $T^h_{X} (v)$.
If every party $v_j$ adjacent to $v$ has $T^{i-1}_{X}(v_j)$ in the $i$th round, 
$v$ can construct $T^{i}_{X}(v)$ in the $(i+1)$st round
by exchanging a copy of $T^{i-1}_{X}(v)$ for a copy of
$T^{i-1}_{X}(v_j)$ for each $j$.
By induction, in the $(h+1)$st round,
each party $v$ can construct $T^h_{X}(v)$.
It is clear that,
for each $v'\in V$,
at least one node in $T_{X}^{n-1}(v)$
corresponds to $v'$, 
since there is at least
one path of length of at most $(n-1)$ between any pair of parties.
Thus party $v$ computes $|\Gamma_{X}|$ and
$|\Gamma_{X}(S)|$ by checking the equivalence
of every pair of views that have their roots in 
$T_{X}^{n-1}(v)$.
The view equivalence
can be checked in finite steps, since $T_{X} (v)\equiv
T_{X}(v')$ if and only if 
$T^{n-1}_{X}(v)\equiv T^{n-1}_{X}(v')$ for $v,v'\in V$~\cite{norris95}.
This implies that 
$|\Gamma_{X}|$ and
$|\Gamma_{X}(S)|$ can be computed from
$T^{2(n-1)}_{X}(v)$.

Note that the size of ${T^{h}_{X}(v)}$ is exponential in $h$,
which results in exponential time/communication complexity
in $n$
when we construct it if $h=2(n-1)$.
To reduce the time/communication complexity to something bounded by a
polynomial, 
we create the new technique called \emph{folded view}
by generalizing Ordered Binary Decision Diagrams (OBDD)~\cite{Bry86IEEETC}.
A \emph{folded view (f-view) of depth $h$}
is a vertex- and edge-labeled directed acyclic multigraph
obtained by merging nodes at the same level in $T^h_{X}(v)$
into one node
if the subtrees rooted at them are isomorphic.
An f-view is said to be \emph{minimal}
and is denoted by $\widetilde{T}^h_{X}(v)$
if it is obtained by maximally
merging nodes of view $T^h_{X}(v)$ under the above condition.
For simplicity, we may call a minimal f-view just an f-view
in this section.
The number of nodes in each level of an f-view is obviously bounded by $n$,
and thus the total number of nodes in an f-view of depth $h$ is at most $hn$.
Actually, an f-view of depth $h$
can be recursively constructed in a similar manner to view construction
without unfolding intermediate f-views.
Details will be described in Section~\ref{sec: view compression}.

\begin{theorem}
\label{th:f-view}
If each party has a label of a constant-bit value,
every f-view of depth $h$ is constructed in
$O(D^2h^2n(\log n)^2)$ time for each party and $O(h)$ rounds
with $O(D|E|h^2n\log D)$ bits of classical communication.
Once $\widetilde{T}^{2(n-1)}_{X}(v)$ is constructed,
each party can compute
$\abs{\Gamma _X}$ and $\abs{\Gamma _X(S)}$
without communication in $O(Dn^5\log n)$ time,
where $S$ is any subset of range $U$ of $X$, and $\abs{E}$ and $D$ are
the number of edges and the maximum degree, respectively,  of the underlying graph.
\end{theorem}
\begin{remark}
  Kranakis and Krizanc~\cite{kranakis-krizanc-vandenberg94} gave
two algorithms that compute a Boolean function for distributed
inputs
on anonymous
networks.
In their first algorithm, every party
essentially constructs a view of depth $O(n)$
in $O(n^2)$ rounds and the total communication complexity over all
parties of $O(n^6\log n)$, followed by local computation
(their model assumes that every party knows the
topology of the network and thus the number of parties, but their
first algorithm can work even when every party knows only the number of parties).
Thus, our folded view quadratically reduces the number of rounds
required to compute a Boolean function,
with the same total communication complexity.
Note that their second algorithm
can compute a symmetric Boolean function with 
lower communication complexity, i.e., $O(n^5(\log n)^2)$,
and $O(n^3\log n)$ rounds, but it requires 
that every party knows the topology of network.
\end{remark}

\subsection{The algorithm}
\label{Subsection: The Algorithm II}

As in the previous section, 
we assume that the network is synchronous
and each party knows the number $n$ of parties prior to algorithm invocation.
Again our algorithm is easily generalized to 
the asynchronous case.
It is also possible to modify our algorithm
so that it works well even if only the upper bound $N$ of the number of parties
is given, which will be discussed in Subsection~\ref{subsec:algorithm II when only the upperbound is given}.

The algorithm consists of two stages, which we call Stages 1 and 2 hereafter.
Stage 1 aims to have the $n$ parties share a certain type of entanglement,
and thus, this stage requires the parties to exchange quantum messages.
In Stage 1, each party performs Subroutine~Q ${s = \ceil{\log n}}$ times
in parallel to share $s$ pure quantum states
${\ket{\phi^{(1)}}, \ldots, \ket{\phi^{(s)}}}$ of $n$ qubits.
Here, each $\ket{\phi^{(i)}}$ is of the form
${(\ket{x^{(i)}} + \ket{\overline{x}^{(i)}})/\sqrt{2}}$
for an $n$-bit string $x^{(i)}$ and its bitwise negation $\overline{x}^{(i)}$,
and the $l$th qubit of each $\ket{\phi^{(i)}}$
is possessed by the $l$th party.
It is stressed that only one round of quantum communication is necessary in Stage 1.

In Stage 2,
the algorithm decides a unique leader among the $n$ parties
by just local quantum operations and classical communications
with the help of the shared entanglement prepared in Stage 1.
This stage consists of at most $s$ phases,
each of which reduces the number of eligible parties
by at least half.
Let ${S_i \subseteq \{1, \ldots, n\}}$
be the set of all $l$s such that party $l$ is still eligible just
before entering phase $i$.
First every party runs Subroutine~\~{A}
to decide if state $\ket{\phi^{(i)}}$ is
consistent or inconsistent over $S_i$.
Here the consistent/inconsistent strings/states are defined in the
same manner as in the
previous section.
If state $\ket{\phi^{(i)}}$ is consistent,
every party performs Subroutine~\~{B},
which first transforms $\ket{\phi^{(i)}}$
into the $\abs{S_i}$-cat state
${(\ket{0}^{\otimes \abs{S_i}} + \ket{1}^{\otimes \abs{S_i}})/\sqrt{2}}$
shared only by eligible parties
and then calls Subroutine~B described in the previous section
to obtain an inconsistent state over $S_i$.
Each party $l$ then measures his qubits to obtain a label 
and performs Subroutine~\texttilde{C}.
to find the minority among all labels.
The number of eligible parties is then 
reduced by at least half via minority voting with respect to the labels.

More precisely, 
each party $l$ 
having $d_l$ adjacent parties
performs Algorithm~II
described in Figure~\ref{Figure: Quantum leader election algorithm II}
with parameters ``$\eligible$,'' $n$, and $d_l$.
The party who obtains  output ``$\eligible$'' is the unique leader.

\begin{figure}[t]
\begin{center}
\hrulefill\\
  \textbf{Algorithm~II}
\vspace{-2mm}
\begin{description}
\setlength\itemsep{-3pt}
  \item[Input:]
    a classical variable $\status\in \{ \mbox{``$\eligible$''}, \mbox{``$\ineligible$''}\}$,
    integers ${n,d}$
  \item[Output:]
    a classical variable $\status\in \{ \mbox{``$\eligible$''}, \mbox{``$\ineligible$''}\}$
\end{description}
\vspace{-6mm}
\begin{description}
\setlength\itemsep{-3pt}
  \item[\underline{Stage 1}:]~\\
    Let ${s: = \ceil{\log n}}$
    and prepare one-qubit quantum registers
    ${\bfR _0 ^{(1)}, \ldots, \bfR _0 ^{(s)}}$
    and ${\bfR _1^{(1)}, \ldots, \bfR _1^{(s)}}$,
    each of which is initialized to the $\ket{0}$ state.\\
    Perform $s$ attempts of Subroutine~Q in parallel,
    each with $\bfR _0 ^{(i)}$ and $d$ for ${1 \leq i \leq s}$,
    to obtain $d$-bit string $y^{(i)}$
    and to share
    ${\ket{\phi^{(i)}} = (\ket{x^{(i)}} + \ket{\overline{x}^{(i)}})/\sqrt{2}}$
    of $n$ qubits.
  \item[\underline{Stage 2}:]~\\
    Let ${k := n}$.\\
    For ${i := 1}$ to $s$, repeat the following:
\vspace{-6pt}
    \begin{enumerate}
\setlength\itemsep{-3pt}
      \item
        Perform Subroutine~\texttilde{A}
        with $\status$, $n$, $d$, and $y^{(i)}$ 
        to obtain its output $\consistency$.
      \item
        If ${\consistency = \mbox{``$\consistent$,''}}$
        perform Subroutine~\texttilde{B}
        with $\bfR _0 ^{(i)}$, $\bfR _1^{(i)}$, $\status$, $k$, $n$, and $d$.
      \item
        If ${\status = \mbox{``$\eligible$,''}}$
        measure the qubits in $\bfR _0 ^{(i)}$ and $\bfR _1^{(i)}$
        in the $\{ \ket{0}, \ket{1} \}$ basis
        to obtain a nonnegative integer $z\ (0\leq z\leq 3)$;
        otherwise set ${z := -1}$.\\
        Perform Subroutine~\texttilde{C}
        with $\status$, $z$, $n$, and $d$ 
        to compute nonnegative integers $z_{\minor}$ and $c_{z_{\minor}}$.
      \item
        If ${z \neq z_{\minor}}$, let ${\status := \mbox{``$\ineligible$.''}}$\\
        Let ${k := c_{z_{\minor}}}$.
      \item
        If ${k=1}$, terminate and output $\status$.
    \end{enumerate}
  \end{description}
\vspace{-1\baselineskip}
\hrulefill
\end{center}
\vspace{-1\baselineskip}
\caption{Quantum leader election algorithm II.}
\label{Figure: Quantum leader election algorithm II}
\end{figure}

\subsubsection{Subroutine~Q:}
Subroutine~Q is mainly for the purpose of sharing a cat-like quantum state
${\ket{\phi}=(\ket{x}+\ket{\overline{x}})/\sqrt{2}}$ for an $n$-bit
random string $x$.
It also outputs a classical string,
which is used in Stage 2 for each party to obtain the information on ${\ket{\phi}}$ 
via  just classical communication.
This subroutine can be performed in parallel,
and thus Stage 1 involves only one round of quantum communication.
First each party prepares the state $(\ket{0}+\ket{1})/\sqrt{2}$
in a quantum register and computes the XOR of the contents of his own
and each adjacent party's registers. The party then measures the qubits 
whose contents are the results of the XORs. This results in the state of
the form $(\ket{x}+\ket{\overline{x}})/\sqrt{2}$.
Figure~\ref{Figure: Subroutine Q} gives a precise description of
Subroutine~Q.

\begin{figure}[t]
\begin{center}
\hrulefill\\
  \textbf{Subroutine~Q}
\vspace{-2mm}
\begin{description}
\setlength\itemsep{-3pt}
\item[Input:]
a one-qubit quantum register ${\bfR}_0$,
an integer $d$
\item[Output:]
a one-qubit quantum register ${\bfR}_0$,
a binary string $y$ of length $d$
\end{description}
\vspace{-6mm}
\begin{enumerate}
\setlength\itemsep{-3pt}
  \item
    Prepare $2d$ one-qubit quantum registers
    ${\bfR'_1, \ldots, \bfR'_d}$ and ${\bfS_1, \ldots, \bfS_d}$,
    each of which is initialized to the $\ket{0}$ state.
  \item
    Generate the $(d+1)$-cat state
    ${(\ket{0}^{\otimes (d+1)} + \ket{1}^{\otimes (d+1)})/\sqrt{2}}$
    in registers $\bfR_0$, ${\bfR'_1, \ldots, \bfR'_d}$.
  \item
    Exchange the qubit in $\bfR'_i$
    with the party connected via port $i$ for ${1 \leq i \leq d}$
    (i.e., the original qubit in $\bfR'_i$ is sent
     via  port $i$,
     and the qubit received via that port
     is newly set in $\bfR'_i$).
  \item
  Set the content of $\bfS_i$ to ${x_0 \oplus x_i}$, for ${1 \leq i \leq d}$,
    where $x_0$ and $x_i$ denote the contents of $\bfR_0$ and $\bfR'_i$, respectively.
  \item
    Measure the qubit in $\bfS_i$ in the $\{ \ket{0}, \ket{1} \}$ basis
    to obtain bit $y_i$, for ${1 \leq i \leq d}$.\\
    Set ${y:= y_1 \cdots y_d}$.
  \item
    Apply CNOT controlled by the content of $\bfR_0$ and targeted to
    the content of each $\bfR'_i$ for ${i=1,2,\ldots ,d}$ to 
    disentangle $\bfR'_i$s.
  \item
    Output $\bfR_0$ and $y$.
\end{enumerate}
\vspace{-1\baselineskip}
\hrulefill
\end{center}
\vspace{-1\baselineskip}
\caption{Subroutine~Q.}  
\label{Figure: Subroutine Q}
\end{figure}
The next two lemmas are for correctness and complexity.
\begin{lemma}
\label{lm:subroutineQ-correctness}
For an $n$-party distributed system,
suppose that every party $l$ calls Subroutine~Q with a one-qubit register whose content is initialized to
$\ket{0}$ and  the number $d_l$ of his neighbors as input $\bfR_0$ and
$d$, respectively.
After performing Subroutine~Q, all parties share
$(\ket{x}+\ket{\overline{x}})/\sqrt{2}$
with certainty, where $x$ is a randomly chosen $n$-bit string.
\end{lemma}
\begin{proof}
After step 2 of Subroutine~Q,
the system state, i.e., the state in 
$\bfR_0$'s, ${\bfR'_1\mbox{'s}, \ldots, \bfR'_d\mbox{'s}}$ and
${\bfS_1\mbox{'s}, \ldots, \bfS_d\mbox{'s}}$
of all parties,
is the tensor product of the states of all parties as described by
formula~(\ref{fo:Q:step2}).
Notice that the state in $\bfR_0$'s and ${\bfR'_1\mbox{'s}, \ldots, \bfR'_d\mbox{'s}}$
of all parties is the
uniform superposition of some basis states in an orthonormal basis
of $2^{\sum _{l=1}^n(d_l+1)}$-dimensional Hilbert space:
the basis states
correspond one-to-one to $n$-bit integers $a$ and each of them
has the form $\ket{a_1}^{\otimes (d_1+1)}\otimes \cdots \otimes \ket{a_n}^{\otimes
  (d_n+1)}$, where $a_l$ is the $l$th bit of the binary expression of
$a$
and $a_l$ is the content of $\bfR_0$ of party $l$.
If we focus on the $l$th party's part of the basis state corresponding
to $a$, step~3 transforms $\ket{a_l}^{\otimes (d_l+1)}$ to 
$\ket{a_l}\left(\bigotimes_{j=1}^{d_l}\ket{a_{l_j}}\right)$,
where party $l$ is connected to party $l_j$ via port $j$.
More precisely, step 3 transforms the system state into the state as described in 
formula~(\ref{fo:Q:step3}).
After step 4, we have the state of 
formula~(\ref{fo:Q:step4}). Next every party $l$ measures the last $d_l$
registers $\bfS_i$'s at step 5.
  \begin{eqnarray}
    \bigotimes _{l=1}^{n}\ket{0}\ket{0}^{\otimes d_l}\ket{0}^{\otimes d_l}&\rightarrow&
    \bigotimes_{l=1}^{n}\frac{\ket{0}^{\otimes
    (d_l+1)}+\ket{1}^{\otimes (d_l+1)}}{\sqrt{2}}\ket{0}^{\otimes d_l}\label{fo:Q:step2}\\
&\rightarrow& 
\frac{1}{\sqrt{2^n}}
\sum _{a=0}^{2^n-1}\bigotimes_{l=1}^{n}\left\{\ket{a_l}
\left(\bigotimes_{j=1}^{d_l}\ket{a_{l_j}}\right)
\ket{0}^{\otimes d_l}\right\}\label{fo:Q:step3}\\
&\rightarrow& 
\frac{1}{\sqrt{2^n}}
\sum _{a=0}^{2^n-1}\bigotimes_{l=1}^{n}\left\{\ket{a_l}
\left(\bigotimes_{j=1}^{d_l}\ket{a_{l_j}}\right)
\left(\bigotimes_{j=1}^{d_l}\ket{a_l\oplus a_{l_j}}\right)\right\}\label{fo:Q:step4}
  \end{eqnarray}

\begin{claim}
\label{claim: correctness of subroutine Q}
  Suppose that every party $l$ has obtained measurement results
  $y{(l)}=y_1{(l)}y_2{(l)}\cdots y_{d_l}{(l)}$ of $d_l$ bits
where $y_j{(l)}\in \{ 0,1\}$.
There are exactly two binary strings $a=a_1a_2\cdots a_n$ 
that satisfy equations
$a_l\oplus a_{l_j}=y_j{(l)} (l=1,\ldots , n, j=1,\ldots ,d_l)$.
If the binary strings are $A$ and $\overline{A}$, then $\overline{A}$ is the bit-wise negation
  of $A$.
\end{claim}
\begin{proof}
We call binary strings $a$ ``solutions'' of the equations.
By definition, there is at least one solution.
If $A$ is such a string, obviously its bit-wise negation
$\overline{A}$ is also a solution by the fact that $a_i\oplus
a_j=\overline{a_i}\oplus \overline{a_j}$ for $1\leq i,j\leq n$.
We will prove that there is the unique solution such that 
$a_1=0$. It follows that
there is the unique solution such that 
$a_1=1$ since the bitwise negation of a solution is also a solution.
This completes the proof.

Let $\{ V_0, V_1, \ldots ,V_{p}\}$ be the partition of the set $V$
 of the indices of parties such that $V_0=\{1\}$ and $V_i=
\Adj(\bigcup_{m=0}^{i-1}V_m)\setminus \bigcup_{m=0}^{i-1}V_m$, where
  $p$ is the maximum length of the shortest path from party 1 to party $l$
  over all $l$, and $\Adj (V')$ for a set $V'\subseteq V$
is the set of neighbors of the parties in $V'$.

Equations $a_l\oplus a_{l_j}=y_j{(l)}$ are equivalent to 
$a_{l_j}=y_j{(l)}\oplus a_l\ (l=1,\ldots , n, j=1,\ldots ,d_l)$.
Assume that $a_1=0$. For all $l$ in $V_1$, $a_l$ is uniquely
determined by the equations.
Similarly, if $a_l$ is fixed for all $l$ in $\bigcup_{m=0}^{i-1}V_{m}$, 
$a_l$ is uniquely determined for all $l$ in $V_{i}$.
Since the underlying graph of the distributed system is connected,
$a_l$ is uniquely determined for all $l$.
\end{proof}

From the above claim, we get the superposition of two basis states
corresponding to $A$ and its bit-wise negation $\overline{A}$  after step 5 as
described by formula~(\ref{fo:Q:step5}), where $A_l$ is the $l$th bit
of $A$.
Step 6 transforms the state into that represented by formula~(\ref{fo:Q:step6}), in which
registers
$\bfR'_i$'s of all parties are disentangled because of 
$\ket{A_l\oplus A_{l_j}}=\ket{\overline{A_l}\oplus \overline{A_{l_j}}}$.
Thus, $\bfR_0$'s
is in the state of $(\ket{x}+\ket{\overline{x}})/\sqrt{2}$.
  \begin{eqnarray}
\lefteqn{
\frac{1}{\sqrt{2}}
\bigotimes_{l=1}^{n}
\left(\ket{A_l}
\bigotimes_{j=1}^{d_l}\ket{A_{l_j}}\right)+
\frac{1}{\sqrt{2}}
\bigotimes_{l=1}^{n}
\left(\ket{\overline{A_l}}
\bigotimes_{j=1}^{d_l}\ket{\overline{A_{l_j}}}\right)
}\label{fo:Q:step5}\\
&\rightarrow& 
\frac{1}{\sqrt{2}}
\bigotimes_{l=1}^{n}
\left(\ket{A_l}
\bigotimes_{j=1}^{d_l}\ket{A_l\oplus A_{l_j}}\right)+
\frac{1}{\sqrt{2}}
\bigotimes_{l=1}^{n}
\left(\ket{\overline{A_l}}
\bigotimes_{j=1}^{d_l}\ket{\overline{A_l}\oplus \overline{A_{l_j}}}\right)\label{fo:Q:step6}
  \end{eqnarray}
\end{proof}
\begin{lemma}
\label{lm:subroutineQ-complexity}
  Let $\abs{E}$ and $D$ be the number of edges and the maximum degree,
  respectively, of the
  underlying graph of an $n$-party
  distributed system. Subroutine~Q takes $O(D)$ time for
  each party, and incurs one round with $2|E|$-qubit communication.
\end{lemma}
\begin{proof}
  Each party $l$ performs just one-round communication of $d_l$
  qubits. The local computations can be done in time linear in $d_l$.
\end{proof}

\subsubsection{Subroutine~\texttilde{A}:}

Suppose that, after Subroutine~Q,
$n$-qubit state ${\ket{\phi}=(\ket{x}+\ket{\overline{x}})/\sqrt{2}}$
is shared by the $n$ parties such that the $l$th party has
the $l$th qubit.
Let $x_l$ be the $l$th bit of $x$, 
and let $X$ and $\overline{X}$ be mappings
defined by ${X(v) = x_l}$ and ${\overline{X}(v) = \overline{x_l}}$
for each $l$, respectively,
where $v\in V$ represents the node corresponding to the $l$th party  in
the underlying graph $G=(V,E)$ of the network topology.
For any $v$ in $V$, let $W[v]:V\rightarrow \{ 0,1\}\times \{
\mbox{``$\eligible$''}, \mbox{``$\ineligible$''}\}$ be the mapping
defined
as $(Y[v],Z)$, where
$Y[v]$ is $X$ if $X(v)=0$ and $\overline{X}$ otherwise,
and $Z:V\rightarrow \{
\mbox{``$\eligible$''}, \mbox{``$\ineligible$''}\}$ maps
$v\in V$ to the value of $\status$ possessed by the party
corresponding to $v$. 
We denote $(\overline{Y[v]},Z)$ by $\overline{W[v]}$,
where $\overline{Y[v]}=\overline{X}$ if $Y[v]=X$
and $Y[v]=X$ otherwise.

Subroutine~\texttilde{A} 
checks the consistency of $\ket{\phi}$, but in  quite a different way
from Subroutine~A.
Every party $l$ constructs the folded view
$\widetilde{T}^{n-1}_{W[v]}(v)$
by using the output $y$ of Subroutine~Q.
The folded view is constructed
by the f-view construction algorithm in~Figure~\ref{fig:f-view
  construction} in subsection~\ref{subsec:minimal-folded-view}
with slight modification;
the modification is required since mapping $W[v]$ is not necessarily 
common over all parties $v$.
The construction still involves only classical communication.
By checking if the nodes for eligible
parties in the folded view have the same labels, 
Subroutine~\texttilde{A} can decide 
whether $\ket{\phi}$
is
consistent or not over the set of the indices of eligible parties.
Figure~\ref{Figure: Subroutine tildeA} gives a precise
description of Subroutine~\texttilde{A}. The next lemmas present the
correctness and complexity of Subroutine~\texttilde{A}.

\begin{figure}[t]
\begin{center}
\hrulefill\\
  \textbf{Subroutine~\texttilde{A}}
\vspace{-2mm}
\begin{description}
\setlength\itemsep{-3pt}
\item[Input:]
a classical variable $\status\in \{ \mbox{``$\eligible$''}, \mbox{``$\ineligible$''}\}$,
integers ${n, d}$,
a binary string $y$ of length $d$
\item[Output:]
a classical variable $\consistency\in \{ \consistent, \inconsistent\}$
\end{description}
\vspace{-6mm}
\begin{enumerate}
\setlength\itemsep{-3pt}
  \item
    Set $\widetilde{T}^0_{W[v]}(v)$ to a node labeled with $(0,
    \status)$, where 
$W[v]:V\rightarrow \{ 0,1\}\times \{
\mbox{``$\eligible$''}, \mbox{``$\ineligible$''}\}$ be the mapping
defined
as $(Y[v],Z)$, 
$Y[v]$ is $X$ if $X(v)=0$ and $\overline{X}$ otherwise,
and  $Z$
is the underlying mapping
naturally induced by the values of $\status$.
  \item
    For ${i:=1}$ to $(n-1)$, do the following:
\vspace{-6pt}
    \begin{itemize}
\setlength\itemsep{0pt}
      \item[2.1]
        Send $\widetilde{T}^{i-1}_{W[v]}(v)$ and 
receive $\widetilde{T}^{i-1}_{W[v_j]}(v_j)$ via port $j$,
        for ${1 \leq j \leq d}$,
        where node $v_j$ corresponds to the party connected via port $j$.
      \item[2.2]
        If the $j$th bit $y_j$ of $y$ is $1$, transform 
$\widetilde{T}^{i-1}_{W[v_j]}(v_j)$ into 
$\widetilde{T}^{i-1}_{\overline{W[v_j]}}(v_j)$ by
        negating the first element of every node label
        for ${1 \leq j \leq d}$, where $\overline{W[v_j]}$ represents $(\overline{Y[v_j]},Z)$.
      \item[2.3]
        Set the root of $\widetilde{T}^i_{W[v_j]}(v)$ to the node labeled with $(0, \status)$.\\
        Set the $j$th child of the root of $\widetilde{T}^i_{W[v]}(v)$
        to $\widetilde{T}^{i-1}_{\overline{W[v_j]}}(v_j)$, for ${1
          \leq j \leq d}$ such that $y_j=1$.\\
        Set the $j$th child of the root of $\widetilde{T}^i_{W[v]}(v)$
        to $\widetilde{T}^{i-1}_{W[v_j]}(v_j)$, for ${1
          \leq j \leq d}$ such that $y_j=0$.\\
        For every level of $\widetilde{T}^i_{W[v]}(v)$,
        merge nodes at that level into one node
        if the views rooted at them are isomorphic.
    \end{itemize}
\vspace{-6pt}
  \item
    If both label $(0, \mbox{``$\eligible$''})$
    and label $(1, \mbox{``$\eligible$''})$
    are found among the node labels in $\widetilde{T}^{n-1}_{W[v]}(v)$,
    let ${\consistency := \mbox{``$\inconsistent$''}}$;
    otherwise let ${\consistency := \mbox{``$\consistent$.''}}$
  \item
    Output $\consistency$.
\end{enumerate}
\vspace{-1\baselineskip}
\hrulefill
\end{center}
\vspace{-1\baselineskip}
\caption{Subroutine~\texttilde{A}.}  
\label{Figure: Subroutine tildeA}
\end{figure}

\begin{lemma}
\label{lm:subroutinetildeA-correctness}
Suppose that the $n$ parties share $n$-qubit cat-like
state $(\ket{x}+\ket{\overline{x}})/\sqrt{2}$, 
where $x$ is $n$-bit string $X(v_1)X(v_2)\cdots X(v_n)$ 
for $v_i\in V$
and $\overline{x}$ is the bitwise negation of $x$.
Let $S$ be the set of the indices of the parties among the $n$ parties
whose variable $\status$ is ``$\eligible$,''
and let $v\in V$ be the corresponding node of party $l$.
If every party $l$ runs Subroutine~\texttilde{A} with the following objects as input:
\begin{itemize}
\item a classical variable $\status\in \{ \mbox{``$\eligible$''}, \mbox{``$\ineligible$''}\}$,
\item $n$ and the number $d_l$ of the neighbors of party $l$,
\item a binary string $y=y_1\cdots y_{d_l}$ of length $d_l$ such that
$y_j=X(v)\oplus X(v_j)$ for $j=1,\dots, d_l$ where $v_j$
is the $j$th adjacent node of $v$,
\end{itemize}
Subroutine~\texttilde{A} outputs classical valuable $\consistency$, which has
value ``$\consistent$'' if
$(\ket{x}+\ket{\overline{x}})/\sqrt{2}$ is consistent over
$S$, and ``$\inconsistent$'' otherwise.
\end{lemma}
\begin{proof}
It will be proved later that steps 1 and 2 construct an f-view of depth $(n-1)$ 
for mapping either $(X,Z)$ or
$(\overline{X},Z)$.
Since the f-view is made by merging those nodes
at the same depth which are the roots of isomorphic views,
the f-view contains at least one node that 
has the same label as $(X(v),Z(v))$ or $(\overline{X}(v),Z(v))$
for any $v\in V$.
Once the f-view is constructed, every party can know whether
$X$ is constant over all $l\in S$ or not in step 3
by checking the labels including ``$\eligible$.''
Notice that no party needs to know for which mapping of 
$(X,Z)$ or $(\overline{X},Z)$ it has constructed the f-view.

In what follows, we prove that steps 1 and 2 construct 
an f-view for mapping either $X$ or $\overline{X}$.
The proof is by induction on depth $i$ of the f-view.
Clearly, step 1 generates $\widetilde{T}^0_{W[v]}(v)$.
Assume that every party  $l'$ has constructed
$\widetilde{T}^{i-1}_{W[v']}(v')$ where node $v'$ represents party $l'$.
In order to construct $\widetilde{T}^{i}_{W[v]}(v)$,
party $l$ needs $\widetilde{T}^{i-1}_{W[v]}(v_j)$ for every node $v_j$ adjacent
to $v$.
Although $W[v]$ is not always identical to $W[v_j]$,
we can transform $\widetilde{T}^{i-1}_{W[v_j]}(v_j)$ to 
$\widetilde{T}^{i-1}_{W[v]}(v_j)$.
Since 
$y_j$ is equal to $X(v)\oplus X(v_j)=\overline{X}(v)\oplus \overline{X}(v_j)$,
each of $X$ and $\overline{X}$ gives the same value for $v$ and $v_j$ if and
only if $y_j=0$.
This fact, together with
$Y[v](v)=Y[v_j](v_j)=0$, implies that
$Y[v]$ is identical to $Y[v_j]$ if and only if $y_j=0$.
It follows that, if $y_j=0$, 
$\widetilde{T}^{i-1}_{W[v_j]}(v_j)$ is isomorphic to
$\widetilde{T}^{i-1}_{W[v]}(v_j)$, and otherwise 
$\widetilde{T}^{i-1}_{\overline{W[v_j]}}(v_j)$ is isomorphic to
$\widetilde{T}^{i-1}_{W[v]}(v_j)$.
In the latter case, the party corresponding to $v$
negates the first
elements of all node labels in $\widetilde{T}^{i-1}_{W[v_j]}(v_j)$
to obtain $\widetilde{T}^{i-1}_{\overline{W[v_j]}}(v_j)$.
Thus step 3 can construct $\widetilde{T}^{i}_{W[v]}(v)$.
This completes the proof.
\end{proof}
\begin{lemma}
  Let $\abs{E}$ and $D$ be the number of edges and the maximum degree,
  respectively, of the underlying graph of an $n$-party
  distributed system. Subroutine~\texttilde{A} takes $O(D^2n^3(\log n)^2)$ time for
  each party, and incurs $O(n)$ rounds with classical communication of $O(D|E|n^3\log
  D)$ bits.
\label{lm:subroutinetildeA-complexity}
\end{lemma}
\begin{proof}
Steps 1 and 2 are basically the f-view construction algorithm in
Figure~\ref{fig:f-view construction} in subsection~\ref{subsec:minimal-folded-view} except step 2.2;  this step takes
$O(Dn^2)$ time
since an f-view of depth $O(n)$ has 
$O(Dn^2)$ edges. 
Thus, steps 1 and 2 take
$O(D^2n^3(\log n)^2)$ time for
  each party, incur $O(n)$ rounds and exchange $O(D|E|n^3\log D)$ bits by
Theorem~\ref{th:f-view}.
Step 3 takes $O(Dn^2)$ time.
\end{proof}

\subsubsection{Subroutine~\texttilde{B}:}
Suppose that ${\ket{\phi}=(\ket{x}+\ket{\overline{x}})/\sqrt{2}}$
shared by the $n$ parties
is consistent over the set $S$ of the indices of eligible parties.
Subroutine~\texttilde{B} has the goal of transforming $\ket{\phi}$ into an
inconsistent state over $S$. Let $k$ be $\abs{S}$.
First every ineligible party measures its qubit in the $\{
\ket{+},\ket{-}\}$ basis, where $\ket{+}$ and $\ket{-}$ denote
$(\ket{0}+\ket{1})/\sqrt{2}$ and $(\ket{0}-\ket{1})/\sqrt{2}$, respectively.
As a result,
the state shared by the eligible parties
becomes either
${\pm (\ket{0}^{\otimes k} + \ket{1}^{\otimes k})/\sqrt{2}}$
or
${\pm (\ket{0}^{\otimes k} - \ket{1}^{\otimes k})/\sqrt{2}}$.
The state 
${\pm (\ket{0}^{\otimes k} - \ket{1}^{\otimes k})/\sqrt{2}}$
is shared
if and only if
the number of ineligible parties
that measured $\ket{-}$ is odd, as will be proved in Lemma~\ref{lm:n-cat to k-cat}.
In this case,
every eligible party
applies unitary operator $W_k$ to its qubit
so that the shared state is transformed into 
${\pm (\ket{0}^{\otimes k} + \ket{1}^{\otimes k})/\sqrt{2}}$,
where the family $\{ W_k \}$ of unitary operators is defined by
\[
  W_k
  =
  \begin{pmatrix}
  1 & 0\\
  0 & e^{i \frac{\pi}{k}}
  \end{pmatrix}.
\]

Again let $v$ denote the node corresponding to
the party that invokes the subroutine.
Figure~\ref{Figure: Subroutine tildeB} gives a precise description of Subroutine~\texttilde{B}.
The correctness and complexity of the subroutine will be described in
Lemmas~\ref{lm:subroutinetildeB-correctness}~and~\ref{lm:subroutinetildeB-complexity},
respectively.
\begin{figure}[t]
\begin{center}
\hrulefill\\
  \textbf{Subroutine~\texttilde{B}}
\vspace{-2mm}
\begin{description}
\setlength\itemsep{-3pt}
\item[Input:]
one-qubit quantum registers ${\bfR _0, \bfR _1}$,
a classical variable $\status\in \{ \mbox{``$\eligible$''}, \mbox{``$\ineligible$''}\}$,
integers ${k, n, d}$
\item[Output:]
one-qubit quantum registers ${\bfR _0, \bfR _1}$
\end{description}
\vspace{-6mm}
\begin{enumerate}
\setlength\itemsep{-3pt}
  \item
    Let ${w:=0}$.
  \item
    If ${\status = \mbox{``$\ineligible$,''}}$
    measure the qubit in $\bfR_0$
    in the $\{ \ket{+}, \ket{-} \}$ basis.\\
    If this results in $\ket{-}$, let ${w:=1}$.
  \item
    Construct f-view $\widetilde{T}^{(2n-1)}_{W}(v)$
    to count the number $p$ of parties with ${w=1}$,
    where $W$ is the underlying mapping
    naturally induced by the $w$ values of all parties.
  \item
    If $p$ is odd and ${\status = \mbox{``$\eligible$,''}}$
    apply $W_k$ to the qubit in $\bfR_0$.
  \item If ${\status = \mbox{``$\eligible$,''}}$ perform Subroutine~B with $\bfR _0$, $\bfR _1$ and $k$.
  \item
    Output quantum registers $\bfR_0$ and $\bfR _1$.
\end{enumerate}
\vspace{-1\baselineskip}
\hrulefill
\end{center}
\vspace{-1\baselineskip}
\caption{Subroutine~\texttilde{B}.}  
\label{Figure: Subroutine tildeB}
\end{figure}

\begin{lemma}
\label{lm:subroutinetildeB-correctness}
Suppose that the $n$ parties share $n$-qubit cat-like
state $\ket{\phi}=(\ket{x}+\ket{\overline{x}})/{\sqrt{2}}$, where $x$
is any $n$-bit string that is consistent over $S$, 
and $\overline{x}$ is the bitwise negation of $x$.
If each party $l$ runs Subroutine~\texttilde{B} with the following objects as input:
\begin{itemize}
  \item one-qubit register $\bfR_0$, which stores one of the $n$ qubits in
  state $\ket{\phi}$,
  \item one-qubit register $\bfR_1$, which is initialized to $\ket{0}$,
  \item a classical variable $\status$, the value of which is 
  ``$\eligible$'' if $l$ is in $S$ and ``$\ineligible$''
  otherwise,
  \item integers $k$, $n$, and the number $d_l$ of neighbors of party $l$,
\end{itemize}
Subroutine~\texttilde{B} outputs two one-qubit registers ${\bfR _0, \bfR_1}$
such that, if given $k$ is equal to $\abs{S}$, the qubits in the registers
satisfy the conditions:
\begin{itemize}
\item the $2k$ qubits possessed by all parties $l'$ for $l'\in S$ are in an inconsistent state
over $S$,
\item the $2(n-k)$ qubits possessed by all parties $l'$ for $l'\not\in S$ are in a classical
  state (as a result of measurement).
\end{itemize}
\end{lemma}
\begin{proof}
Lemma~\ref{lm:n-cat to k-cat} guarantees that, after step 2,
the eligible parties (i.e., the parties who have $\status=\mbox{``$\eligible$''}$) share 
$(\ket{0}^{\otimes k}+\ket{1}^{\otimes k})/{\sqrt{2}}$
($(\ket{0}^{\otimes k}-\ket{1}^{\otimes k})/{\sqrt{2}}$)
if the number of those parties who  have measured
$\ket{-}$
is even (respectively, odd).
When the eligible parties share 
$(\ket{0}^{\otimes k}-\ket{1}^{\otimes k})/{\sqrt{2}}$,
step 4 transforms the shared state
into 
$(\ket{0}^{\otimes k}+\ket{1}^{\otimes k})/{\sqrt{2}}$.
Due to Lemma~\ref{lm:subroutineB-correctness}, the eligible parties share an
inconsistent state over $S$ after step 5. This completes the proof.
\end{proof}

\begin{lemma}
\label{lm:subroutinetildeB-complexity}
  Let $\abs{E}$ and $D$ be the number of edges of the underlying graph of an $n$-party
  distributed system. Subroutine~\texttilde{B} takes $O(Dn^5\log n)$ time for
  each party, takes $O(n)$ rounds and requires $O(D|E|n^3\log D)$-bit
communication.
\end{lemma}
\begin{proof}
  Since Subroutine~B takes $O(1)$ time and does no communication, 
step 3 is dominant. The proof is completed by
Theorem~\ref{th:f-view}.
\end{proof}

\begin{lemma}
\label{lm:n-cat to k-cat}
Let $S$ be an arbitrary subset of $\{ 1,2,\ldots n\}$ parties such that $|S|=k$.
Suppose that $n$ parties share $n$-qubit cat-like
state $(\ket{x}+\ket{\overline{x}})/\sqrt{2}$, where $x$ is any
$n$-bit string that is consistent over $S$, 
and $\overline{x}$ is the bitwise negation of $x$.
If every party $l$ for $l\not\in S$ measures his qubit with respect to the Hadamard
basis 
$\{\ket{+},\ket{-}\}$,
the resulting state is
$(\ket{0}^{\otimes k} +\ket{1}^{\otimes k})/\sqrt{2}$
($(\ket{0}^{\otimes k}-\ket{1}^{\otimes k})/\sqrt{2}$)
when the number of those parties 
is even (respectively, odd) who have measured
$\ket{-}$.
\end{lemma}
\begin{proof}
From the next two claims (a) and (b), and the induction on the number of parties,
the lemma follows.
  \begin{itemize}
  \item [(a)] If $m$ parties share $m$-qubit 
state $(\ket{z_1\dots z_m}+\ket{\overline{z}_1\dots
  \overline{z}_m})/\sqrt{2}$ for any $z_i\in \{0,1\}$ $(i=1,\dots ,m)$
($\overline{z}_i$ is the negation of $z_i$)
and the last party measures his qubit with respect to the Hadamard
basis, then
the resulting state is
$(\ket{z_1\dots z_{m-1}}+\ket{\overline{z}_1\dots
  \overline{z}_{m-1}})/\sqrt{2}$
($(\ket{z_1\dots z_{m-1}}-\ket{\overline{z}_1\dots
  \overline{z}_{m-1}})/\sqrt{2}$)
when he measured $\ket{+}$ ($\ket{-}$).
  \item [(b)] If $m$ parties share $m$-qubit 
state $(\ket{z_1\dots z_m}-\ket{\overline{z}_1\dots
  \overline{z}_m})/\sqrt{2}$ for any $z_i\in \{0,1\}$ $(i=1,\dots ,m)$,
and the last party measures his qubit with respect to the Hadamard
basis, then
the resulting state is
$(\ket{z_1\dots z_{m-1}}-\ket{\overline{z}_1\dots
  \overline{z}_{m-1}})/\sqrt{2}$
($(\ket{z_1\dots z_{m-1}}+\ket{\overline{z}_1\dots
  \overline{z}_{m-1}})/\sqrt{2}$) up to global phases
when he measured $\ket{+}$ ($\ket{-}$).
  \end{itemize}

We first prove claim (a). By simple calculation, we have
\[
\frac{\ket{z_1\dots z_m}+\ket{\overline{z}_1\dots
  \overline{z}_m}}{\sqrt{2}}
=\frac{\ket{z_1\dots z_{m-1}}+\ket{\overline{z}_1\dots
  \overline{z}_{m-1}}}{\sqrt{2}}\ket{+} + 
 \frac{\ket{z_1\dots z_{m-1}}-\ket{\overline{z}_1\dots
  \overline{z}_{m-1}}}{\sqrt{2}}\ket{-}.
\]
Thus, claim (a) follows.

Similarly, claim (b) is proved by the next equation:
\[
\frac{\ket{z_1\dots z_m}-\ket{\overline{z}_1\dots
  \overline{z}_m}}{\sqrt{2}}
=\frac{\ket{z_1\dots z_{m-1}}-\ket{\overline{z}_1\dots
  \overline{z}_{m-1}}}{\sqrt{2}}\ket{+} +(-1)^{z_m} 
 \frac{\ket{z_1\dots z_{m-1}}+\ket{\overline{z}_1\dots
  \overline{z}_{m-1}}}{\sqrt{2}}\ket{-}.
\]
\end{proof}

\subsubsection{Subroutine~\texttilde{C}:}
Suppose that each party $l$ has value $z_l$.
Subroutine~C is a classical algorithm
that computes value $z_{\minor}$
such that the number of parties with value $z_{\minor}$
is non-zero and the smallest among all possible non-negative $z_l$ values.
It is stressed that 
the number of parties with value $z_{\minor}$
is at most half the number of parties having non-negative $z_l$
values,
and that the parties having non-negative $z_l$ values are eligible 
from the construction of Algorithm~II.
Figure~\ref{Figure: Subroutine tildeC} gives a precise description of Subroutine~\texttilde{C}.

\begin{figure}[t]
\begin{center}
\hrulefill\\
  \textbf{Subroutine~\texttilde{C}}
\vspace{-2mm}
\begin{description}
\setlength\itemsep{-3pt}
\item[Input:]
integers $z\in \{ -1,0,1,2,3 \}$, $n$, $d$
\item[Output:]
integers $z_{\minor}$, $c_{z_{\minor}}$
\end{description}
\vspace{-6mm}
\begin{enumerate}
\setlength\itemsep{-3pt}
  \item
    Construct f-view $\widetilde{T}^{(2n-1)}_{Z}(v)$,
    where $Z$ is the underlying mapping
    naturally induced by the $z$ values of all parties.
  \item
    For ${i:=0}$ to ${3}$,
    count the number, $c_i$, of parties having a value ${z=i}$
    using $\widetilde{T}^{(2n-1)}_{Z}(v)$.\\
    If ${c_i = 0}$, let ${c_i := n}$.
  \item
  Let $z_{\minor}\in \{ m \mid c_m=\min_{0 \leq i \leq 3} c_i\}$.
  \item
    Output $z_{\minor}$ and $c_{z_{\minor}}$.
\end{enumerate}
\vspace{-1\baselineskip}
\hrulefill
\end{center}
\vspace{-1\baselineskip}
\caption{Subroutine~\texttilde{C}}
\label{Figure: Subroutine tildeC}
\end{figure}

The next two lemmas give the correctness and complexity of Subroutine~\texttilde{C}.
\begin{lemma}
  \label{lm:subroutinetildeC-correctness}
Suppose that each party $l$ among $n$ parties has an integer $z_l\in
\{ -1,0,1,2,3\}$.
If every party $l$ runs Subroutine~\texttilde{C} with $z_l$, $n$ and the number $d_l$ of
neighbors as input, Subroutine~\texttilde{C} outputs 
$z_{\minor}\in \{ z_1,\ldots ,z_n\}\setminus \{ -1\}$, and $c_{z_{\minor}}$ such that
the number $c_{z_{\minor}}$ of parties having $z_{\minor}$ is
not more than that of parties having any other $z_l$
(ties are broken arbitrarily). 
\end{lemma}
\begin{proof}
The first line of step 2 in Figure~\ref{Figure: Subroutine tildeC} counts the number $c_i$ of parties having $i$
as $z$ for each $i\in \{0,1,2,3\}$ by using f-view.
Since $c_i=0$ implies $z_{\minor}\neq i$, $c_i$ is set to $n$ so that
$i$ cannot be selected as $z_{\minor}$ in step 3. 
Thus, $z_{\minor}$ is selected among $\{ z_1,\ldots ,z_n\}\setminus \{ -1\}$.
\end{proof}
\begin{lemma}
  \label{lm:subroutinetildeC-complexity}
  Let $\abs{E}$ and $D$ be the number of edges and the maximum degree of the underlying graph of an $n$-party
  distributed system. Subroutine~\texttilde{C} takes $O(Dn^5\log n)$ time for
  each party, takes $O(n)$ rounds, and requires $O(D|E|n^3\log D)$-bit
communication.
\end{lemma}
\begin{proof}
Steps 1 and 2 are dominant. The proof is completed by Theorem~\ref{th:f-view}.
\end{proof}

\subsection{Complexity analysis}

Now we prove Theorem~\ref{Theorem: Complexity of Algorithm II for n}.
\setcounter{tmptheorem}{\value{theorem}}
\setcounter{theorem}{2}
\addtocounter{theorem}{-1}
\begin{theorem}
Let $\abs{E}$ and $D$ be the number of edges and the maximum degree of the
underlying graph, respectively.
Given the number $n$ of parties,
Algorithm~II exactly elects a unique leader
in $O(Dn^5(\log n)^2)$ time and $O(n\log n)$ rounds
of which only the first round requires quantum communication.
The total communication complexity over all parties
is $O(D\abs{E} n^3(\log D)\log n)$
which includes the communication of only $O(\abs{E} \log n)$ qubits.
\label{Theorem: Complexity of Algorithm II for n}
\end{theorem}
\setcounter{theorem}{\value{tmptheorem}}
\begin{proof}
  Lemma~\ref{lm:subroutineQ-correctness} guarantees that Stage~1 works
  correctly. We will prove that steps 1 to 5 of Stage~2 decrease the number of
  eligible parties by at least half,
without eliminating  all eligible parties,
if there are at least two eligible parties.
This directly leads to the correctness of Algorithm~II, since 
$s:=\lceil \log n \rceil$.

The proof is by induction on phase number $i$.
At the beginning of the first phase,
$k$ obviously represents  the number of eligible parties.
Next we prove that if $k$ is equal to the number of eligible parties
immediately before entering phase $i$,
steps 1 to 5 decrease the number of
eligible parties by at least half
without eliminating all such parties, and set $k$ to the updated number
of the eligible parties.
By Lemmas~\ref{lm:subroutinetildeA-correctness} and
\ref{lm:subroutinetildeB-correctness} and the assumption that 
$k$ is the number of eligible parties, 
only eligible parties share an
inconsistent state with certainty after steps 1 and 2. Thus, it is impossible that all
eligible parties get the same value by measurement at step 3.
Subroutine~\texttilde{C} correctly computes $z_{\rm minor}$ and the number $c_{z_{\rm
    minor}}$ as proved in 
Lemma~\ref{lm:subroutinetildeC-correctness}.
Hence, step 4 reduces eligible parties by at least half with certainty
and sets $k$ to the updated number of the eligible parties.

Next, we analyze the complexity of Algorithm~II.
By Lemma~\ref{lm:subroutineQ-complexity}, Stage 1 takes
$O(D)$ time for
each party and one round, and requires $O(\abs{E}\log n)$-qubit communication.
Stage 2 iterates Subroutines~\texttilde{A},  \texttilde{B} and \texttilde{C} at most $O(\log n)$
times.
By Lemmas~\ref{lm:subroutinetildeA-complexity},
\ref{lm:subroutinetildeB-complexity} and \ref{lm:subroutinetildeC-complexity},
Subroutines~\texttilde{A},  \texttilde{B} and \texttilde{C}
take $O(n)$ rounds, $O(Dn^5\log n)$ time and require $O(D\abs{E}n^3\log D)$
classical bit communication for each iteration.
Hence, Stage 2 takes $O(n\log n)$ rounds and $O(Dn^5(\log n)^2)$ time,
and
requires $O(D\abs{E}n^3(\log D)\log n)$ classical bit communication.
This completes the proof.
\end{proof}

\subsection{Generalization of the algorithm}
\label{subsec:algorithm II when only the upperbound is given}
In the case where only the upper bound $N$ of the number of parties is given, 
we cannot apply Algorithm~II as it is, since Algorithm~II strongly depends
on counting the exact number of eligible parties and this
requires the exact number of parties.

We modify Algorithm~II so that it outputs $\status=\mbox{``$\error$''}$
and halts (1) if steps 1 to 5 of
Stage 2 are iterated over $\log n$ times, or (2) if it is found that
non-integer values are being stored into the variables whose values should be
integers. 
Notice that we can easily see that this modified Algorithm~II can run (though it may
halt with output $\mbox{``$\error$''}$) even when it is given the wrong number of
parties as input, unless the above condition (2) becomes true during execution.
Let the modified Algorithm~II be $\LE(\status,n,d)$.

The basic idea is to run $\LE(\mbox{``$\eligible$''},m,d)\ (2\leq
m\leq N)$ in parallel.
Here we assume that every party has one processor,
and all local computations are performed sequentially.
Message passing, on the other hand, is done in parallel,
i.e., at each round the messages of 
$\LE(\mbox{``$\eligible$''},m,d)\ (2\leq
m\leq N)$ are packed into one message and sent to adjacent parties.
Although this parallelism cannot reduce time/communication complexity,
it can reduce the  number of rounds needed.
Let $M$ be the largest $m\in \{ 2,3,\ldots ,N\}$ such that $\LE(\mbox{``$\eligible$''},m,d)$ 
terminates with output $\mbox{``$\eligible$''}$ or $\mbox{``$\ineligible.$''}$ 
The next lemma implies that
$M$ is equal to the hidden number of parties, i.e., $n$,  and thus
$\LE(\mbox{``$\eligible$''},M,d)$ elects the unique leader.
Figure~\ref{Figure:GeneralizedAlgorithm II} 
describes this generalized algorithm.
We call it the Generalized Algorithm II.

\begin{figure}[t]
\begin{center}
\hrulefill\\
  \textbf{Generalized Algorithm~II}
\vspace{-2mm}
\begin{description}
\setlength\itemsep{-3pt}
\item[Input:]
    a classical variable $\status\in \{ \mbox{``$\eligible$''}, \mbox{``$\ineligible$''}\}$,
    integers ${N,d}$
  \item[Output:]
    a classical variable $\status\in \{ \mbox{``$\eligible$''}, \mbox{``$\ineligible$''}\}$
\end{description}
\vspace{-6mm}
\begin{enumerate}
\setlength\itemsep{-3pt}
\item Run in parallel $\LE(\mbox{``$\eligible$''},m,d)$ for $m=2,3,\ldots
  ,N$.\\
\item Output $\status$ returned by $\LE(\mbox{``$\eligible$''},M,d)$, where
$M$ is the largest value of $m\in \{ 2,3,\ldots
  ,N\}$ such that $\LE(\mbox{``$\eligible$''},m,d)$ 
outputs $\status \in \{ \mbox{``$\eligible$''}, \mbox{``$\ineligible$''}\}$.
\end{enumerate}
\vspace{-1\baselineskip}
\hrulefill
\end{center}
\vspace{-1\baselineskip}
\caption{Generalized Algorithm~II.}  
\label{Figure:GeneralizedAlgorithm II}
\end{figure}

\begin{lemma}
  For any number $m$ larger than the number $n$ of parties, 
if every party $l$ runs $\LE(\mbox{``$\eligible$''},m,d_l)$,
it always outputs ``$\error$,''
where $d_l$ is the number of neighbors of party $l$.
\label{lm:modifiedAlgorithmII}
\end{lemma}
\begin{proof}
  It is sufficient to prove that $k$ is never equal to 1
in the modified Algorithm II,
since step 5 of Stage 2 outputs $\status$ only when $k=1$.
$k$ is set to $c_{z_{\minor}}$ at step 4 of Stage 2 and 
$c_{z_{\minor}}$ is computed at step 3 of Subroutine~\texttilde{C}.
We prove that, for any $i$, $c_i> 1$ at step 2 of 
Subroutine~\texttilde{C}, from which
the lemma follows.
Subroutine~\texttilde{C} computes
$$c_i=m\frac{|\Gamma ^{2(m-1)}(i)|}{|\Gamma ^{2(m-1)}|}.$$
If $i$ is in $\{0,1,2,3\}\setminus \{z_1,\ldots, z_n\}$ where $z_l$ is
the $z$ value of party $l$, $c_i$ is set to $n$; otherwise 
$$c_i\geq m/n>1,$$
since $|\Gamma ^{2(m-1)}(i)|\geq 1$ and 
$|\Gamma ^{2(m-1)}|\leq n$.
\end{proof}

\setcounter{tmptheorem}{\value{theorem}}
\setcounter{theorem}{4}
\addtocounter{theorem}{-1}
Now we prove a corollary of Theorem~\ref{Theorem: Complexity of
  Algorithm II for n}.
\begin{corollary}
Let $\abs{E}$ and $D$ be the number of edges and the maximum degree of the
underlying graph, respectively.
Given the upper bound $N$ of the number of parties,
the Generalized Algorithm~II exactly elects a unique leader
in $O(DN^6(\log N)^2)$ time and $O(N\log N)$ rounds
of which only the first round requires quantum communication.
The total communication complexity over all parties
is $O(D\abs{E} N^4(\log D)\log N)$
which includes the communication of only $O(\abs{E} N\log N)$ qubits.
\label{corollary: Complexity of Algorithm II for N}
\end{corollary}
\setcounter{theorem}{\value{tmptheorem}}
\begin{proof}
  From Lemma~\ref{lm:modifiedAlgorithmII}, the correctness is obvious.
  As for the complexity, we can obtain the complexity of the modified
  Algorithm~II, i.e., $\LE(\mbox{``$\eligible$''},m,d)$, simply by
  replacing $n$ with $m$ in the complexity of the (original) Algorithm
  II.  Since the Generalized Algorithm~II runs
  $\LE(\mbox{``$\eligible$''},m,d)$ for $m=2,\ldots ,N$ in parallel,
  the number of rounds required is the same as the maximum of that of
  the modified Algorithm~II over $m=2,\ldots N$.  The
  time/communication complexity is $O(\sum _{m=2}^NC(m))=O(N\cdot
  C(N))$, where $C(m)$ is that of 
Algorithm~II stated in Theorem~\ref{Theorem: Complexity of Algorithm
  II for n} for an $m$-party case.
\end{proof}

In fact, it is possible to reduce the time/communication complexity
at the expense of the number of rounds.
Suppose that every party $l$ runs $\LE(\mbox{``$\eligible$''},m,d_l)$
sequentially in decreasing order of $m$ starting at $N$.
When $\LE(\mbox{``$\eligible$''},m,d_l)$ outputs $\status$ that is either
``$\eligible$'' or ``$\ineligible$,''
the algorithm halts. From Lemma~\ref{lm:modifiedAlgorithmII}, it is clear that
the algorithm halts when $m=n$, which saves the time and communication
that would be otherwise required by $\LE(\mbox{``$\eligible$''},m,d_l)$
for $m<n$.

\subsection{Modification for directed network topologies}
\label{subsec:algorithm II  for directed network}
Algorithm~II can be easily modified so that it can be applied to the
network topologies whose underlying graph is directed and strongly-connected.

We slightly modify the network model as follows.
In a quantum distributed system,
every party can perform quantum computation and communication,
and each pair of parties
has at most one \emph{uni-directional} quantum communication link 
in each direction
between them.
For each pair of parties, there is at least one directed path 
between them for each direction.
When the
parties and links are regarded as nodes and edges, respectively,
the topology of the distributed system is expressed by a 
strongly-connected graph, denoted by ${G=(V,E)}$. 
Every party has two kinds of \emph{ports}: \emph{in-ports} and \emph{out-ports};
they correspond one-to-one to incoming and outgoing communication
links, respectively, incident to the party.
Every port of party $l$ has a unique label $i$, $(1 \leq i \leq d_l)$,
where $d_l$ is the number of parties adjacent to $l$ and
$d_l=d_l^I+d_l^O$ for the number $d_l^I$ ($d_l^O$) of in-ports
(resp. out-ports) of party $l$.
For 
$G=(V,E)$,
the port numbering $\sigma$ is defined in the same way as in the case of
the undirected graph model.
Just for ease of explanation, we assume that
in-port $i$ of party $l$ corresponds to the incoming communication link
connected to the $i$th party 
among all adjacent parties that have an outgoing communication link destined to
party $l$; out-port $i$ of party $l$ is also interpreted in a similar way.

The view for the strongly-connected underlying graph can be naturally defined.
For each $v$ and port numbering $\sigma$,
\emph{view} ${T_{G,\sigma, X} (v)}$ is
a labeled, rooted tree with infinite depth defined recursively as follows:
(1) ${T_{G,\sigma, X} (v)}$ has the root $w$ with label $X(v)$, corresponding to $v$,
(2) for the source $v_j$ of every directed edge coming to $v$ in $G$,
${T_{G,\sigma,X} (v)}$ has vertex $w_j$ labeled with $X(v_j)$,
and an edge from root $w$ to $w_j$ with label ${\Label((v,v_j))}$
given by ${\Label((v,v_j))=(\sigma[v] (v,v_j),\sigma[v_j] (v,v_j))}$,
and (3) $w_j$ is the root of ${T_{G,\sigma,X} (v_j)}$.
${T^{h}_{X}(v)}$ is defined in the same way as in the case of
the undirected graph model.
The above definition also gives a way of constructing ${T^{h}_{G,\sigma,X}(v)}$.
It is stressed that every party corresponds to at least one node of
the view if the underlying graph is strongly-connected.
It can be proved 
in almost the same way as~\cite{YamKam96IEEETPDS-1,yamashita-kameda99}
that 
the equivalence classes 
with respect to the isomorphism of views
have the same cardinality for fixed $G=(V,E),\sigma$ and $X$;
$c_{G,\sigma, X}(S)$ can be computed from a view.
\begin{lemma}
For the distributed system whose underlying graph $G=(V,E)$ is strongly-connected,
the number of views isomorphic to view $T$ is constant over all $T$
for fixed $\sigma$ and $X$.
\end{lemma}
\begin{proof}
Let $T(v)$ and $T(v')$ be any two non-isomorphic views
for fixed $\sigma$ and $X$, and
let $\{T(v_1),\ldots, T(v_m) \}$ be the set of all views isomorphic to
$T(v)$, where $v\in \{v_i\mid i=1,\ldots ,m\}\subseteq V$.
Since the underlying graph is strongly-connected, 
there is a subtree of $T(v_1)$ which is isomorphic to 
$T(v')$. 
Let $s$ be the sequence of edge and node labels from the root of $T(v_1)$
to the root $u_1$ of the subtree.
Since all views in $\{T(v_1),\ldots, T(v_m) \}$ are
isomorphic
to one another, there is a node $u_i$ that can be reached from the root of 
$T(v_i)$ along $s$ for each $i=1,\ldots m$.
Clearly, every $u_i$ is the root of a tree 
isomorphic to 
$T(v')$. 
Since $v_i$ and $v_j$ correspond to different parties if $i\neq j$,
$u_i$ and $u_j$ also correspond to different parties if $i\neq j$.
This implies that the number of views 
isomorphic to $T(v')$ is not less than that of views 
isomorphic to $T(v)$.
By replacing $T(v)$ with $T(v')$,
we can see that the number of views 
isomorphic to $T(v)$ is not less than that of views 
isomorphic to $T(v')$.
This completes the proof.
\end{proof}

Since f-view is essentially a technique for compressing a tree structure
by sharing isomorphic subtrees, 
f-view also works for views of any strongly-connected underlying graph.

From the above, it is not difficult to see that Subroutines \texttilde{A}, 
{\texttilde{B}} and {\texttilde{C}} work well (with only slight modification),
since they use only classical communication.
In what follows, we describe a modification,
called Subroutine~Q',
 to Subroutine~Q.
With the subroutine, the correctness of the modification to
 Algorithm~II will be obvious
for any strongly-connected underlying graph.

Subroutine~Q' simply restricts Subroutine~Q so that
every party can send qubits only via out-ports
and receive qubits only via in-ports.
Figure~\ref{Figure: Subroutine Q'} gives a precise description of
Subroutine~Q';
the subroutine requires
two integers $d^I$ and $d^O$ together with
$\bfR_0$ as input, which are taken to be $d^I_l$ and $d^O_l$, respectively.
Thus, Algorithm~II needs to be slightly modified so that it can handle
$d^I$ and $d^O$ instead of $d$.

\begin{figure}[t]
\begin{center}
\hrulefill\\
  \textbf{Subroutine~Q'}
\begin{description}
\setlength\itemsep{-3pt}
\item[Input:]
a one-qubit quantum register ${\bfR}_0$,
integers $d^I$ and $d^O$
\item[Output:]
a one-qubit quantum register ${\bfR}_0$,
a binary string $y$ of length $d^I$
\end{description}
\vspace{-6mm}
\begin{enumerate}
\setlength\itemsep{-3pt}
  \item
    Prepare $d^O$ one-qubit quantum registers
    ${\bfR'_1, \ldots, \bfR'_{d^O}}$ and 
$2d^I$ one-qubit quantum registers ${\bfR''_1, \ldots, \bfR''_{d^I}}$,
${\bfS_1, \ldots, \bfS_{d^I}}$,
    each of which is initialized to the $\ket{0}$ state.
  \item
    Generate the $(d^O+1)$-cat state
    ${(\ket{0}^{\otimes (d^O+1)} + \ket{1}^{\otimes (d^O+1)})/\sqrt{2}}$
    in registers $\bfR_0$, ${\bfR'_1, \ldots, \bfR'_{d^O}}$.
  \item
     Send the qubit in $\bfR'_i$
     to the party connected via out-port $i$ for ${1 \leq i \leq d^O}$.\\
     Receive the qubit from the party connected via in-port $i$ and
     store it in one-qubit register $\bfR''_i$  for ${1 \leq i \leq d^I}$.
  \item
  Set the content of $\bfS_i$ to ${x_0 \oplus x_i}$, for ${1 \leq i \leq d^I}$,
    where $x_0$ and $x_i$ denote the contents of $\bfR_0$ and $\bfR''_i$, respectively.
  \item
    Measure the qubit in $\bfS_i$ in the $\{ \ket{0}, \ket{1} \}$ basis
    to obtain bit $y_i$, for ${1 \leq i \leq d^I}$.\\
    Set ${y:= y_1 \cdots y_{d^I}}$. 
  \item
    Apply CNOT controlled by the content of $\bfR_0$ and targeted to
    the content of each $\bfR''_i$ for ${i=1,2,\ldots ,d^I}$ to 
    disentagle $\bfR''_i$s.
  \item
    Output $\bfR_0$ and $y$.
\end{enumerate}
\vspace{-1\baselineskip}
\hrulefill
\end{center}
\vspace{-1\baselineskip}
\caption{Subroutine~Q'}  
\label{Figure: Subroutine Q'}
\end{figure}
We can prove the next lemma in a similar way to
Lemma~\ref{lm:subroutineQ-correctness}.
\begin{lemma}
\label{lm:subroutineQ'-correctness}
For an $n$-party distributed system,
suppose that every party $l$ calls Subroutine~Q' with a one-qubit register whose content is initialized to
$\ket{0}$ and  $d^I_l$ and $d^O_l$ as input $\bfR_0$,
$d^I$ and $d^O$, respectively.
After performing Subroutine~Q', all parties share
$(\ket{x}+\ket{\overline{x}})/{\sqrt{2}}$
with certainty, where $x$ is a random $n$-bit string.
\end{lemma}
\begin{proofsketch}
After step 2,
the state in $\bfR_0$'s and ${\bfR''_1\mbox{'s}, \ldots, \bfR''_{d^I}\mbox{'s}}$
of all parties is a
uniform superposition of some basis states in an orthonormal basis
of $2^{\sum _{l=1}^n(d^O_l+1)}$-dimensional Hilbert space:
the basis states
correspond one-to-one to $n$-bit integers $a$ and each of the basis states
is of the form $\ket{a_1}^{\otimes (d^O_1+1)}\otimes \cdots \otimes \ket{a_n}^{\otimes
  (d^O_n+1)}$, where $a_l$ is the $l$th bit of the binary expression of $a$.
If we focus on the $l$th party's part of the basis state corresponding
to $a$, step 3 transforms $\ket{a_l}^{\otimes (d^O_l+1)}$ into 
$\ket{a_l}\left(\bigotimes_{j=1}^{d^I_l}\ket{a_{l_j}}\right)$,
where we assume that party $l$ is connected to party $l_j$ via in-port
$j$.
Notice that the total number of qubits over all parties is preserved, since
$\sum _{l=1}^nd^I_l=\sum _{l=1}^nd^O_l$.
More precisely, steps 1 to 4 transform the system state as follows:
  \begin{eqnarray*}
    \bigotimes _{l=1}^{n}\ket{0}\ket{0}^{\otimes
    d^O_l}\ket{0}^{\otimes d^I_l}&\rightarrow&
    \bigotimes_{l=1}^{n}\frac{\ket{0}^{\otimes
    (d^O_l+1)}+\ket{1}^{\otimes (d^O_l+1)}}{\sqrt{2}}\ket{0}^{\otimes d^I_l}\\
&\rightarrow& 
\frac{1}{\sqrt{2^n}}
\sum _{a=0}^{2^n-1}\bigotimes_{l=1}^{n}\left\{\ket{a_l}
\left(\bigotimes_{j=1}^{d^I_l}\ket{a_{l_j}}\right)
\ket{0}^{\otimes d^I_l}\right\}\\
&\rightarrow& 
\frac{1}{\sqrt{2^n}}
\sum _{a=0}^{2^n-1}\bigotimes_{l=1}^{n}\left\{\ket{a_l}
\left(\bigotimes_{j=1}^{d^I_l}\ket{a_{l_j}}\right)
\left(\bigotimes_{j=1}^{d^I_l}\ket{a_l\oplus a_{l_j}}\right)\right\}.\\
  \end{eqnarray*}
After every party measures registers $\bfS_i$'s at step 5,
the state transformation can be described as follows, due to a similar
argument to Claim~\ref{claim: correctness of subroutine Q}
(using the strong connectivity of the underlying graph):
  \begin{eqnarray*}
\lefteqn{
\frac{1}{\sqrt{2}}
\bigotimes_{l=1}^{n}
\left(\ket{A_l}
\bigotimes_{j=1}^{d^I_l}\ket{A_{l_j}}\right)+
\frac{1}{\sqrt{2}}
\bigotimes_{l=1}^{n}
\left(\ket{\overline{A_l}}
\bigotimes_{j=1}^{d^I_l}\ket{\overline{A_{l_j}}}\right)
}\\
&\rightarrow& 
\frac{1}{\sqrt{2}}
\bigotimes_{l=1}^{n}
\left(\ket{A_l}
\bigotimes_{j=1}^{d^I_l}\ket{A_l\oplus A_{l_j}}\right)+
\frac{1}{\sqrt{2}}
\bigotimes_{l=1}^{n}
\left(\ket{\overline{A_l}}
\bigotimes_{j=1}^{d^I_l}\ket{\overline{A_l}\oplus \overline{A_{l_j}}}\right).
  \end{eqnarray*}
Finally, the qubits in $\bfR_0$'s
are in the state $(\ket{x}+\ket{\overline{x}})/{\sqrt{2}}$.
\end{proofsketch}

It is easy to see that the complexity has the same order as that
in the case of undirected networks: Theorem~\ref{Theorem: Complexity
  of Algorithm II for n} and Corollary~\ref{corollary: Complexity of
  Algorithm II for N} hold, if we define
$D$ as the maximum value of the sum of the numbers of in-ports and
out-ports over all parties.

\section{Folded view and its algorithms}
\label{sec: view compression}
View size is exponential against its depth since a view is a tree. Therefore, 
exponential communication bits are needed 
if the implementation simply exchanges intermediate views. 
Here we introduce a technique to compress views by
sharing isomorphic subtrees of the views.
We call a compressed view \emph{a folded view} (or an \emph{f-view}).
The key observation is that there are at most $n$ isomorphic subtrees
in a view when the number of parties is $n$.
This technique reduces not only communication complexity but also
local computation time by \emph{folding} all intermediate views and 
constructing larger f-views 
\emph{without unfolding} intermediate f-views.

In the following, we present an algorithm that constructs an f-view for
each party instead of a view of depth $(n-1)$, and then describe an
algorithm that counts the number of the non-isomorphic views by using the
constructed f-view.
For simplicity, we 
assume the underlying graph of the network
is undirected.
It is not difficult to generalize the algorithms 
to the case of directed network topologies as described at the end of
this section.

\subsection{Terminology}
The folded view has all information possessed by the corresponding view.
To describe such information,
we introduce a new notion, \emph{``path set,''} which is  equivalent
to a view in the sense that any view can be reconstructed from 
the corresponding path set, and vice versa.

A path set, $P_{G,\sigma,X}(v)$, is defined for view
$T_{G,\sigma,X}(v)$.  
Let $u_{\rm root}$ be the root of $T_{G,\sigma,X}(v)$.
Suppose that every edge of a view is directed
and its source is the end closer to $u_{\rm root}$.
$P_{G,\sigma,X}(v)$ is the set
of directed labeled paths starting at $u_{\rm root}$ with infinite length in
$T_{G,\sigma,X}(v)$.  More formally,  let $s(p)=(\Label(u_0),\Label(e_0),\Label(u_1),\cdots )$
be the sequence of labels of those nodes and edges which form an
infinite-length directed path
$p=(u_0,e_0,u_1,\cdots )$
starting at $u_0(=u_{\rm root})$, where $u_i$ is a node, $e_i$ is the directed
edge from $u_i$ to $u_{i+1}$ in 
$T_{G,\sigma,X}(v)$, and $\Label(u_i)$ and $\Label(e_i)$ are the
labels of $u_i$ and $e_i$, respectively.
It is stressed that $u_i$ and $e_i$ are not node or edge identifiers,
and are just used for definition.
$P_{G,\sigma,X}(v)$ is the set of $s(p)$ for all $p$  in
$T_{G,\sigma,X}(v)$. For $T^h_{G,\sigma,X}(v)$, we naturally define
$P^h_{G,\sigma,X}(v)$, i.e., the set of all sequences of labels of
those nodes
and edges which form directed paths of length $h$ starting at $u_{\rm root}$ in 
$T^h_{G,\sigma,X}(v)$.
In the following, we simply call a sequence in a path set, a
\emph{path}, and identify
the common length of the paths in a path set
with the \emph{length of the path set}.

By the above definition, 
$P^h_{G,\sigma,X}(v)$ is easily obtained by traversing view $T^h_{G,\sigma,X}(v)$.
On the other hand, given $P^h_{G,\sigma,X}(v)$, we can construct the view rooted at
$u_{\rm root}$ by sharing the maximal common prefix of any
pair of paths in $P^h_{G,\sigma,X}(v)$.
In this sense, $P^h_{G,\sigma,X}(v)$ has all information possessed by view
$T^h_{G,\sigma,X}(v)$. 
Let $u^j$ be any node at depth $j$ in $T^h_{G,\sigma,X}(v)$,
and suppose that $u^j$ corresponds to node $v_{u^j}$ of $G$.
Since a view is defined recursively, 
we can define the path set $P^{h'}_{G,\sigma,X}(v_{u^j})$ 
for the $h'$-depth subtree rooted at
$u^j$,
as the set of $h'$-length directed paths starting at $u^j$
for $h'\leq h-j$.
To avoid complicated notations, we may use
$P^{h'}_{G,\sigma,X}(u^j)$ instead of $P^{h'}_{G,\sigma,X}(v_{u^j})$.
We call $P^{h'}_{G,\sigma,X}(u^j)$ 
\emph{the path set of length
  $h'$ defined
  for $u^j$}.
In particular,
when $h'$ is the length from $u^j$ to a leaf, i.e.,
$h-j$,
we call $P^{h-j}_{G,\sigma,X}(u^j)$
\emph{the path set defined  for $u^j$}.
For any node $u$ of a view,
we use $\depth(u)$ to represent the depth of $u$, i.e., 
the length of the path from the root to $u$,
in the view.

Finally, for any node $u$ in a view and its corresponding party $l$, 
if an outgoing edge $e$ of $u$ corresponds to the communication link
incident to party $l$ via port $i$,
we call edge $e$ ``the \emph{$i$-edge of $u$}''
and denote the destination of $e$  by $\Adj_i(u)$.

\subsection{Folded view}
\label{subsec:folded view}
We now define a key operation, called the ``merging operation,'' which
folds a view.
\begin{definition}[Merging operation]
\label{def:merging-operation}
For any pair of nodes $u$ and $u'$ at the same depth in a view, 
the merging operation eliminates one of the nodes, $u$ or $u'$, and
its outgoing edges,
and redirects
  all incoming edges of the eliminated node to the remaining one,
if $u$ and $u'$ satisfy the following conditions:
(1) $u$ and $u'$ have the same label, and (2) when $u$ and $u'$
have outgoing edges (i.e., neither $u$ nor $u'$ is a leaf), 
$u$ and $u'$ have the same number of outgoing edges, and the $i$-edges
of $u$ and $u'$ have the same label and
are directed to the same node for all $i$. 
\end{definition}
Obviously, the merging operation never eliminates the root of a view.
Further, the merging operation does not change the length of the
directed path from the root
to each (remaining) node. 
Thus, we define the depth of each node $u$ that remains after 
applying
the merging operation as 
the length of
the path from the root
to $u$ and denote it again by $\depth (u)$.
We call the directed acyclic graph obtained by  applying the merging operation to
a view, a \emph{folded view (f-view)}, and
define \emph{the size of an f-view} as the number of nodes in the f-view.
Since views of finite depth are sufficient for our use, we only consider
f-views that are obtained by applying the merging operation to views of
finite depth hereafter.
For any view $T^h_{G,\sigma,X}(v)$, its minimal f-view is uniquely determined
up to isomorphism as will be proved later and is denoted by
$\widetilde{T}^h_{G,\sigma,X}(v)$.
We can extend the definition of the path set to f-views:
the path set of length $h$ defined for node $u$ in an f-view is the
set of all directed labeled paths of length $h$ from $u$ in the f-view.
For any node $u$ in an f-view, we often use $d_u$ to represent the
number of the outgoing edges of $u$ when describing algorithms later.

The next lemma is essential.
\begin{lemma}
\label{lm:merging operation}
For any (f-)view, the merging operation does not change the path set defined for every
node of the (f-)view if the node exists after the operation.
Thus, the path set of any f-view obtained from a view by
applying the
merging operation is identical to the path set of the view.
\end{lemma}
\begin{proof}
  Let $u'$ be the node that will be merged into $u$ (i.e., $u'$ will be
  eliminated). 
By the definition of the merging operation, the
  set of the maximal-length directed paths starting at $u$ is identical to the set
  of those starting at $u'$. Thus, by eliminating $u'$ and 
redirecting 
  all incoming edges of $u'$ to $u$, the path set defined for every remaining
  node does not change.
\end{proof}
We can characterize f-views by using ``path sets.''
Informally, for every distinct path set $P$ defined for a node at any depth $j$ in a view,
any f-view obtained from the view has at least one node
which defines $P$, at depth $j$.

Before giving a formal characterization of f-views
in the next lemma, we need to define some notations.
Suppose that $u^j$ is any node at depth $j$
in $T^h_{G,\sigma,X}(v)$.
We define $\calP^{j}_{G,\sigma,X}(v)$
as the set of path sets
$P^{h-j}_{G,\sigma,X}(u^j)$'s
defined for all $u^j$'s.
For any path set $P$, 
let $P|_x$ be the set 
obtained by
cutting off the first node and edge
from all those paths in $P$ that
have $x$ as the first edge label.

\begin{lemma}
\label{lm:f-view}
Let $G^f(V^f,E^f)$
be any labeled connected directed acyclic graph such
that $V^f$ is the union of disjoint sets $V^f_j\ (j=0,\ldots, h)$ of
nodes with $\abs{V^f_0}=1$, 
$E^f$ is the union of disjoint sets $E^f_j\subseteq V^f_j\times V^f_{j+1}\
(j=0,\ldots, h-1)$ of directed edges, and
every node in $V^f_j\ (j=1,\ldots, h)$ can be reached from $u_r\in
V^f_0$ via directed edges in $E^f$.
$\widehat{T}_{G,\sigma,X}^h(v)$ is an f-view of $T^h_{G,\sigma,X}(v)$
if and only if
$\widehat{T}_{G,\sigma,X}^h(v)$ is $G^f(V^f,E^f)$ such that
there is a mapping $\psi$ from $V^f_j$
onto $\calP^{j}_{G,\sigma,X}(v)$
 for each $j=0,\ldots, h$ satisfying the following two conditions:
  \begin{enumerate}
\item[C1] each node $u\in V^f$ has the label that is identical to
the common label of the first nodes of paths in $\psi(u)$,
\item[C2] 
for each $u\in V^f$, 
$u$ has an outgoing edge $(u,u')$ with label $x$ if and only if
there is a path in $\psi (u)$ whose first edge is labeled with $x$,
and $\psi (u')=\psi (u)|_{x}$.
  \end{enumerate}
\end{lemma}
\begin{proof}
($\Rightarrow$) We will prove that, for any f-view obtained by applying the merging operation to 
$T^h_{G,\sigma,X}(v)$, there exists $\psi$ that satisfies C1 and C2.
  From Lemma~\ref{lm:merging operation}, the merging operation does not change the
  path set defined for any node (if it exists after the operation). It follows that the path set defined for any
  node at depth $j$ in the f-view belongs to $\calP^{j}_{G,\sigma,X}(v)$.
Conversely, for every path set $P$ in $\calP^{j}_{G,\sigma,X}(v)$, 
there is at least one node at depth $j$ in the f-view such that
the path set defined for the node is $P$; this is because
the merging operation just merges two nodes defining the same path set.
Let $\psi$ be the mapping that
maps every node $u$ of the f-view 
to the path set defined for $u$.
From the above argument,
$\psi$ is a mapping
from $V^f_j$ onto $\calP^{j}_{G,\sigma,X}(v)$ and
meets C1. 
To show that $\psi$ meets C2, we use simple induction on the sequence of the
  merging operation.
By the definition,  $T^h_{G,\sigma,X}(v)$ meets C2.
Suppose that one application of the merging operation
transformed an f-view to a smaller f-view,
and that $\psi$ meets C2 for the f-view before the operation.
If we define $\psi '$ for the smaller f-view as the mapping obtained by
  restricting $\psi$ to the node set of the smaller f-view,
$\psi'$ meets C2  by
the definition of the merging operation.

($\Leftarrow$) We will prove that any graph $G^f$
for which there is $\psi$ satisfying C1 and C2,
can be
obtained
by applying the merging operation (possibly, more than once) to $T^h_{G,\sigma,X}(v)$.
We can easily show by induction that the set $P$ of all maximal-length labeled directed paths starting at $u_0\in V^f_0$
is identical to the path set 
$P^h_{G,\sigma,X}(v)$
of $T^h_{G,\sigma,X}(v)$ from the
definition of $G^f$. 
We will give an 
inversion of the merging
operation 
that does not change $P$
when it is applied to $G^f$,
and show that we can obtain
the view that defines $P^h_{G,\sigma,X}(v)$
by repeatedly applying the inversion to $G^f$
until the inversion cannot be applied any more.
This view is isomorphic to $T^h_{G,\sigma,X}(v)$,
since the view is uniquely determined for a fixed path set.
It follows that $G^f$ can be obtained from 
$T^h_{G,\sigma,X}(v)$
by reversing the sequence of the inversion, i.e., applying the merging
operation repeatedly.

The inverse operation of the merging operation is defined as follows:
if some node $u^j\in V^f_j \ (1\leq j\leq h)$ has multiple incoming edges, say,
$e_1,\ldots, e_t\in E_{j-1}$, the inverse operation
makes a copy $u'$ of $u^j$ together with its outgoing edges 
(i.e, creates a new node $u'$ with the same label as $u^j$,
and edge $(u',w)$ with label $x$ if and only if 
edge $(u^j,w)$ has label $x$ for every outgoing edge $(u^j,w)$ of $u^j$)
and
redirects $e_2,\ldots,e_t$ to $u'$
($e_1$ is still directed to the original node $u^j$).
Let $G^{f'}=(V_{f'},E_{f'})$ be the resulting graph.
Consider mapping $\psi'$ such that
$\psi'$ is identical to $\psi$ of $G^f$
for all nodes except $u'$, and $\psi'$ maps $u'$ to $\psi (u)$.
Then $\psi'$ is a mapping
from $V^{f'}_j$ onto $\calP^{j}_{G,\sigma,X}(v)$ and
meets C1 and C2. 
The sets of maximal-length paths from $u_0\in V^f_0$ and $u_0'\in
V^{f'}_0$
are obviously identical to each other.
Thus, the inverse operation can be applied repeatedly
until there are no nodes
that have multiple incoming edges,
which does not change the set of 
maximal-length paths.
It follows that $G^f$ is transformed into a 
view that defines $P^h_{G,\sigma,X}(v)$, i.e., $T^h_{G,\sigma,X}(v)$,
by repeatedly applying the operation until it can no longer be applied.
\end{proof}
From this lemma, we obtain the next corollary.
\begin{corollary}
\label{cr:minimal-f-view}
Any minimal f-view $\widetilde{T}^h_{G,\sigma,X}(v)$ is unique up to
isomorphism and has exactly $|\calP^{j}_{G,\sigma,X}(v)|$ nodes at
depth $j\ (0\leq j\leq h)$.
The minimal f-view of depth $h$ for any $n$-party distributed system has $O(hn)$ nodes
and $O(hDn)$ edges, where $D$ is the maximum degree over all nodes of the underlying
graph.
\end{corollary}
\begin{proof}
When $\psi$ in Lemma~\ref{lm:f-view} is a \emph{bijective} mapping from $V^f_j$
to $\calP^{j}_{G,\sigma,X}(v)$ for all $j$,
the f-view is minimal.
Thus, the f-view has $|\calP^{j}_{G,\sigma,X}(v)|$ nodes at
depth $j\ (0\leq j\leq h)$. 

Let $\widetilde{T}^h_{X,a}(v)$ and $\widetilde{T}^h_{X,b}(v)$ be
any two minimal f-views
of $T^h_{G,\sigma,X}(v)$, and 
let $\psi _a$ and $\psi _b$ be 
their corresponding bijective mappings $\psi$
defined in Lemma~\ref{lm:f-view}, respectively.
If we define $\phi=\psi _b^{-1}\psi _a$ 
for the inverse mapping $\psi^{-1} _b$ of $\psi _b$,
$\phi$ is a bijective mapping from the node set of 
$\widetilde{T}^h_{X,a}(v)$ to that of $\widetilde{T}^h_{X,b}(v)$.
Suppose that any node $u_a$ at depth $j$ of 
$\widetilde{T}^h_{X,a}(v)$ is mapped by $\psi _a$
to some path set $P$ in $\calP^{j}_{G,\sigma,X}(v)$, 
which is mapped to some node $u_b$ at depth $j$ by 
$\psi ^{-1}_b$. Obviously, $u_a$ and $u_b$ have the same degree and 
have the same label as the first node of paths in $P$.
Let $(u_a,u_a')$ be an edge with any label $x$. 
Node $u_a'$ is then mapped to $\psi_a (u)|_x$, 
which is mapped to node $u'_b$
incident to the directed edge with label $x$
emanating from $u_b$.
For each $u_a$ and $x$,
thus, there is edge $(\phi(u_a),\phi(u_a'))$ with label $x$ if 
edge $(u_a,u_a')$ with label $x$ exists in $\widetilde{T}^h_{X,a}(v)$.
By a similar argument, 
there is edge $(\phi^{-1}(u_b),\phi^{-1}(u_b'))$ with label $x$ if 
edge $(u_b,u_b')$ with label $x$ exists in $\widetilde{T}^h_{X,b}(v)$
over all $u_b$ and $x$.
Thus, $\phi$ is an isomorphism from 
$\widetilde{T}^h_{X,a}(v)$ to $\widetilde{T}^h_{X,b}(v)$.

For the second part of the lemma, if there are $n$ parties,
it is obvious that
$|\calP^{j}_{G,\sigma,X}(v)|\leq n$ for any $j\ (0\leq j\leq h)$.
Since each node has at most $D$ outgoing edges, the lemma follows.
\end{proof}
\subsection{Folded-view minimization}
The idea of the minimization algorithm is to
repeatedly apply the merging operation to the (f-)view to be
minimized until it can no longer be applied.
This idea works well
because of the next lemma.
\begin{lemma}
\label{lm: minimization by merging operation}
  Let $\widehat{T}^{h}_{G,\sigma,X}(v)$ be an f-view for view $T^{h}_{G,\sigma,X}(v)$. $\widehat{T}^{h}_{G,\sigma,X}(v)$ is isomorphic to the
  minimal f-view $\widetilde{T}^{h}_{G,\sigma,X}(v)$ if and only if
  the merging operation is not
  applicable to $\widehat{T}^{h}_{G,\sigma,X}(v)$.
\end{lemma}
\begin{proof}
  Obviously, no more merging operations can be applied to the minimal
  f-view. We will prove the other direction in the following.
Suppose that there is a non-minimal f-view $\widehat{T}^h$ expressing $P^{h}_{G,\sigma,X}(v)$, to which no
  more merging operations
  can be applied. $\widehat{T}^h$ must have more than $|\calP^{j}_{G,\sigma,X}(v)|$ nodes at depth $j$
  for some $j$.
Let $j'$ be the largest such $j$.  
From Lemma~\ref{lm:merging operation}, 
for any node $u^{j'}$ at depth $j'$ in $\widehat{T}^h$, 
the path set defined for $u^{j'}$ is in $\calP^{j'}_{G,\sigma,X}(v)$
if $\widehat{T}^h$ is an f-view of $T^{h}_{G,\sigma,X}(v)$.
Thus, $\widehat{T}^h$ must have at
  least one pair of nodes at depth $j'$ such that the path sets
  defined for
  the two nodes are identical. 
This implies the next facts:
the two nodes have the same label;
if $j'+1\geq h$,
the outgoing edges of the two nodes
  with the same edge label are directed to the same node at depth $j'+1$, since no two
  nodes at depth $j'+1$ have the same path set. Thus, the merging
  operation can still be applied to the node pair. This is a contradiction.
\end{proof}
\subsubsection{The algorithm}
The minimization algorithm applies the merging operation to every node
in the (f-)view in a bottom-up manner, i.e.,
in decreasing order of node depth.
Clearly, this ensures that no application of the merging operation at any depth $j$
creates a new node pair at depth larger than $j$ to which the
merging operation is applicable.
Thus, no more merging operations can be applied when the algorithm halts.
It follows that the algorithm outputs the minimal f-view
by Lemma~\ref{lm: minimization by merging operation}.

In order to apply the merging operation, we need to be able to decide
if two edges are directed to the same node, which implies that we need to
identify each node. We thus assign a unique identifier, denoted by
$\id (u)$, to each node $u$ in the (f-)view. 
In order to efficiently check condition (2)
of Definition~\ref{def:merging-operation} (i.e., the definition of the merging operation), we also construct a data structure for each node that includes
the label of the node, and 
the labels and destination node id of all outgoing edges of the node: 
the data structure, called \emph{key}, for node $u$ is of the form of
$(\Label(u),\ekey(u))$. Here, 
$\Label(u)$ is the label of $u$;
$\ekey(u)$
is a linked list
of pairs $(x_1,y_1)\cdots (x_{d_u},y_{d_u})$
of the label $x_i$ and
destination $y_i$, respectively, of the $i$-edge of $u$ for all $i$,
where $d_u$ is the number of the outgoing edges of $u$.

We prepare another linked list $V^f_j$ 
of all nodes at depth $j$ for each $j$ to make sure that
the merging operation is applied to 
the set of all nodes at depth $j$ before moving on to depth $j-1$.
Data structures  $\id(u)$ and $\ekey(u)$ for every node $u$ and $V^f_j$
for every depth $j$
can be constructed by one traversal of the
input f-view in a breadth-first manner.

We then perform the merging operation in a bottom-up manner.
For each depth $j$ from $h$ to $1$, the next operations are
performed. (Notice that the algorithm terminates at $j=1$, since depth 0 holds
only the root, which is never removed.)
First, we sort nodes $u$'s in $V^f_j$ by regarding their keys, i.e.,  $(\Label(u),\ekey(u))$ as a binary string,
which makes
all nodes having the same key adjacent to each
other in $V^f_j$.
For each maximal subsequence of those nodes $u$'s in $V^f_j$ which have the same pair $(\Label(u),\ekey(u))$,
we will eliminate all nodes but the first node in the subsequence
and redirect all incoming edges of the eliminated nodes to the first
node, which realizes the merging operation.
Here we introduce variables $\primary$ and $\primarykey$
to store 
the first node $u$ and its key $(\Label(u),\ekey(u))$, respectively,
of the subsequence currently being processed.
Concretely, we perform the next operation on every node $u$ in $V^f_j$ in sorted order:
if $(\Label(u),\ekey(u))$ is equal to $\primarykey$ (i.e., conditions (1) and (2)
of the merging operation are met),
remove $u$ from $V^f_j$,
redirect all incoming edges of $u$ to $\primary$, 
remove $u$ and all its outgoing edges from $\widehat{T}^h$,
and set $\id(u)$ to $\id (\primary)$ 
to make $\ekey$ consistent with the merger;
otherwise, set $\primary$ to $u$ and $\primarykey$ to $(\Label(u),\ekey(u))$.

More precisely, the minimization algorithm described in
Figure~\ref{fig:f-view minimization} is invoked with a(n) (f-)view $T^h$ and its
depth $h$ (actually, $h$ can be computed from $T^h$, but we give $h$ as
input to simplify the algorithm). 
The minimization algorithm calls Subroutine~Traversal~(I) shown
in Figure~\ref{fig:traversal I} to
compute $\id$, $\ekey$ and $V_j^f$,
where
$\CONTINUE$ starts the
next turn of the inner-most loop where it runs with the updated index;
$\DEQUEUE(Q)$ removes an element from FIFO queue $Q$ and returns the
element; $\ENQUEUE (Q,q)$ appends element $q$
to FIFO queue $Q$;
$\CON(L,l)$ appends element $l$ to the end of list $L$.

The minimization algorithm will be used as a
subroutine when constructing the minimal f-view 
from scratch as described later.
\subsubsection{Time Complexity}
The next lemma states the time complexity of the minimization algorithm.
\begin{figure}[t]
\begin{center}
\hrulefill\\
\textbf{F-View Minimization Algorithm}
\vspace{-2mm}
\begin{description}
\setlength\itemsep{-3pt}
  \item[Input:]
    a(n) (f-)view $\widehat{T}^h$ of depth $h$, and a positive integer $h$
  \item[Output:]
    minimal f-view $\widetilde{T}^h$
\end{description}
\vspace{-6mm}
\noindent
\begin{enumerate}
\setlength\itemsep{-3pt}
\item Call Subroutine~Traversal~(I)
with $\widehat{T}^h$,\\
to compute $\id$, $\ekey$ and $V^f_j\ (j=1,\ldots ,h)$ by breadth-first traversal of $\widehat{T}^h$.
\item For $j:=h$ down to 1, do the following steps.
\vspace{-6pt}
  \begin{enumerate}
\setlength\itemsep{0pt}
\item Initialize $\primarykey$ to an empty list.
  \item Sort all elements $u$ in $V^f_j$ 
       by the value obtained by regarding $(\Label(u),\ekey(u))$ as a binary string.
  \item While $V^f_j\neq \emptyset$, repeat the following steps.
\vspace{-6pt}
\begin{enumerate}
\setlength\itemsep{0pt}
\item Remove the first element of $V^f_j$ and set $u$ to the element.
  \item If $(\Label(u),\ekey(u))=\primarykey$, \\
\hspace{5mm}    redirect all incoming edges of $u$ to $\primary$,\\
\hspace{5mm}    eliminate $u$ and all its outgoing edges (if they
exist), and\\
\hspace{5mm}    set $\id(u):=\id(\primary)$ 
to make $\ekey$ consistent with this merger;\\
\hspace{4mm}otherwise,\\
\hspace{6mm}set $\primary:=u$ and $\primarykey:=(\Label(u),\ekey(u))$.
\end{enumerate}
\end{enumerate}
\vspace{-6pt}
\item Output the resulting graph $\widetilde{T}^h$.
\end{enumerate}
\vspace{-1\baselineskip}
\hrulefill
\end{center}
\vspace{-1\baselineskip}
  \caption
[Folded view minimization algorithm.]
{F-view minimization algorithm.}
  \label{fig:f-view minimization}
\end{figure}
\begin{figure}[t]
\begin{center}
\hrulefill\\
\textbf{Subroutine~Traversal~(I)}
\vspace{-2mm}
\begin{description}
\setlength\itemsep{-3pt}
  \item[Input:]
    a(n) (f-)view $\widehat{T}^h$
  \item[Output:]
$\id$, $\ekey$ and $V^f_j\ (j=1,\ldots, h)$ 
\end{description}
\vspace{-6mm}
\begin{enumerate}
\setlength\itemsep{-3pt}
\item
  Perform $\ENQUEUE$$(Q,u_r)$,
where $u_r$ is the root of $\widehat{T}^h$ and $Q$
  is a FIFO queue initialized to an empty queue.\\
  Set $\depth(u_r):=0$ and $\fsize:=1$.\\
\item 
  Set $\id(u_r):=\fsize$ and then set $\fsize := \fsize +1$.
\item While $Q$ is not empty, repeat the following steps.
\vspace{-6pt}
\begin{enumerate}
\setlength\itemsep{0pt}
  \item[3.1] Set $u:=\DEQUEUE(Q)$ and then initialize $\ekey(u)$ to an
    empty list.\\
    If $u$ is a leaf, $\CONTINUE$.
  \item[3.2] For $i:=1$ to $d_u$, where $d_u$ is the degree of $u$,
\vspace{-6pt}
\begin{enumerate}
\setlength\itemsep{0pt}
    \item[3.2.1] Set $u_i:=\Adj_i(u)$.\\
                 If $u_i$ has already been traversed, $\CONTINUE$.
    \item[3.2.2] Set $\id(u_i):=\fsize$ and then set $\fsize:=\fsize+1$.
    \item[3.2.3] Set $\depth(u_i):=\depth(u)+1$.\\
                 Perform $\CON(V^f_{\depth(u_i)},u_i)$ and $\ENQUEUE$$(Q,u_i)$.
    \item[3.2.4] Perform $\CON(\ekey(u), (\Label((u,u_i)),\id(u_i)))$.\\
\end{enumerate}
\end{enumerate}
\vspace{-6pt}
\item Output $\id$, $\ekey$ and $V^f_j\ (j=1,\ldots ,h)$.
\end{enumerate}
\vspace{-1\baselineskip}
\hrulefill
\end{center}
\vspace{-1\baselineskip}
  \caption{Subroutine~Traversal~(I).%
}
  \label{fig:traversal I}
\end{figure}
\begin{lemma}
\label{lm:complexity of f-view minimization}
If the input of the minimization algorithm is an f-view with node set
$V^f$ for an $n$-party distributed system, and any node label is
represented by an $O(\log L)$-bit value for some positive integer $L$,
the time complexity of the algorithm is $O(|V^f|(\log |V^f|)(\log
L+D\log (D|V^f|)))$, where $D$ is the maximum degree of the nodes in
the underlying graph.
\end{lemma}
\begin{proof}
  We first consider Subroutine~Traversal~(I) in
  Figure~\ref{fig:traversal I}.  It is easy to see that step~1 takes
  constant time, and that step~2 takes $O(\log|V^f|)$ time.  Notice
  that, for a standard implementation of $\DEQUEUE$, $\ENQUEUE$ and
  $\CON$, each call of them takes constant time.  Step 3 traverses the
  input (f-)view in a breadth-first manner.  Hence, the time required
  for step 3 is proportional to the number of edges, which is at most
  $D|V^f|$.  Step~3.1 takes constant time,  while steps~3.2.2 and
  3.2.3 take $O(\log \abs{V^f})$.  It follows that step~3 takes
  $O(D|V^f|\log \abs{V^f})$ time.  The total time complexity of
  Traversal~(I) is $O(D|V^f|\log \abs{V^f})$.
  
  Now we consider the minimization algorithm in Figure~\ref{fig:f-view
    minimization}.  In step 2, steps 2.2 and 2.3 are dominant.  In
  step 2.2, sorting all elements in $V^f_j$ for all $j$ needs
  $O(|V^f|\log |V^f|)$ comparisons and takes $O(\log
  L+D\log(D\abs{V^f}))$ time for each comparison, since $\Label(u)$
  and $\ekey(u)$ have $O(\log L+D\log(D\abs{V^f}))$ bits for any $u$
  ($\ekey(u)$ has at most $D$ pairs of an edge label and a node $\id$,
  which are $\lceil\log D\rceil$ bits and $\lceil\log|V^f|\rceil$
  bits, respectively).   Thus step 2.2 takes $O(|V^f|(\log|V^f|)(\log
  L+D\log(D|V^f|))$ time.
  
  Step 2.3 repeats steps 2.3.1 and 2.3.2 at most $|V^f|$ times, since
  steps 2.3.1 and 2.3.2 are each performed once for every node in
  $\widehat{T}^h$ except the root (the steps are never performed on the
  root).  Clearly, each run of step 2.3.1 takes constant time.  In
  each run of step 2.3.2,  (a) it takes $O(\log L +D\log (D|V^f|))$
  time to compare $(\Label(u),\ekey(u))$ with  $\primarykey$,
  (b) it takes $O(d_u^I)$ time to redirect all incoming edges of $u$
  to $\primary$ where $d_u^I$ is the number of the incoming edges of
  $u$, and (c) it takes $O(d_u)$ time to remove $u$ and all its
  outgoing edges (if they exist) from $\widehat{T}^h$.  For more precise
  explanations of (b) and (c), the next data structure is assumed to
  represent $\widehat{T}^h$: each $u$ has two linked lists of incoming
  edges and outgoing edges such that  each edge is registered in the
  incoming-edge list of its destination and the outgoing-edge list of
  its source, and the two entries of the edge lists have pointers to
  each other.
  By using this data structure, we can easily see that the edge
  redirection in (b) and the edge removal in (c) can be done in
  constant time for each edge, since only a constant number of
  elements need to be appended to or removed from the linked lists.

  Finally, it takes
  constant time to set $\id (u)$,  $\primary$ and $\primarykey$ to new
  values.
  The total time required for step 2.3 is proportional to
\begin{eqnarray*}
\lefteqn{
O\left(
|V^f|(\log L+D\log (D|V^f|)) +
\sum _{u\in V^f}d_u^I
+
\sum _{u\in V^f}d_u
\right)
}
\\
&
=
&O\left(
|V^f|(\log L+D\log (D|V^f|)) +D|V^f|+D|V^f|
\right),
\end{eqnarray*}
since no edge can be redirected or removed more than once in step~2.3.2.  Hence, step 2.3 takes $O(|V^f|(\log L+D\log (D|V^f|))$ time.

By summing up these elements, the total time complexity is given as
$O(|V^f|(\log|V^f|)(\log L+D\log(D|V^f|)).$
\end{proof}
\subsection{Minimal folded-view construction}
\label{subsec:minimal-folded-view}
We now describe the entire algorithm that constructs a minimal f-view
of depth $h$ from scratch by using the f-view minimization algorithm
as a subroutine.  
This construction algorithm is almost the same as
the original view construction algorithm except that parties exchange 
and perform local operations on
f-views instead of views: every party
constructs an f-view $\widehat{T}^j_{G,\sigma,X}(v)$ of depth $j$  by
connecting each received minimal f-view
$\widetilde{T}^{j-1}_{G,\sigma,X}(v_i)$ of depth $j-1$ with the root
without unfolding them, and then applies the f-view minimization
algorithm to $\widehat{T}^j_{G,\sigma,X}(v)$.  
It is stressed that
$\widehat{T}^{j}_{G,\sigma,X}(v)$ is an f-view, since
$\widehat{T}^{j}_{G,\sigma,X}(v)$ can be constructed from view
$T^{j}_{G,\sigma,X}(v)$ by applying the merging operation to every
subtree rooted at depth 1.
Thus, the minimization algorithm can be applied to
$\widehat{T}^{j}_{G,\sigma,X}(v)$.  More precisely, each party $l$
having $d_l$ adjacent parties  and $x_l$ as his label performs the
f-view construction algorithm described in  Figure~\ref{fig:f-view
  construction} with $h$, $d_l$ and $x_l$, in which we assume that $v$
is the node corresponding to party $l$ in the underlying graph.
\begin{figure}[t]
\begin{center}
\hrulefill\\
\textbf{F-View Construction Algorithm}\\
\begin{description}
\setlength\itemsep{-3pt}
  \item[Input:]
    integers $h$, $d$ and $x$ 
  \item[Output:]
    minimal f-view $\widetilde{T}^{h}_{G,\sigma,X}$, where $X$ is
    the underlying mapping naturally induced by the $x$ values of all parties.
\end{description}
\vspace{-6mm}
\begin{enumerate}
\setlength\itemsep{-3pt}
  \item Generate $\widetilde{T}^0_{G,\sigma,X}(v)$, which consists of only one node
  with label $x$.\\
  \item For $j:=1$ to $h$, perform the following steps.
\vspace{-6pt}
  \begin{enumerate}
\setlength\itemsep{0pt}
      \item Send a copy of $\widetilde{T}^{j-1}_{G,\sigma,X}(v)$ to every adjacent party.
    \item Receive the minimal f-view
      $\widetilde{T}^{j-1}_{G,\sigma,X}(v_i)$ via port $i$ for $1\leq
      i\leq d$, where $v_i$ is the node corresponding to the party
      connected via port $i$.
    \item Construct an f-view $\widehat{T}^{j}_{G,\sigma,X}(v)$
      from $\widetilde{T}^{j-1}_{G,\sigma,X}(v_i)$'s as follows.
\vspace{-6pt}
      \begin{enumerate}
      \item Let root $u_{\rm root}$ of
      $\widehat{T}^{j}_{G,\sigma,X}(v)$ be $\widetilde{T}^0_{G,\sigma,X}(v)$.
      \item Let the $i$th child $u_i$ of $u_{\rm root}$ be the root of
      $\widetilde{T}^{j-1}_{G,\sigma,X}(v_i)$.
      \item Label edge $(u_{\rm root},u_i)$ with
$(i,i')$, where $i'$ is the port through which
$v_i$ sent $\widetilde{T}^{j-1}_{G,\sigma,X}(v_i)$,\\
i.e.,
$i':=\sigma [v_i](v,v_i)$.
      \end{enumerate}
\vspace{-6pt}
      \item  Call the f-view minimization algorithm with
        $\widehat{T}^{j}_{G,\sigma,X}(v)$ and $j$ to obtain $\widetilde{T}^{j}_{G,\sigma,X}(v)$.\\
  \end{enumerate}
\item Output $\widetilde{T}^{h}_{G,\sigma,X}(v)$.
\end{enumerate}
\vspace{-1\baselineskip}
\hrulefill
\end{center}
\vspace{-1\baselineskip}
  \caption
[Folded view construction algorithm.]
{F-view construction algorithm.}
\label{fig:f-view construction}
\end{figure}

\begin{lemma}
  For any distributed system of $n$ parties labeled with  $O(\log
  L)$-bit values, the f-view construction algorithm described in  Figure~\ref{fig:f-view
  construction} constructs the
  minimal f-view of depth $h(=O(n))$ in $O(Dh^2n(\log n)(\log Ln^D))$
  time for each party and $O(|E|h^2n\log (LD^D))$ communication
  complexity, where $|E|$ and $D$ are the number of edges and the
  maximum degree, respectively, of the nodes of the underlying graph.
\label{lm: f-view construction}
\end{lemma}
\begin{proof}
  $\widetilde{T}^{j-1}_{G,\sigma,X}(v)$ has at most $j\cdot n$ nodes,
  and every node has at most $D$ outgoing edges, each of which is
  labeled with an $O(\log L)$-bit value.
  Thus, $\widetilde{T}^{j-1}_{G,\sigma,X}(v)$ can be expressed by
  $O(jn\log L+jDn\log D)=O(jn\log (LD^D))$ bits.  It follows that
  steps 2.1 and 2.2 take $O(jDn\log (LD^D))$ time, since any party has
  at most $D$ neighbors.  
Since $\widehat{T}^{j}_{G,\sigma,X}(v)$ consists of a root and
$D$ minimal f-views of depth $j-1$, 
$\widehat{T}^{j}_{G,\sigma,X}(v)$ has at
  most $(j\cdot D\cdot n+1)$ nodes.  From Lemma~\ref{lm:complexity of f-view minimization}, step 2.4 in Figure~\ref{fig:f-view construction} takes
$$
O(jDn\log (jDn)(\log L+D\log(D\cdot jDn)))
=O(jDn(\log n)(\log L+D\log n))
$$
time for each $j$, since $j=O(n)$.  Thus the total time complexity
is $$
O\left(\sum_{j=1}^{h}jDn(\log n)(\log L+D\log n)\right)
=O(Dh^2n(\log n)(\log Ln^D)).
$$

We now consider the communication complexity.
Since the
minimal f-view of depth $j$ can be expressed by $O(jn\log
(LD^D))$ bits as described above,
the total number of the bits exchanged by all parties is  $O(j\cdot\abs{E} n\log (LD^D))$
for each $j$.  It follows that the total communication complexity  to
construct an f-view of depth $h$ is
$$
O\left(
\sum_{j=1}^{h}(j\abs{E}n\log (LD^D))
\right)
=O\left(
|E|h^2n\log (LD^D)
\right).
$$
\end{proof}
\subsection{Counting the number of parties having specified values}
In many cases, including ours,
the purpose of constructing a view 
is
to 
compute $|\Gamma _{X}^{(n-1)}(S)|$ 
for any
set $S\subseteq X$
in order to compute $c _{X}(S)=n|\Gamma _{X}^{(n-1)}(S)|/|\Gamma
_{X}^{(n-1)}(X)|$, i.e.,
the number of parties having values in $S$.
We will describe an algorithm that computes $|\Gamma ^{(n-1)}_X(S)|$
for given minimal f-view $\widetilde{T}^{2(n-1)}_X(v)$,
set $S$, and $n$.  
Hereafter, ``\emph{a sub-f-view rooted at $u$}'' means
the subgraph of an f-view
induced by node $u$ and all other nodes
that can be reached from $u$ via directed edges.
\subsubsection{View Counting Algorithm}
The algorithm computes the maximal set $W$ of those
nodes of depth of at most $n-1$ 
in $\widetilde{T}^{2(n-1)}_{X}(v)$
which define distinct path sets of length
$n-1$.
The algorithm then
computes $|\Gamma _{X}^{(n-1)}(S)|$ by counting
the number of those nodes in $W$ which
are labeled with values in $S$.

To compute $W$, the algorithm 
first sets $W$ to $\{ u_r\}$, where $u_r$ is the root of $\widetilde{T}^{2(n-1)}_{X}(v)$,
and repeats the following operations for
every node $u$ of $\widetilde{T}^{2(n-1)}_{X}(v)$
at depth of at most $n-1$
in a breadth-first order:
For each node $\hat{u}$ in $W$,
the algorithm calls Subroutine P (described later)
with $u$ and $\hat{u}$
to test if the sub-f-view rooted at $u$ 
has the same path set of length $n-1$ as that rooted at $\hat{u}$,
and sets $W:=W\cup \{u\}$ if the test is false
(i.e., the two sub-f-views do not have the same path set of length $n-1$).
After processing all nodes at depth of at most $n-1$,
we can easily see that $W$ is the maximal 
subset of nodes at depth of at most $n-1$
such that no pair of sub-f-views rooted at nodes in $W$ have
a common path set of length $n-1$.

The algorithm is precisely described in
Figure~\ref{fig: view-counting},
in which
Subroutine~Traversal~(II) is called in the first step
in order to prepare the next objects,
which helps the breadth-first traversals performed in the algorithm:
(1) the size, denoted by $\fsize$, of $\widetilde{T}^{2(n-1)}_{X}(v)$,
(2) function 
$\depth:V^f\rightarrow \{0,1,\ldots, 2(n-1)\}$
that gives the depth of any node in $\widetilde{T}^{2(n-1)}_{X}(v)$,
(3) bijective mapping
$\id :V^f\rightarrow \{1,\ldots,|V^f|\}$
that gives the order of breadth-first traversal,
and (4) the inverse mapping $\id ^{-1}$ of $\id$.
Although Subroutine~Traversal (II) is 
just a breadth-first-traversal based subroutine,
we give a precise description
in
Figure~\ref{fig:traversal II}
just to support complexity analysis described later,
where
$\CONTINUE$ starts the
    new turn of the inner-most loop where it runs with the updated
    index; 
$\BREAK$ quits the inner-most loop and moves on to the next operation;
$\DEQUEUE(Q)$ removes an element from FIFO queue Q and returns the element; $\ENQUEUE(Q,q)$ appends $q$
to $Q$. These operations are assumed to be implemented
in a standard way.

\begin{figure}[t]
\begin{center}
\hrulefill\\
\textbf{View Counting Algorithm}\\
\begin{description}
\setlength\itemsep{-3pt}
  \item[Input:]
    minimal f-view $\widetilde{T}^{2(n-1)}_X(v)$, subset $S$ of the
    range of $X$, and integer $n$
  \item[Output:]
    $|\Gamma_X ^{(n-1)}(S)|$
\end{description}
\vspace{-6mm}
\begin{enumerate}
\setlength\itemsep{-3pt}
\item Call Subroutine~Traversal~(II) 
with $\widetilde{T}^{2(n-1)}_X(v)$ to
  compute $\fsize$,
  $\depth$, $\id$ and
  $\id^{-1}$.
\item Let $W$ be $\{u_r\}$, where $u_r$ is the root of $\widetilde{T}^{2(n-1)}_X(v)$.
\item For $i:=2$ to $\fsize$, perform the next operations.
\vspace{-6pt}
  \begin{enumerate}
\setlength\itemsep{-0pt}
  \item If $\depth(\id^{-1}(i))>n-1$, $\BREAK$; otherwise set $u:=\id^{-1}(i)$.
  \item For each $\hat{u}\in W$, perform the next operations.\\
    \begin{enumerate}
    \item Call Subroutine~P with integer $n$, two functions
      ($\depth$, $\id$), and two
      sub-f-views rooted at $u$ and $\hat{u}$, in order to test if the two sub-f-views have the same path set of length
$n-1$.\\
    \item Set $W := W\cup \{u\}$ if Subroutine~P outputs ``$\mathsf{No}$.''
    \end{enumerate}
  \end{enumerate}
\vspace{-6pt}
\item Count the number $n_S$ of the nodes in $W$ that are labeled
  with some value in $S$.
\item Output $n_S$.
\end{enumerate}
\vspace{-1\baselineskip}
\hrulefill
\end{center}
\vspace{-1\baselineskip}
  \caption{View counting algorithm.}
\label{fig: view-counting}
\end{figure}

\begin{figure}[t]
\begin{center}
  \hrulefill\\
\textbf{Subroutine~Traversal~(II)}\\
\begin{description}
\setlength\itemsep{-3pt}
  \item[Input:]
    minimal f-view $\widetilde{T}^{2(n-1)}_X(v)$.
  \item[Output:]
     variable $\fsize$ of $\widetilde{T}^{2(n-1)}_X(v)$, and functions
     $\depth$, $\id$ and $\id^{-1}$.
\end{description}
\vspace{-6mm}
  \begin{enumerate}
\setlength\itemsep{-3pt}
\item $\ENQUEUE (Q,u_{\rm root})$,
 where $Q$ is a FIFO queue
  initialized to an empty queue, 
and $u_{\rm root}$ is the root of
$\widetilde{T}^{2(n-1)}_X(v)$.\\
Set $\fsize:=1$.
\item Set $\id(u_{\rm root}):=\fsize$ and $\id^{-1}(\fsize):=u_{\rm
    root}$.\\ 
   Set $\depth(u_{\rm root}):=0$.
\item Set $\fsize := \fsize +1$.
\item While $Q$ is not empty, repeat the following steps.
\vspace{-6pt}
  \begin{enumerate}
\setlength\itemsep{0pt}
  \item Set $u:=\DEQUEUE(Q)$.
  \item If $u$ is a leaf, $\CONTINUE$.
  \item For $i:=1$ to $d_u$, do the following.
    \begin{enumerate}
\setlength\itemsep{0pt}
    \item Set $u_i:=\Adj_i(u)$.
    \item If $u_i$ has already been traversed, $\CONTINUE$.
    \item Set $\id(u_i):=\fsize$ and $\id^{-1}(\fsize):=u_i$.\\
      \hspace{-3pt} Set $\depth(u_i):=\depth(u)+1$.\\
    \item $\ENQUEUE(Q,u_i)$.
    \item Set $\fsize:=\fsize+1$.
   \end{enumerate}
  \end{enumerate}
\vspace{-6pt}
\item Output $\fsize$, $\depth$, $\id$ and $\id^{-1}$.
  \end{enumerate}
\vspace{-1\baselineskip}
  \hrulefill
\end{center}
\vspace{-1\baselineskip}
  \caption{Subroutine~Traversal~(II).}
  \label{fig:traversal II}
\end{figure}
\subsubsection{Subroutine~P}
Subroutine P is based on the next lemma.
\begin{lemma}
\label{lm:test common path set}
 Suppose that $\widehat{T}^{(n-1)}_{X,a}$ and $\widehat{T}^{(n-1)}_{X,b}$ are any two
  sub-f-views of depth $(n-1)$ of a minimal f-view
  $\widetilde{T}^{2(n-1)}_X(v)$, such that,
for roots $u_r$ and $w_r$ of $\widehat{T}^{(n-1)}_{X,a}$ and
  $\widehat{T}^{(n-1)}_{X,b}$, respectively,
$\depth (u_r)\leq \depth (w_r)\leq (n-1)$.
Let $V_a$ and $V_b$ be the vertex sets of 
$\widehat{T}^{(n-1)}_{X,a}$ and $\widehat{T}^{(n-1)}_{X,b}$, respectively, and let
$E_a$ and $E_b$ be the edge sets of 
$\widehat{T}^{(n-1)}_{X,a}$ and $\widehat{T}^{(n-1)}_{X,b}$, respectively.

$\widehat{T}^{(n-1)}_{X,a}$ and $\widehat{T}^{(n-1)}_{X,b}$ have a common path
set of length $(n-1)$, if and only if
there is a unique homomorphism $\phi$ from $V_a$
onto $V_b$
such that,
\begin{description}
\item[C1:] 
for each $u\in V_a$, $\phi (u)$ has the same label as $u$,
\item [C2:]
for each $u\in V_a$,  there is an edge-label-preserving bijective mapping
from the set of outgoing edges of $u$
to the set of outgoing edges of $\phi(u)$
such that 
any outgoing edge $(u,u')$ of $u$ is mapped to $(\phi(u),\phi(u'))$.
\end{description}
\end{lemma}
\begin{proof}
$(\Rightarrow)$ 
Let $\phi':V_a\rightarrow V_b$ be the mapping defined algorithmically
as follows.
We first set $\phi'(u_r):=w_r$,
and then define $\phi'$ by 
repeating the next operations for each $j$ from $0$ to $(n-1)-1$:
for every node $u\in V_a$ of depth $(j+\depth (u_r))$
and every edge $(u, u')\in E_a$,
set $\phi' (u'):=w'\in V_b$ 
if $(u,u')$ and $(\phi'(u),w')\in E_b$ have the same label.
Notice that $\phi'(u)$ has been already fixed, 
since the above operations proceed toward leaves in a breadth-first manner.

Under the condition that $\widehat{T}^{(n-1)}_{X,a}$ and
$\widehat{T}^{(n-1)}_{X,b}$ have a common path set of length $(n-1)$,
we will prove that
$\phi' $ is well-defined
(i.e., the operations in the definition work well) and
meets C1 and C2 by induction with respect to depth,
and then we will prove that $\phi'$ is an onto-mapping from $V_a$ to $V_b$
and that $\phi$ is unique (i.e., $\phi$ is equivalent to $\phi'$).

Suppose that $\widehat{T}^{(n-1)}_{X,a}$ and $\widehat{T}^{(n-1)}_{X,b}$
have a common path set of length $(n-1)$.  

The base case is as follows.
Clearly, $\phi'$ is well-defined for $u_r$, and
meets C1 for $u_r$.
Furthermore,
$u_r$ and $w_r$ define the same path set of length
$2(n-1)-\depth(w_r)$,
since two views of infinite depth are isomorphic if and
only if the two views are isomorphic up to depth $(n-1)$.
Thus, $\phi'$ is well-defined for every node 
incident to an outgoing edge of $u_r$ (i.e., every node at depth
$(1+\depth(u_r))$, and meets C2 for $u_r$.

For $j\geq 1$, we assume that for any $u\in V_a$ at depth of at most $(j+\depth (u_r))$,
(1) $\phi'$ is well-defined, (2) $\phi'$ meets C1,
and (3) $u$ and $\phi' (u)$ (denoted by $w$) define the same
path set of length $(2(n-1)-\depth(w))$.
Further, we assume that
(4) $\phi'$ meets C2
for any $u\in V_a$ at depth of at most $((j-1)+\depth (u_r))$. 

For any fixed $u\in V_a$ at depth $(j+\depth (u_r))$
and every outgoing edge $(u, u')\in E_a$ of $u$,
there is a node $w'$ 
such that
$(\phi'(u),w')$ has the same label as
$(u,u')$ and $u'$ has the same label as $w'$,
since 
$u$ and $w(=\phi' (u))$ define the same
path set of length
$(2(n-1)-\depth(w))$ by assumption.
If $\phi'(u')$ is set to $w'$,
$\phi'$ meets C1 for $u'$
and C2 for $u$,
and $u'$ and $w'$ have the same path set of length
$(2(n-1)-\depth(w'))$.
To show that $\phi'$ is well-defined for any node at depth
$((j+1)+\depth (u_r))$,
we have to prove that, 
for any two nodes $u_1$ and $u_2$ at depth $j+\depth (u_r)$,
there are no two edges from $u_1$ and $u_2$, respectively, destined to
some identical node, or, $(u_1,u'),(u_2,u')\in E_a$,
that induce two distinct images of $u'$ by $\phi'$
(while performing the operations given in the definition of $\phi'$).
We assume that such two edges exist and
let $w'_1$ and $w'_2$ be the distinct images of $u'$.
This implies that
path set $P_X^{2(n-1)-\depth(w'_1)}(w'_1)$ is identical to
$P_X^{2(n-1)-\depth(w'_2)}(w'_2)$, since
both of them are identical to
the path set of length $2(n-1)-\depth(w'_1)$ $(=2(n-1)-\depth(w'_2))$
defined for $u'$.
This contradicts Corollary~\ref{cr:minimal-f-view}, since 
$w'_1$ and $w'_2$ are nodes
at the same depth of minimal f-view $\widetilde{T}^{2(n-1)}_X(v)$.
This complete the proof that $\phi'$ is well-defined and meets C1 and C2.

We now prove that $\phi'$ is an onto-mapping.
Assume that $\phi'$ is not an onto-mapping;
there is at least one node in $V_b\setminus \{ \phi' (u)\mid u\in
V_a\}$.
Let $w$ be the node of the smallest depth in
$V_b\setminus \{ \phi' (u)\mid u\in V_a\}$.
Thus, for each edge of $E_b$ into $w$,
its source is in
$\{ \phi' (u)\mid u\in
V_a\}$. 
Since the source satisfies C2,
$w$ needs to be in $\{ \phi' (u)\mid u\in
V_a\}$. This is a contradiction.

Finally, we prove the uniqueness of $\phi$.
We assume that there are two different homomorphisms, $\phi _1$ and
$\phi _2$, satisfying C1 and
C2;
there is at least one node $u$ in $V_a$ such that $\phi _1(u)=w_1\neq
w_2=\phi _2(u)$.
Since any $(n-1)$-length directed path in $\widehat{T}^{(n-1)}_{X,b}$ emanates from
$w_r$,
$\phi_1(u_r)=\phi _2(u_r)=w_r$.
For any directed path from $u_r$ to $u$, 
$\phi_1$ and $\phi_2$ define a path from $w_r$ to $w_1$
and a path from $w_r$ to $w_2$, respectively.
By conditions C1 and C2, these two paths are both isomorphic to the path from
$u_r$ to $u$.
Since there is at most one such path in 
$\widehat{T}^{(n-1)}_{X,b}$ by the definition of f-views,
$w_1$ must be identical to $w_2$.
This is a contradiction.

$(\Leftarrow)$ 
If $\phi$ meets C1 and C2, 
any directed edge $(u,u')$ in $\widehat{T}^{(n-1)}_{X,a}$
is mapped by $\phi$ to
a directed edge $(\phi(u),\phi(u'))$ in 
$\widehat{T}^{(n-1)}_{X,b}$ of the same edge and node labels.
It follows that
any directed path in $\widehat{T}^{(n-1)}_{X,a}$ is mapped to
an isomorphic directed path in 
$\widehat{T}^{(n-1)}_{X,b}$.
Thus, any $(n-1)$-length directed path from $u_r$
in $\widehat{T}^{(n-1)}_{X,a}$ has to be mapped to
some isomorphic $(n-1)$-length directed path in $\widehat{T}^{(n-1)}_{X,b}$.
Therefore, the path set of $\widehat{T}^{(n-1)}_{X,a}$
is a subset of that of $\widehat{T}^{(n-1)}_{X,b}$.

Conversely, fix an $(n-1)$-length directed path $p$ starting at $w_r$.
Let the $j$th node on $p$ be the node on $p$ that can be reached
via $(j-1)$ directed edges from $w_r$.
Since 
any $(n-1)$-length directed path in
$\widehat{T}^{(n-1)}_{X,a}$ is mapped to 
some $(n-1)$-length directed path in $\widehat{T}^{(n-1)}_{X,b}$,
$u_r$ is only the preimage of $w_r$ by $\phi$.
If $u$ is a preimage of the $j$th node on $p$ with respect to $\phi$,
there is only one preimage of the $(j+1)$st node on $p$ 
among nodes incident to the outgoing edges of $u$ due to C2.\@
By induction, the preimage of $p$ is uniquely determined as 
some directed path from $u_r$.
Thus, the path set of $\widehat{T}^{(n-1)}_{X,b}$
is a subset of that of $\widehat{T}^{(n-1)}_{X,a}$.
\end{proof}
Lemma~\ref{lm:test common path set} implies that,
if we can construct $\phi'$ (defined in the proof) that meets C1 and C2,
$\widehat{T}^{(n-1)}_{X,a}$ and $\widehat{T}^{(n-1)}_{X,b}$ have a common
path set of length $(n-1)$.
Conversely, if we cannot construct $\phi'$,
there is no mapping $\phi$ that satisfies C1 and C2;
$\widehat{T}^{(n-1)}_{X,a}$ and $\widehat{T}^{(n-1)}_{X,b}$ do not have a common
path set of length $(n-1)$.
As described in the proof of Lemma~\ref{lm:test common path set},
$\phi'$ can be constructed by
simultaneously traversing 
 $\widehat{T}^{(n-1)}_{X,a}$ and $\widehat{T}^{(n-1)}_{X,b}$ in a
breadth-first manner.
Namely, Subroutine~P first sets $\phi'(u_r):=w_r$
if $u_r$ and $w_r$ have the same label,
and then defines $\phi'$ by 
repeating the next operations for each $j$ from $0$ to $(n-1)-1$.
For every node $u\in V_a$ of depth $(j+\depth (u_r))$,
and every $i$-edge, $(u, u_i)\in E_a$, of $u$,
set $\phi' (u_i):=w_i\in V_b$, where $w_i$ is the destination of
$i$-edge of $\phi'(u)$,
if (1) $d_u=d_{\phi(u)}$,
(2) $u$ and $\phi(u)$ have the same label,
(3) $(u,u_i)$ and $(\phi'(u),w_i)\in E_b$ have the same label,
(4) when $u_i$ has already been visited
and thus $\phi'(u_i)$ has already been defined,
$\phi' (u_i)$ is identical to $w_i$.

Figure~\ref{fig:Subroutine P} gives a precise description of 
Subroutine P, where $\ENQUEUE$, $\DEQUEUE$ and $\CONTINUE$ 
are defined in the same way as in the case of Subroutine~Traversal~(II), 
and are assumed to be implemented in a standard way.
\begin{figure}[t]
\begin{center}
  \hrulefill\\
\textbf{Subroutine P}\\
\begin{description}
\setlength\itemsep{-3pt}
  \item[Input:]
 an integer $n$;\\
 function $\depth$ that gives the depth of nodes in
 $\widetilde{T}^{2(n-1)}_X(v)$, and function $\id$;\\ 
 two sub-f-views, $\widehat{T}^{(n-1)}_{X,a}$ rooted at $u_r$ and
 $\widehat{T}^{(n-1)}_{X,b}$ rooted at $w_r$,
of a minimal f-view $\widetilde{T}^{2(n-1)}_X(v)$ such that $\depth
 (u_r)\leq \depth (w_r)\leq n-1$.
   \item[Output:]
     ``$\mathsf{Yes}$'' or ``$\mathsf{No}$.''
\end{description}
\vspace{-6mm}
  \begin{enumerate}
\setlength\itemsep{-3pt}
\item Perform $\ENQUEUE(Q,u_r)$, where $Q$ is a FIFO queue initialized
  to an empty queue.\\
\item Set $\phi(u_r):=w_r$
\item While $Q$ is not empty, repeat the following steps.
\vspace{-6pt}
  \begin{enumerate}
  \item Set $u:=\DEQUEUE(Q)$.
  \item If $\Label(u)\neq \Label (\phi(u))$, go to step~5.
  \item If $\depth (u)=(n-1)+\depth (u_r)$, $\CONTINUE$.
  \item If $d_{u}\neq d_{\phi(u)}$, go to step~5.
  \item Perform the next steps for $i:=1$ to $d_u$.
    \begin{enumerate}
    \item Set $u_i:=\Adj_i(u)$ and  $w_i:=\Adj_i(\phi(u))$.
     \item If $\Label((u,u_i))\neq \Label((\phi(u),w_i))$, go to step~5. 
    \item If $u_i$ has already been traversed and
    $\id(\phi(u_i))\neq \id(w_i)$, go to step~5.
  \item Set $\phi(u_i):=w_i$.
    \item $\ENQUEUE(Q,u_i)$.
   \end{enumerate}
  \end{enumerate}
\item Halt and output ``$\mathsf{Yes}$.''
\item Halt and output ``$\mathsf{No}$.''
  \end{enumerate}
\vspace{-1\baselineskip}
  \hrulefill
\end{center}
\vspace{-1\baselineskip}
  \caption{Subroutine~P.}
  \label{fig:Subroutine P}
\end{figure}

\begin{lemma}
\label{lm: Subroutine P}
Suppose that
minimal f-view $\widetilde{T}^{2(n-1)}_X(v)$
is a view of
a distributed system of $n$ parties
having $O(\log L)$-bit values.
Given two sub-f-views $\widehat{T}^{(n-1)}_{X,a}$ and
 $\widehat{T}^{(n-1)}_{X,b}$ of depth $(n-1)$
of a minimal f-view $\widetilde{T}^{2(n-1)}_X(v)$,
Subroutine~P outputs ``$\mathsf{Yes}$'' 
if and only if
$\widehat{T}^{(n-1)}_{X,a}$ and $\widehat{T}^{(n-1)}_{X,b}$
have a common path set of length $(n-1)$.
The time complexity is $O(n^2\log (n^DL))$,
where $D$ is the maximum degree over all nodes of the underlying graph of the
distributed system.
\end{lemma}
\begin{proof}
Subroutine~P constructs $\phi:=\phi'$ (defined
in the proof
of Lemma~\ref{lm:test common path set}).
Subroutine P outputs ``$\mathsf{Yes}$''
only when $Q$ is empty, i.e., 
when the subroutine has already visited all nodes
in $\widehat{T}^{(n-1)}_{X,a}$.
It is easy to see that,
when Subroutine P outputs ``$\mathsf{Yes}$,''
$\phi$ meets C1 of Lemma~\ref{lm:test common path set}
(due to step~3.2) and C2 (due to step~3.4 and step~3.5.2).
Thus, 
$\widehat{T}^{(n-1)}_{X,a}$ and $\widehat{T}^{(n-1)}_{X,b}$
have a common path set of length $(n-1)$ by Lemma~\ref{lm:test common path set}.
Conversely, if 
$\widehat{T}^{(n-1)}_{X,a}$ and $\widehat{T}^{(n-1)}_{X,b}$
have a common path set of length $(n-1)$,
the subroutine outputs ``$\mathsf{Yes}$,''
by the \emph{only-if} part in the proof of Lemma~\ref{lm:test common path set}.
This proves the correctness.

Let $V_a$ and $E_a$ be the edge set and node set, respectively,
of $\widehat{T}^{(n-1)}_{X,a}$.
Step~3 is obviously dominant in terms of time complexity.
Step~3.1 takes just constant time for each evaluation.
Step~3.2 takes $O(\log L)$ time for each $u$
since node labels are $O(\log L)$-bit values;
it takes $O(|V_a|\log L)$ time in total.
Step~3.3 takes $O(\log n)$ time for each $u$;
it takes $O(|V_a|\log n)$ time in total.
Step 3.4 takes
$O(d_u)$ time for each $u$; it takes $O(|E_a|)$ time
in total.

Next, we estimate the time complexity of step~3.5.
Steps~3.5.1, 3.5.4, and 3.5.5 take constant time.
Each execution of step~3.5.2 takes $O(\log D)$ time, since edge labels
are $O(\log D)$-bit values.
Each execution of step~3.5.3 takes $O(\log n)$ time, since 
$\widetilde{T}^{2(n-1)}_{X}(v)$ has $O(n^2)$ nodes.
Since every edge is visited exactly once, step~3.5 takes $O(|E_a|(\log
n+\log D))=O(|E_a|\log n)$ time in total.
The time complexity of step~3 is thus
$O(|V_a|\log (nL)+ |E_a|\log n)$; this is
$O(n^2\log(Ln^D))$ since $|V_a|=O(n^2)$ and $|E_a|=O(n^2D)$. 
\end{proof}

\subsubsection{Analysis of View Counting Algorithm}
The correctness and complexity of the view counting algorithm is described in the
 next lemmas.
\begin{lemma}
  Given a minimal f-view $\widetilde{T}^{2(n-1)}_X(v)$, a subset $S$ of the
    range of $X$, and the number, $n$, of parties, the view counting algorithm
    in Figure~\ref{fig: view-counting} correctly outputs $\abs{\Gamma _X^{(n-1)}(S)}$.
\end{lemma}
\begin{proof}
Let $\widetilde{\calC}$ 
be 
the collection of distinct path sets of length $(n-1)$ defined for all
  nodes $u^j$ at depth $j$ in minimal f-view
  $\widetilde{T}^{2(n-1)}_X(v)$ over
all  $j\leq n-1$,
and let
$\calC$ be the counterpart of $\widetilde{\calC}$
for (original) view $T^{2(n-1)}_X(v)$.
Suppose that $\tilde{n}_S$ and $n'_S$ are the numbers of those path
sets in $\widetilde{\calC}$ and $\calC$, respectively,
of which the first node is labeled with some value in $S$.
Since Lemma~\ref{lm: Subroutine P} implies that 
$\tilde{n}_S$ is equal to the number $n_S$ of the nodes in $W$
that are labeled with values in $S$,
the lemma holds if we prove 
$\tilde{n}_S=n'_S$ and $n'_S=|\Gamma _X^{(n-1)}(S)|$.

Recall mapping $\psi$ defined in the proof of Lemma~\ref{lm:f-view}:
for any given f-view, $\psi$ maps every node $u$ of
the f-view to the path set defined for $u$.
As described in the proof of Corollary~\ref{cr:minimal-f-view},
$\psi$ is a bijective mapping from the set of nodes of depth $j$
to $\calP _{G,\sigma,X}^j(v)$, when the corresponding f-view is minimal.
Thus, $\widetilde{\calC}$ is identical to $\calC$,
implying $\tilde{n}_S=n'_S$.

The fact that $n'_S=|\Gamma _X^{(n-1)}(S)|$
is obtained from the following two properties:
(1) two path sets are identical to each other if and only
if the corresponding two views are isomorphic to each other;
(2) the first node of the path set 
has the same label as the root of
the corresponding view.
\end{proof}

\begin{lemma}
\label{lm: operations on cv}
For a given distributed system of $n$ parties, each of which has a
value of $O(\log L)$ bits,
the view counting algorithm in Figure~\ref{fig: view-counting} can compute
$|\Gamma _{X}^{(n-1)}(S)|$ for any subset $S$
of the range of $X$ from $\widetilde{T}^{2(n-1)}_{X}(v)$
  in $O(n^5\log (n^DL))$ time, where $D$ is the maximum
  degree over all nodes of the underlying graph.
\end{lemma}
\begin{proof}
We first consider Subroutine Traversal~(II) in Figure~\ref{fig:traversal II},
which is called in the first step of the view counting algorithm.
The dominant part of Subroutine Traversal~(II) 
is step~4, 
which just performs a simple breadth-first traversal of
$\widetilde{T}^{2(n-1)}_{X}(v)$.
The traversal takes $O(\log\abs{V^f})$ time for each edge.
Thus, the time complexity of the subroutine, i.e.,
the time complexity of step~1 in the view counting algorithm,
is $O(n^2D\log\abs{V^f})$.

Next we analyze step~3 of the view counting algorithm in Figure~\ref{fig: view-counting},
which is clearly dominant in terms of time complexity.
We can see that
(1) $|W|$ is at most $n$ since there are $n$ parties
in the system, and (2) there are $O(n^2)$ nodes whose depth 
is at most $n-1$ in $\widetilde{T}^{2(n-1)}_{X}(v)$ since there are at
most $n$ nodes
at each depth.
Hence, Subroutine P is called
for each of $O(n^3)$ pairs of sub-f-views.
Since one call of Subroutine P takes 
$O(n^2\log (n^DL))$ time by Lemma~\ref{lm: Subroutine P},
step 3.2 thus takes $O(n^5\log (n^DL))$ time;
the other operations can be perfomed with the same order of 
the time complexity.
The total time complexity is thus $O(n^5\log (n^DL))$.
\end{proof}
\subsection{Directed network topologies}
As in the case of an undirected network,
Norris's theorem can be proved to be still valid in the case of a
directed network in almost the same way as the original proof.
\begin{theorem}[Norris~\cite{norris95}]
\label{th:norris:directed-case}
Suppose that there is any $n$-party distributed system 
whose underlying graph $G$ is directed and
strongly connected.
For any nodes, $v$ and $v'$, of $G$,
$T_{G,\sigma, X}(v)\equiv T_{G,\sigma, X}(v')$ if and only if
$T^{n-1}_{G,\sigma, X}(v)\equiv T^{n-1}_{G,\sigma, X}(v')$.
\end{theorem}
(The proof is given in the appendix.)

Therefore, it is sufficient to count the number of non-isomorphic
views
of depth $n-1$ in order to count the number of non-isomorphic views
of infinite depth; 
a natural idea is that, as in the case of undirected network topologies,
every party constructs an f-view of depth $2(n-1)$
and then counts the number of non-isomorphic views of depth $n-1$.
This idea can work well for the next reason.

An f-view is obtained
by just sharing isomorphic subgraphs of a view; 
the f-view construction algorithm in Figure~\ref{fig:f-view construction}
does not care whether the view is derived from 
a directed network or an undirected network.
Thus, the f-view construction algorithm works well
for networks whose underlying graphs are directed and strongly
connected.
It is obvious that the complexity in the directed network case
is the same order as in the undirected network case,
since it depends only on the number of nodes and edges of the
underlying graph, and the number of bits used
to represent node and edges labels.

As for the view-counting algorithm in Figure~\ref{fig: view-counting},
it is easy to see that Lemma~\ref{lm:test common path set} 
does not depend on the fact that the underlying topology
is undirected except for Norris's theorem.
Since, as stated above, Norris's theorem is still valid in the 
case of directed networks,
the view counting algorithm works well.

\section{Conclusion}
\label{sec: conclusion}
It is well-known that $\LE_n$ in an anonymous network cannot be solved
classically in a
deterministic sense for a certain broad class of network topologies such
as regular graphs, even if all parties know the exact number of
parties. This paper proposed two quantum algorithms that exactly solve $\LE_n$
for any topology of anonymous networks when each party initially knows
the number of parties, but does not know the network topology. 
The two algorithms have their own characteristics.

The first algorithm is simpler and more efficient in 
time and communication complexity than the second one:
it has $O(n^3)$ time complexity for each party
and $O(n^4)$ communication complexity.

The second algorithm is more general than the first
one, since it can work even on networks whose underlying graph
is directed.
Moreover, the second algorithm is better than the first one
in terms of some complexity measures.
It has the
total communication complexity of $O(n^6(\log n)^2)$,
but involves the quantum communication of just $O(n^2 \log n)$ qubits of
one round, while
the first algorithm requires quantum communication of $O(n^4)$
qubits. It runs in $O(n\log n)$ rounds,
while the first one 
runs in $O(n^2)$ rounds.

As for local computation time, the second algorithm requires 
$O(n^6(\log n)^2)$ time for each party. 
To attain this level of time and communication complexity,
we introduced folded view, a view compression technique that
enables views to be constructed in  polynomial time and
communication. 
The technique can be used to deterministically check if the unique leader is selected
or not in polynomial time and communication and linear rounds in the
number of parties.
Furthermore, the technique can also be used to compute any symmetric
Boolean function, i.e., any Boolean function that depends only on 
the Hamming weight of input in $\{ 0,1\}^n$, on anonymous networks,
when every party is given one of the $n$ bits of the function's input.

Our leader election algorithms can exactly solve the problem even when
each party initially knows only the upper bound of the number of
parties, whereas, in this setting for any topology with cycles, 
it was proved that no zero-error probabilistic algorithms exist.

Our algorithms use unitary gates depending on the number of
parties that are eligible to be a leader during their execution.
Thus, the algorithms require a set of elementary unitary gates whose
cardinality is linear in the number, $n$, of parties.
From a practical point of view, however, it would be desirable to perform
leader election for any $n$
by using a fixed and constant-sized set of elementary unitary gates.
It is open as to whether the leader election problem can, in an anonymous
network,
be exactly solved in the quantum setting
by using that set of gates.

It would also be interesting to improve the upper bound
and find a lower bound, of the complexity of solving the problem.
In general, however, it is difficult to optimize both communication complexity
and round complexity (i.e., the number of rounds required).
A reasonable direction is to clarify
the tradeoff between them.
As for communication complexity,
quantum communication cost per qubit would be quite different from classical
communication cost per bit. Hence, it is also a natural open question
as to what tradeoff between quantum and classical communication
complexity exists, and how many qubits need to be communicated.

It is also open whether the problem can be solved
by a processor terminating algorithm (i.e., an algorithm
that terminates when every party enters a halting state)
in the quantum setting
even without knowing the upper bound of the number of parties.
In this situation, there are just message terminating
algorithms with bounded error in the classical setting.

\newpage
\section*{Appendix: Proof of Theorem~\ref{th:norris:directed-case}}
Suppose that the underlying directed graph of the distributed system is $G=(V,E)$.
Let $\pi _k$ be the partition induced on $V$ by the isomorphism
of views $T^{k}_{X}(v)$ of depth $k$ for $v\in V$.
Obviously, $\pi _{k+1}$ is a refinement of $\pi_k$, i.e., 
if $v$ and $w$ are in distinct blocks of $\pi _k$, then they are in
distinct blocks of $\pi _{k+1}$.\\

\begin{proof}\emph{ of Theorem~\ref{th:norris:directed-case}}
By Lemma~\ref{appdx:lm:pi}, which will be stated later, there is some $k>0$ such that
$|\pi_j|< |\pi_{j+1}|$ for every $j<k$, and
$|\pi_j|=|\pi_{j+1}|$ for every $j\geq k$, 
where $|\pi_j|$ is the number of blocks in $\pi_j$.
Thus, $n\geq |\pi _k|\geq k+ |\pi _0|\geq k+1$, 
implying $k\leq n-1$.
Therefore, if $T^{n-1}_{X}(v)\equiv T^{n-1}_{X}(w)$,
then $T_{X}(v)\equiv T_{X}(w)$.

Conversely, if $T^{n-1}_{X}(v)\not\equiv T^{n-1}_{X}(w)$,
it is obvious that $T_{X}(v)\not\equiv T_{X}(w)$.
\end{proof}
\begin{lemma}
\label{appdx:lm:pi}
  If $\pi_{k-1}=\pi _k$ for some $k>0$, then $\pi_j=\pi_k$ for all
  $j\geq k$.
\end{lemma}
\begin{proof}
  Assume that $\pi_k\neq\pi_{k+1}$.
Then, there are a pair of nodes $v$ and $w$ such that 
$T^{k}_{X}(v)\equiv T^{k}_{X}(w)$ and $T^{k+1}_{X}(v)\not\equiv
T^{k+1}_{X}(w)$.
We will prove the next claim.
\begin{claim}
\label{appdx:claim:Tk}
  Suppose that $k\geq 1$ is such that 
$T^{k}_{X}(v)\equiv T^{k}_{X}(w)$ but
$T^{k+1}_{X}(v)\not\equiv T^{k+1}_{X}(w)$ for some nodes $v$ and $w$ in $V$.
Then there are
children $s$ and $t$ of the roots
of $T^{k}_{X}(v)$ and $T^{k}_{X}(w)$, respectively, such that
$T^{k-1}_{X}(\check{s})\equiv T^{k-1}_{X}(\check{t})$ but
$T^{k}_{X}(\check{s})\not\equiv T^{k}_{X}(\check{t})$,
where $\check{s}$ and $\check{t}$ are the nodes of the underlying
directed graph,
corresponding to $s$ and $t$, respectively.
\end{claim}
% there exist children $r$ and $s$ of the roots
% of $T^{k}_{X}(v)$ and $T^{k}_{X}(w)$, respectively, such that
% $T^{k-1}_{X}(v_r)\equiv T^{k-1}_{X}(w_s)$ and 
% $T^{k}_{X}(v_r)\not\equiv T^{k}_{X}(w_s)$,
% where $v_r$ and $w_s$ are the nodes of the underlying graph,
% corresponding to $r$ and $s$.
By this claim,  $\pi_k\neq\pi_{k+1}$ implies that $\pi_{k-1}\neq\pi
_k$, which is a contradiction.
Thus, if $\pi_{k-1}=\pi _k$ for some $k>0$, then $\pi_k=\pi_{k+1}$.
By induction, the lemma holds.
\end{proof}
\begin{proof}\emph{of Claim~\ref{appdx:claim:Tk}}
Let $\beta$ be any isomorphism from $T^k_X(v)$ to $T^k_X(w)$.
Let $\hat{v}$ and $\hat{w}$ be the roots of $T^k_X(v)$ and $T^k_X(w)$, respectively.
Then, for any child $s$ of $\hat{v}$, $T^{k-1}_X(\check{s})\equiv
T^{k-1}_X(\check{\beta(s)})$, where
$\check{s}$ and $\check{\beta(s)}$ are the nodes in $V$ corresponding to $s$ and $\beta(s)$.

We assume that, for every child $s$ of $\hat{v}$,
there exists
isomorphism $\beta_s$ from $T^{k}_X(\check{s})$ to 
$T^{k}_X(\check{\beta(s)})$.
We define a new map $\beta'$ such that
$\beta '(\hat{v}) =\beta (\hat{v})=\hat{w}$ and
$\beta '(u)=\beta_s (u)$ for every node $u$ in $T^{k}_X(\check{s})$ for each
child $s$ of $\hat{v}$.
It is easy to see that $\beta _s(s)=\beta (s)$ for each $s$,
since only $s$ and $\beta (s)$ are the sources of
the $k$-length directed path
in $T^{k}_X(\check{s})$ and $T^{k}_X(\check{\beta(s)})$, respectively.
Thus, $\beta '$ is isomorphism from $T^{k+1}_X(v)$ to $T^{k+1}_X(w)$.
This is a contradiction. Thus, the claim holds.
\end{proof}

%%% Local Variables:
%%% mode: latex
%%% TeX-master: "LEjournal"
%%% End:

\end{document}